\documentclass[a4paper,reqno,12pt]{amsart}

\usepackage{amsmath, amsfonts, amssymb, amsthm, amscd}
\usepackage{perpage}
\usepackage{url}
\usepackage{color}
\usepackage{framed} 
\usepackage{mathrsfs}
\usepackage{tikz}
\usepackage{mathabx}
\usetikzlibrary{arrows,decorations.pathmorphing,backgrounds,positioning,fit,petri} 
\usepackage{tikz,tikz-cd}\usepackage{caption,subcaption}
\usepackage{wrapfig,stackengine}
\usepackage{esint}
\usepackage{dsfont} 
\usepackage{textcomp}

\usepackage[utf8]{inputenc}
\usepackage[T1]{fontenc}
\usepackage{microtype}

\usepackage[a4paper,scale={0.75,0.8},marginratio={1:1},footskip=7mm,headsep=10mm]{geometry}

\usepackage{hyperref}
\usepackage{cleveref}
\usepackage{xfrac}
\usepackage{stackengine}

\makeatletter
\def\@secnumfont{\bfseries\scshape}

\def\section{\@startsection{section}{1}%
  \z@{.8\linespacing\@plus\linespacing}{.5\linespacing}%
  {\normalfont\large\bfseries\scshape\centering}}

\def\subsection{\@startsection{subsection}{2}%
  \z@{.5\linespacing\@plus.7\linespacing}{-.5em}%
  {\normalfont\bfseries\scshape}}

\def\subsubsection{\@startsection{subsubsection}{3}%
  \z@{.5\linespacing\@plus.7\linespacing}{-.5em}%
  {\normalfont\scshape}}

\def\specialsection{\@startsection{section}{1}%
  \z@{\linespacing\@plus\linespacing}{.5\linespacing}%
  {\normalfont\centering\large\bfseries\scshape}}
\makeatother

\numberwithin{equation}{section}

\definecolor{shadecolor}{gray}{.94}

\makeatletter
\newenvironment{myshade}{%
  \topsep4\p@\@plus4\p@\relax%
  \MakeFramed{\advance\hsize-\width \FrameRestore}}%
 {\par\unskip\endMakeFramed}%
\makeatother

\definecolor{shadecolor2}{rgb}{0.94, 0.9, 0.55}

\makeatletter
\newenvironment{myshade2}{%
  \topsep4\p@\@plus4\p@\relax%
  \MakeFramed{\advance\hsize-\width \FrameRestore}}%
 {\par\unskip\endMakeFramed}%
\makeatother

\definecolor{shadecolor3}{rgb}{1.0, 0.63, 0.48}

\makeatletter
\newenvironment{myshade3}{%
  \topsep4\p@\@plus4\p@\relax%
  \MakeFramed{\advance\hsize-\width \FrameRestore}}%
 {\par\unskip\endMakeFramed}%
\makeatother

\newtheoremstyle{mytheorem}{0}{0}%
     {\itshape}
     {}
     {\bfseries}
     {. }
     {0.3ex}
     {\thmname{{\bfseries #1}}\thmnumber{ {\bfseries #2}}\thmnote{ (#3)}}  

\theoremstyle{mytheorem}

\newtheorem{theo}{Theorem}[section]
\newenvironment{theorem}{\begin{myshade}\begin{theo}}{\end{theo}\end{myshade}}

\newtheorem{prop}[theo]{Proposition}
\newenvironment{proposition}{\begin{myshade}\begin{prop}}{\end{prop}\end{myshade}}

\newtheorem{lem}[theo]{Lemma}
\newenvironment{lemma}{\begin{myshade}\begin{lem}}{\end{lem}\end{myshade}}

\newtheorem{quest}[theo]{Question}
\newenvironment{question}{\begin{myshade2}\begin{quest}}{\end{quest}\end{myshade2}}

\newtheorem{defin}[theo]{Definition}

\newtheorem{cor}[theo]{Corollary}
\newenvironment{corollary}{\begin{myshade}\begin{cor}}{\end{cor}\end{myshade}}

\newtheorem{conj}[theo]{Conjecture}
\newenvironment{conjecture}{\begin{myshade3}\begin{conj}}{\end{conj}\end{myshade3}}

\newtheoremstyle{mydefinition}{.7\linespacing\@plus.3\linespacing}{.7\linespacing\@plus.3\linespacing}%
     {\rmfamily}
     {}
     {\bfseries}
     {. }
     {0.3ex}
     {\thmname{{\bfseries #1}}\thmnumber{ {\bfseries #2}}\thmnote{ (#3)}}  

\theoremstyle{mydefinition}

\newtheorem{remark}[theo]{Remark}

\newenvironment{myenumerate}{%
\renewcommand{\theenumi}{\arabic{enumi}}%
\renewcommand{\labelenumi}{{\rm(\theenumi)}}%
\begin{list}{\labelenumi}
	{%
	\setlength{\itemsep}{0.4em}%
	\setlength{\topsep}{0.5em}%
	\setlength\leftmargin{2.45em}%
	\setlength\labelwidth{2.05em}%
	\setlength{\labelsep}{0.4em}%
	\usecounter{enumi}%
	}%
	}%
{\end{list}
}

{\end{list}
}

{\end{list}
}

\renewenvironment{enumerate}{
\begin{myenumerate}}%
{\end{myenumerate}}

\newenvironment{myitemize}{%
\begin{list}{$\bullet$}%
 	{%
	\setlength{\itemsep}{0.4em}%
	\setlength{\topsep}{0.5em}%
	\setlength\leftmargin{2.65em}%
	\setlength\labelwidth{2.65em}%
	\setlength{\labelsep}{0.4em}%
	}%
	}%
{\end{list}}

\renewenvironment{itemize}{
\begin{myitemize}}%
{\end{myitemize}}

\MakePerPage[2]{footnote} 

\newcommand{\R}{\mathbb{R}}

\renewcommand{\epsilon}{\varepsilon}

\newcommand{\vertiii}[1]{{\left\vert\kern-0.25ex\left\vert\kern-0.25ex\left\vert #1 
    \right\vert\kern-0.25ex\right\vert\kern-0.25ex\right\vert}}

\def\and{\ and }

\newcommand{\re}{\mathsf{Re}}
\newcommand{\be}{\begin{equation}}
\newcommand{\ee}{\end{equation}}
\newcommand{\EE}{\mathbb{E}}
\newcommand{\bea}{\begin{eqnarray}}
\newcommand{\eea}{\end{eqnarray}}

\renewcommand{\div}{{\mbox{div}\,}}

\newcommand{\ve}{{\varepsilon}}

\newcommand{\rmd}{{\rm d}}
\newcommand{\bx}{{ {x} }}

\newcommand{\br}{{{r}}}
\newcommand{\bv}{{ {v}}}

\newcommand{\bX}{{X}}

\newcommand{\bu}{{{u}}}

\def\wc{\rightharpoonup}
\def\dist{\text{dist}}

\def\div{\text{div}}

\newcommand{\ol}{\overline}

\def\Xint#1{\mathchoice
{\XXint\displaystyle\textstyle{#1}}%
{\XXint\textstyle\scriptstyle{#1}}%
{\XXint\scriptstyle\scriptscriptstyle{#1}}%
{\XXint\scriptscriptstyle\scriptscriptstyle{#1}}%
\!\int}
\def\XXint#1#2#3{{\setbox0=\hbox{$#1{#2#3}{\int}$ }
\vcenter{\hbox{$#2#3$ }}\kern-.6\wd0}}

\def\dashint{\Xint-}
\usepackage{trimclip}

\makeatletter
\newcommand{\toitself}{\mathrel{\mathpalette\toitself@\relax}}

\newcommand{\toitself@}[2]{%
  \vbox{%
    \sbox\z@{$#1\supset$}
    \hbox{\clipbox{0 {-0.1\height} {0.65\width} {-0.4\height}}{$\m@th#1\leftarrow$}}
    \nointerlineskip
    \kern-0.57\ht\z@
    \ifx#1\scriptstyle\kern0.03\ht\z@\fi
    \ifx#1\scriptscriptstyle\kern0.04\ht\z@\fi
    \hbox{$\m@th#1\supset$}
  }%
}
\makeatother

\title[Mathematical Theorems on Turbulence]{Mathematical Theorems on Turbulence}

\author[Theodore D. Drivas]{Theodore D. Drivas} \address{Department of Mathematics,  Stony Brook University, Stony Brook, NY, 11794} \email{tdrivas@math.stonybrook.edu}

\begin{document}

\begin{abstract}
In these notes, we emphasize Theorems rather than Theories concerning turbulent fluid motion.  Such theorems can be viewed as constraints on the theoretical predictions and expectations of some of the greatest scientific minds of the 20th century: Lars Onsager, Andrey Kolmogorov,  Lev Landau, Lewis Fry Richardson among others.
\end{abstract}

\maketitle

\begin{quotation}
\emph{``Kolmogorov considered his work on turbulence to be non-mathematical. He
 wanted to explain observed phenomena from first principles.''} \\
 \phantom{asdf} \hfill  --  V. I. Arnol'd \cite{A91}
\end{quotation}

\tableofcontents

In this Chapter, we collect mathematical results which shed light on turbulent fluid motion. Our focus is on theorems that connect seemingly distinct features to one another, such as properties of energy dissipation to precise regularity properties of the flow field.  Most of these results appear in the literature in some form, while others are small novelties.  A brief outline is as follows: in \S \ref{sec:e}, we introduce the incompressible Euler equations and discuss well-posedness theory and conservation laws.    In \S \ref{nssec} we introduce Navier-Stokes, Reynolds number and  anomalous dissipation (without the eventual aid of viscosity) in turbulence.  In \S \ref{onssec} we review Onsager's ideas which relate anomalous dissipation to $\sfrac{1}{3}$ differentiability of the flow field.  In \S
\ref{K41sec}, we discuss Kolmogorov's celebrated 1941 theory, as well as rigorous versions of his $\sfrac{4}{3}$ and $\sfrac{4}{5}$ laws.   In \S \ref{landausec} we describe Landau's objection to Kolmogorov's theory based upon intermittency of the dissipation field, along with rigorous constraints.  In \S \ref{models}, model problems (the Burgers equation and passive scalar transport) are studied in some depth, illustrating some of the features noted above.  Finally \S \ref{particles} collects results pertaining to motion of particles in turbulence. Some open questions are posed.

\section{The Euler equations of ideal fluid motion}\label{sec:e}
\vspace{2mm}
\begin{quotation}
\emph{“When a continuous medium is deprived of its physical properties (elasticity, thermal and electrical conductivity, and so on) its property of occupying a definite position in space remains, as do elementary interactions through the mutual pressure of its parts, due to Aristotle’s principle that it is impossible for two bodies to occupy the same space. It is amazing that it is these elementary interactions that cause the most complicated effects, including turbulence\dots”} \hfill – V. I. Yudovich \cite{Y06}
\end{quotation}
\vspace{2mm}

A vector field  $u:\mathbb{R}\times M\to \mathbb{R}^d$ on a domain $M\subset \mathbb{R}^d$ is called the velocity field  of an ideal incompressible fluid if it satisfies the Euler equations 
\begin{align}\label{ee1}
\partial_t u + u\cdot\nabla u&= -\nabla p,\\\label{ee2}
\nabla \cdot u &=0,\\
u\cdot n|_{\partial M} &=0. \label{ee3}
\end{align}
The unknown function $p$, named the pressure,  is a Lagrange multiplier, which enforces the condition  that $u$ remain divergence-free for all times (cf. \eqref{pressureeqn}--\eqref{pressurebc}).

As written,  \eqref{ee1}--\eqref{ee3}  require differentiable velocity fields.  However, the notion of solution to the equations of motion can be generalized in more irregular regimes such as  \emph{turbulence}, which is a singular regime of rough and tumble flow.  As a general rule, we will not attempt to give a precise definition of turbulence, but instead will enumerate a number of features that are associated to this phenomenon. 
For reasons that will become clear further on, we will refer to the situation where the velocity field is differentiable as a \emph{non-turbulent regime}. It is a regime in which the Euler equations may give a definite prediction for the motion of the fluid -- a notion that three-dimensional turbulence seems to defy.  One such prediction concerns conserved quantities.  Being variational in nature (arising due to Hamilton's principle), this regime is subject to the conservation laws implied by symmetries of its action via Noether's theorem \cite{Salmon}.  Specifically:
\begin{itemize}
\item \emph{time translation invariance} implies energy conservation\footnote{We remark that the Euler equations have a beautiful geometric interpretation; the continuum flow of particles they generate follows a geodesic on the manifold of volume preserving diffeomorphisms, with respect to the $L^2$ (kinetic energy) metric \cite{A99}. Conservation of the energy is a feature of all geodesic motion -- geodesics move with constant speed.  In fact, there is a more general families of quantities conserved along geodesic -- angles made between their tangent $u$ and   parallel transported vectors $v$:
\be
\tfrac{\rmd}{\rmd t} \langle u(t), v(t)\rangle_{L^2(M)} = 0 \qquad \text{where} \qquad \partial_t v + \mathbb{P}\big[ u\cdot \nabla v\big] = 0, \qquad v|_{t=0} = v_0,
\ee
where $\mathbb{P}$ is the $L^2$ orthogonal projection onto divergence-free velocity fields.}:
\be
\tfrac{\rmd}{\rmd t} \tfrac{1}{2} \|u(t)\|_{L^2(M)}^2 = 0.
\ee
\item \emph{particle relabelling symmetry} implies conservation of circulation (Kelvin theorem):
\be\label{circulation}
\tfrac{\rmd}{\rmd t} \langle u(t), v(t)\rangle_{L^2(M)} =0 \qquad \text{where} \qquad \partial_t v + [u,v] = 0, \qquad v|_{t=0} = v_0.
\ee
\end{itemize}
where $ [u,v]= u\cdot \nabla v- v\cdot \nabla u$ is the Lie bracket and $v_0$ is an arbitrary smooth divergence-free vector field, tangent to the boundary $\partial M$.
Space translation/rotation invariance gives rise to conservation of linear and angular momenta, provided that the fluid vessel is invariant under those actions.  Momentum conservation  thus holds for a broad family of conservation laws, both finite and infinite dimensional. 
On the other hand, circulation conservation is a truly infinite dimensional phenomenon --  the number of conservation laws is in correspondence with the infinitesimal generators of volume preserving diffeomorphisms.  As noted by Kelvin, they, in fact,  \emph{define} ideal fluid motion \cite{Kelvin,DH20}. The circulation theorem \eqref{circulation} is often presented in a slightly different, derived form.  To recover this, let $\Gamma$ be a rectifiable loop, and  $v_0(x) \rightharpoonup  \oint_\Gamma \delta(x-\cdot) \rmd \ell$, a singular loop. We see that $v_0$ is distributionally divergence-free $\langle v_0, \nabla \varphi\rangle_{L^2(M)} = \oint_\Gamma \nabla \varphi  \cdot \rmd \ell=0$ and that the solution to the Lie transport is $v(t,x) =  \oint_{X_t(\Gamma)} \delta(x-\cdot) \rmd \ell$, where $X_t$ is the flow generated by $u$, \eqref{eel1}--\eqref{eel2}.  As such, the circulation theorem \eqref{circulation} takes the form that justifies its name:
\be
\frac{\rmd}{\rmd t} \oint_{X_t(\Gamma)} u(t, \cdot)\cdot  \rmd \ell = 0 \qquad \text{for any rectifiable loop} \ \Gamma \subset M.
\ee
By Stokes' theorem, this line integral is equivalent to the flux of the curl of the velocity vector field $u$, the \emph{vorticity} $\omega$, through any  comoving bounding surface.  Letting the loop become infinitesimal about any given point $x\in M$, with arbitrary orientation, we recover the Lie advection equation for the vorticity.
In two dimensions one can identify the vorticity with a scalar field $\omega= \nabla^\perp \cdot u$ where $\nabla^\perp = (-\partial_2,\partial_1)$ and in three dimensions with a vector field $\omega= \nabla \times u$.  These are transported by $u$:
  \begin{alignat}{2}\label{2dvort}
d=2: \qquad \partial_t\omega + u\cdot \nabla \omega &=0, \\ \label{3dvort}
d=3: \ \  \qquad  \partial_t\omega +  [u,\omega] &= 0.
\end{alignat}
Notice that, in three dimensions, the circulation theorem in its primal form \eqref{circulation} contains conservation of helicity $H[u] = \langle u, \omega\rangle_{L^2(M)} $, at least provided $\omega_0\cdot \hat{n}|_{\partial M}=0$ (a Lie-propagated condition, since  $u\cdot \hat{n}|_{\partial M}=0$) if boundary is present.

For all of the above deductions to hold, we require the solution to be sufficiently regular (classical). Let us understand when these may exist. The Cauchy problem for \eqref{ee1}--\eqref{ee3}  consists of prescribed data at some moment $u|_{t=0}=u_0\in X$ as an element of an appropriate functional space $X$.  To understand what the space $X$ should satisfy in order for \eqref{ee1}--\eqref{ee3} to be a predictive dynamical system, we formally differentiate  \eqref{ee1}:
\begin{align}
(\partial_t + u \cdot \nabla) \nabla u +(\nabla u)^2 &= -\nabla^2 p,\\
-\Delta p &= {\rm tr} (\nabla u)^2, \label{pressureeqn}\\
\partial_n p|_{\partial M} &= u\cdot \nabla  \hat{n} \cdot u|_{\partial M}, \label{pressurebc}
\end{align}
and observe that, to have a good local theory in the space $X$, it should be the case that 
\be\label{estimate}
\frac{\rmd}{\rmd t} \| \nabla u \|_{Y} \lesssim \| \nabla u \|_{Y} ^2,
\ee
where the space $Y$ is such that if $u\in X$ then $\nabla u\in Y$.
For \eqref{estimate} to hold, $Y$ should encode enough regularity, namely $Y$ should be
\begin{enumerate}
\item  an algebra, e.g.  if $f, g\in Y$ then $fg\in Y$,
\item  compatible with the Neumann problem for the pressure \eqref{pressureeqn}--\eqref{pressurebc}, e.g. 
\be
\|\nabla^2 p\|_Y \lesssim \| \nabla u \|_{Y} ^2,
\ee
\item such that the transport equation is well-posed in $X$ for $\nabla u\in Y$.
\end{enumerate}
For (1), H\"{o}lder spaces $X= C^{k+1,\alpha}$, $Y=C^{k,\alpha}$ with $k\geq 0$, $\alpha\in (0,1)$ or Sobolev spaces $X=H^{s+1}$, $Y= H^s$ with $s>d/2$ work. For (2), one should avoid endpoints spaces like $X= C^k$ or $W^{k,p}$ with $p=1$ and $\infty$. As for (3), the transport equation is generally well-posed at this level of regularity, which encodes stronger-than-Lipschitz information on $u$.  With such a space, one can obtain the \emph{a priori} estimate
\be
\|\nabla u(t)\|_Y \leq e^{C\int_0^t \|\nabla u(s)\|_{L^\infty} \rmd s }\|\nabla u_0\|_Y.
\ee
If $T_*>0$ is the maximal existence time, either $T_*=\infty$ or $T_*<\infty$ and $\int_0^{T_*}  \|\nabla u(s)\|_{L^\infty} \rmd s=\infty$. The classical result, due to  Lichtenstein and Gunther near the turn of the century, is
\begin{theorem}[Local existence in H\"{o}lder spaces \cite{L25,G27}]\label{lethm}
Suppose $u_0\in C^{1,\alpha}(M)$. Then there is a time $T>0$ such that there is a unique solution $u\in C^{1,\alpha}((-T,T)\times M)$.
\end{theorem}
Such solutions are called classical, since each term in the equation makes sense pointwise on their domain of definition.  The key to our understanding of the Cauchy problem for the Euler equations is the fact that vorticity (the antisymmetric part of the velocity gradient), is transported geometrically as a two-form.  In two-dimensions, the vorticity can be identified with a scalar field while in three dimensions, it is a vector field.  As a consequence,  it is long known that in two space dimensions the equations are globally predictive, namely $T_*=\infty$ \cite{H33}.  In three dimensions, it remained an open question until, nearly a century later, Elgindi and Elgindi-Ghoul-Masmoudi showed $T$ may be finite:

\begin{theorem}[Finite time singularity in H\"{o}lder spaces \cite{E21,EGM19}]\label{tarekthm}
There exists an $\alpha\in(0,1)$,  $u_0\in C^{1,\alpha}\cap L^2(\mathbb{R}^3)$ and  $T_*>0$ such that $u(T_*)$ ceases to be in $C^{1}(\mathbb{R}^3)$.
\end{theorem}

In fact, a flurry of recent progress indicates that even infinitely smooth initial conditions become singular, at least in the case where there is a solid boundary \cite{CH22}, see discussion in \cite{DE23}.
Thus, turbulence --  which we will argue  lies outside the class of functions which result in a local theory --  cannot be avoided with certainty, even if one starts from the non-turbulent regime.  Of  arguably  more relevance is the fact that turbulence in real world scenarios should not be modeled as arising from regular initial conditions -- rather it often results from initial conditions which are already much too rough for a predictive Cauchy problem to exist. We shall discuss this point in the subsequent section.

\section{Navier-Stokes and the inviscid limit}\label{nssec}
\vspace{2mm}

\begin{quotation}
\emph{“\dots It is amazing that it is these elementary interactions that cause the most complicated effects, including turbulence (\underline{viscosity} of course plays an essential role in generating it)."} \hfill – V. I. Yudovich \cite{Y06}
\end{quotation}
\vspace{2mm}

Real fluids that are encountered in nature are not ideal.  R.P. Feynman referred to ideal fluid motion as the \emph{flow of dry water}, highlighting its somewhat absurd character.  Real fluids, made of an extremely large but finite collection of molecules,  have intrinsic stickiness arising from friction in the mutual interaction of their parts, making them wet. The material-specific, dimensional parameter that measures the strength of said effect is called \emph{viscosity}.  The simplest model for wet water is the \emph{Navier-Stokes equations}: 
\begin{align}\label{ns1}
\partial_t u^\nu + u^\nu\cdot\nabla u^\nu&= -\nabla p^\nu+\nu \Delta u^\nu  + f,\\\label{ns2}
\nabla \cdot u^\nu &=0,
\end{align}
where $\nu>0$ represents the kinematic viscosity of the medium, and $f$ an external force. Dry water corresponds to a material with $\nu=0$ exactly, but no such exists. However, it turns out the many flows that we encounter in Nature are \emph{almost} dry!  To understand this point, note that the force of this new viscous effect depends on the state of motion.  Indeed, if -- for example -- we non-dimensionalize the system using a characteristic velocity $\mathsf{U}=\|u_0\|_{L^2}$, force $\mathsf{F} := \|f\|_{L^2}$, length $\mathsf{L}={\rm diam}(M)$ and time $\mathsf{T}=\mathsf{L}/\mathsf{U}$, then the resulting non-dimensional equations now feature the \emph{Reynolds number} $\re= \mathsf{U}\mathsf{L}/\nu$ and (modified) \emph{Grashof number}  $\mathsf{Gr}  :=\mathsf{F}/(\mathsf{U}^2/\mathsf{L})$\footnote{{Grashof number} conventionally means $\widetilde{\mathsf{Gr}}  :=\mathsf{F}L^3/\nu^2$, which mixes small scale with large scale quantities \cite{CF88}.  It is easy to see that $\widetilde{\mathsf{Gr}}   =  \mathsf{Gr}  \re^2$.  The reason behind this definition is that it involves purely large scale quantities, and (assuming Taylor's hypothesis $\langle \ve \rangle =\mathsf{U}^3/L$), where $\langle \ve\rangle $ is the mean dissipation \eqref{locdissform}, it equates to $\widetilde{\mathsf{Gr}}  =\mathsf{F}\mathsf{U}/\langle \ve \rangle$.  Thus, we see that $\mathsf{Gr}$ measures the "inefficiency” of the forcing function, becoming very large  when $\langle \ve \rangle\ll \mathsf{F}\mathsf{U}$, so that force injects energy inefficiently compared with a naive
dimensional estimate. See \cite{EP25,Lopez25}. }, two dimensionless parameters of the system
\begin{align}\label{ns1}
\partial_t u^\re + u^\re\cdot\nabla u^\re&= -\nabla p^\re+\tfrac{1}{\re} \Delta u^\re + \mathsf{Gr} f  ,\\\label{ns2}
\nabla \cdot u^\re &=0.
\end{align}
The principle of hydrodynamic similarity is that flows with the same Reynolds and Grashof number behave the same way. These numbers may be varied independently, although we will consider only $\mathsf{Gr}\equiv 1$ from now on (which is why we do not decorate the velocity as $u^{\re,\mathsf{Gr}}$). That is, it is the dimensionless combination, $\re$, rather than any individual feature that dictate the observed phenomena.  Moreover, $\re$ can be varied by considering setups where various combinations of $\mathsf{U},\mathsf{L}$ and $\nu$ are altered.  Most often, the scales $\mathsf{U}$ and $\mathsf{L}$ can greatly vary, and are responsible for producing very large Reynolds numbers in Nature and in the laboratory. For example, Figure \ref{album} shows flow of water past a fixed cylinder -- $\re$ is increasing because the inflow velocity is being systematically increased.
\begin{figure}[h!] 
  \begin{center}
    \includegraphics[width=0.8\textwidth]{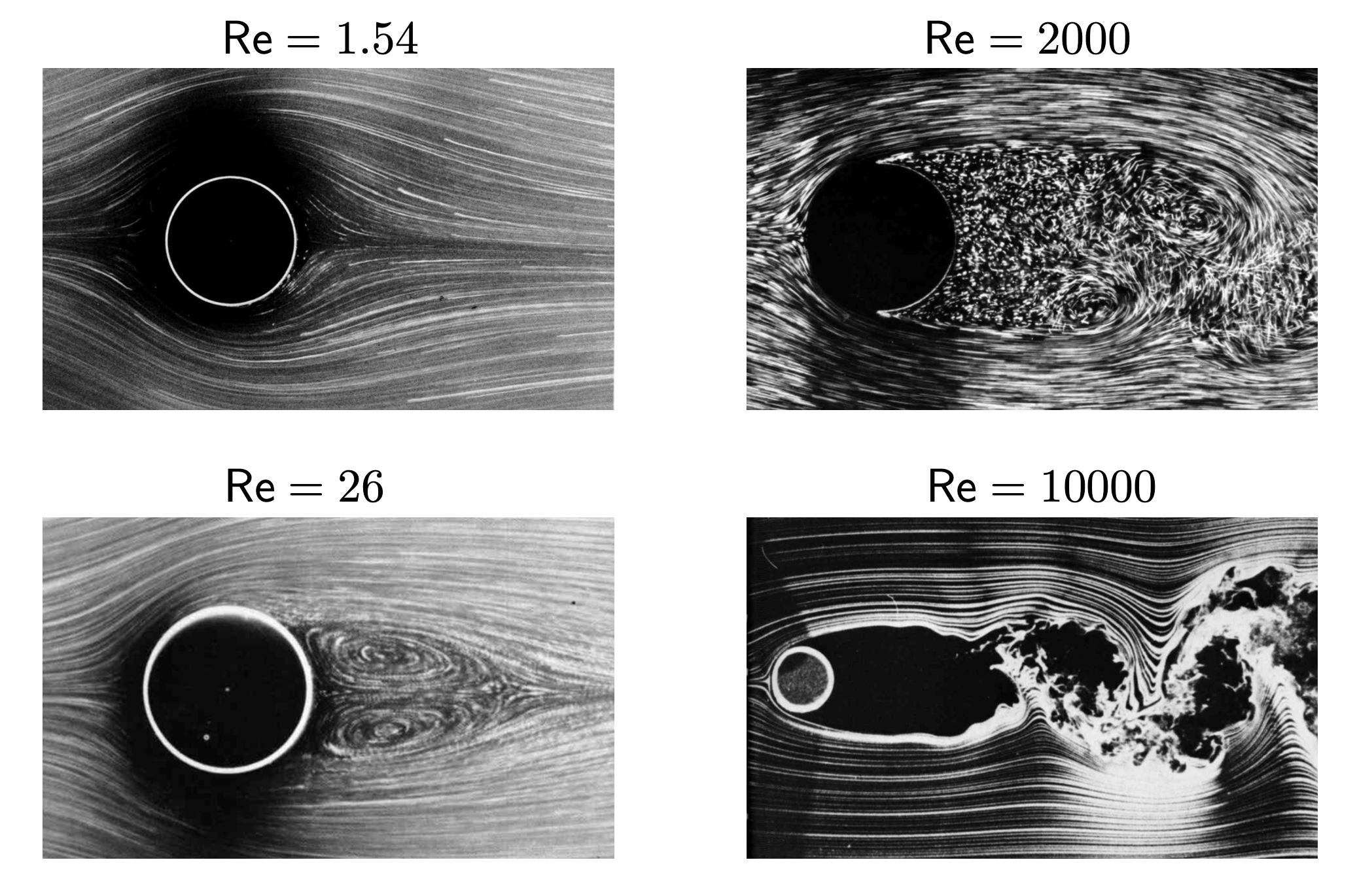}
      \end{center}
  \caption{Flow of water past a cylinder (Album of Fluid Motion \cite{VD82}).}
  \label{album}
\end{figure}

We notice a few points from the diagrams; as the Reynolds number is increased, the flow appears to simultaneously become rougher (including smaller and smaller scale motions), and also spottier (what can and will be termed intermittent). Reynolds numbers encountered in nature may be many orders of magnitude larger than those depicted.  For instance, in the wake of a commercial airline, typical Reynolds numbers are  around $\re \approx 10^{8}$, around $10^{12}$ in a tropical storm, and much much larger in astrophysics.
A clearly relevant abstraction is to continue the procedure indefinitely to probe the limit $\re\to \infty$, which we shall conflate with the \emph{inviscid limit} $\nu\to 0$ without much danger of confusing the reader.  What will likely cause confusion is that it is this limit what we shall use to operationally \emph{define} fully developed (as opposed to transitional) turbulence for the sake of our discussion.  Appropriate caveats here should be made, since the limit $\re\to \infty$ may, in some situations, not look turbulent at all.  For example, if Navier-Stokes is posed in periodic three-space and initialized with some smooth initial conditions, then for times within the period of local existence for Euler given by Theorem \ref{lethm}, Navier-Stokes converges to the classical Euler solution as $\nu\to 0$ which we have already declared to be non-turbulent.  The caveat is then that we are not in this good situation -- namely, either there are physical boundaries that trigger singular behaviors, or that the solution is becoming non-smooth in some other way, either from being forced or seeded with rough objects or by their emergence in finite time from smooth ones (Theorem \ref{tarekthm}).

Formally, the limiting $\re\to \infty$ velocity $u^\re \to u^\infty$ solves the "Euler equation"
\begin{align}\label{eel1}
``\partial_t u^\infty + u^ \infty\cdot\nabla u^ \infty&= -\nabla p^ \infty" ,\\\label{eel2}
``\nabla \cdot u^\infty &=0".
\end{align}
Thus, the motion of dry water is relevant, even universal, as a description! But in what sense?  Clearly from Figure \ref{album}, this limit is plagued by singularities. In particular,  $u^\infty$ is expected to be well below continuously differentiable, so the Euler equation cannot be satisfied in a classical sense.

\begin{figure}[h!]
  \begin{center}
    \includegraphics[width=0.43\textwidth]{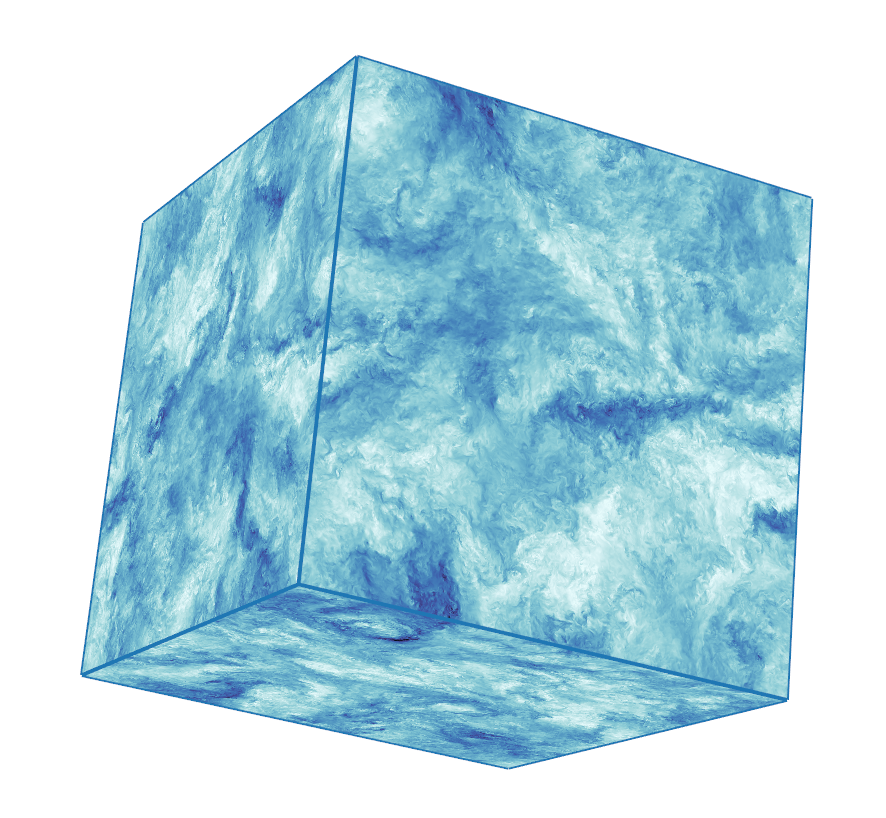}
       \hspace{-2mm} 
       {
        \includegraphics[width=0.45\textwidth]{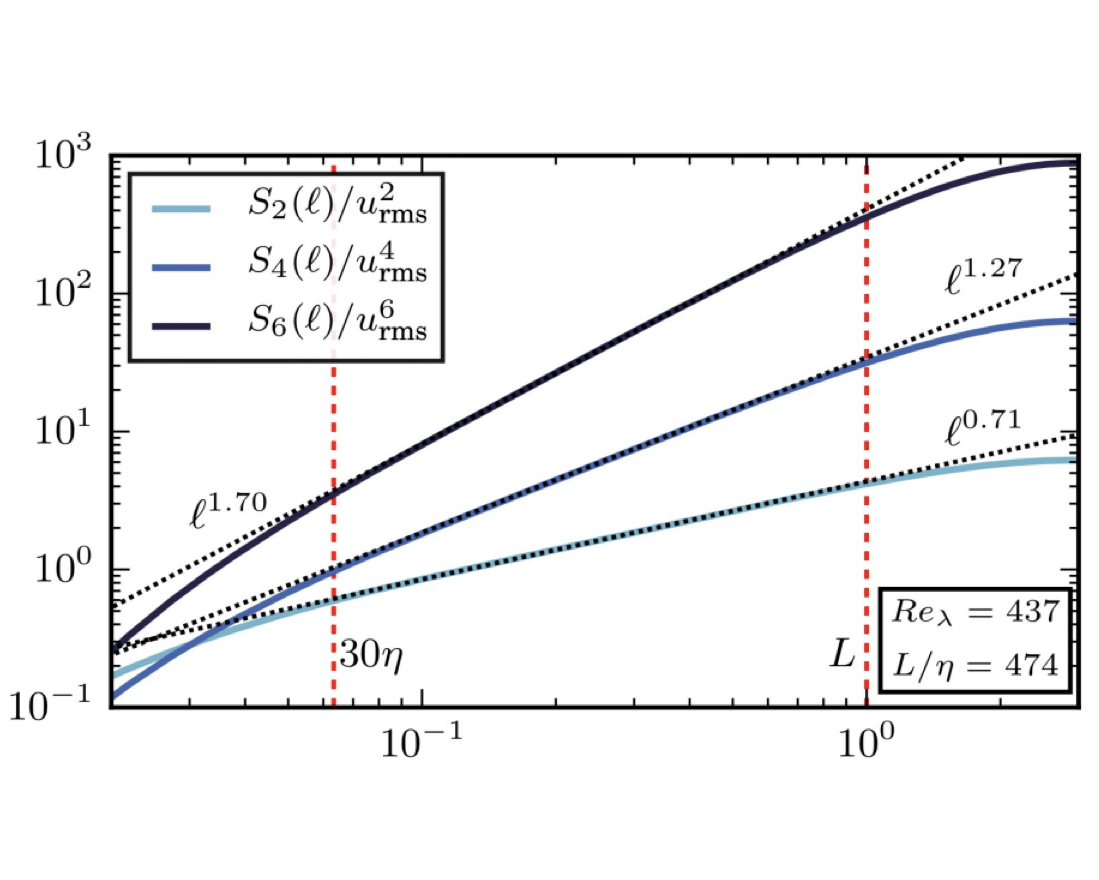}}
  \end{center}
  \caption{Turbulence in periodic box. Second, fourth and sixth order structure functions indicate emergence of non-smooth velocities. \cite{Lcomp, DJLW17} }
  \label{fixbox}
\end{figure}

Before continuing, we issue a small caveat regarding our treatment of the Navier-Stokes fluid.  Here-forth, we will assume that Navier-Stokes solutions do not suffer from any finite time singularities and are smooth enough to justify any manipulation performed.  This is not essential for four reasons.  First, all the mathematical analysis that we perform can be slightly modified to hold for Leray-Hopf weak solutions of Navier-Stokes, which always exist starting from any finite energy data \cite{CF88}.  I will not do this in these notes because it unnecessarily complicates the discussion, and also for the next three reasons. Second, it appears that Navier-Stokes singularities, if they exist, are non-generic and do not play a relevant role in the subject. Thirdly, the specific mechanism of dissipation does not, apparently, qualitatively alter the turbulent features that are our focus.  Thus, even if the Navier-Stokes model develops singularities, models with stronger -- and possibly more fundamental --  dissipative operators will not.  Finally, it has been argued that a more accurate description of non-ideal fluids is a finite dimensional stochastic dynamical system, and there is no such issue in these models, which arise only in the continuum idealization \cite{BGME22, EP25}. Luckily, we are in good company in doing this as, according to Y. Sinai \cite{VL} \emph{"Kolmogorov was never seriously interested in the problem of existence and uniqueness of solutions of the Navier-Stokes system"}.

We now return to the issue of \emph{in what sense is the Euler equation satisfied}?   Apriori, the only known bound (which holds with or without boundaries) that holds uniformly in the viscosity, is the one from the energy -- the $L^\infty_tL^2_x$ norm of the velocity.  Such weak control  allows to deduce only weak convergence as $\nu \to 0$. Due to the quadratic nonlinearity, these limits are only known to be solutions of Euler in a very weak sense, either measure valued \cite{DM87,CD95} or their alternatives such as Lion's dissipative solutions \cite{L96}.  However, in practice, a great deal more regularity than $L^2$ is retained \emph{uniformly in the viscosity}.  But this is based on observations of the system, from either physical or numerical experiment.  For example, by numerical experiments of so-called homogeneous isotropic turbulence on the periodic box.  See the left panel of Figure \ref{fixbox}. Given such a simulation of a Navier-Stokes solution $u^\nu$, one can compute a measure of its fractional regularity via its (absolute) $p$th--order \emph{structure functions}:
\begin{myshade}
\be\label{SP}
S_p^{u^\nu}(\ell) := \fint_0^T \fint_M\fint_{\mathbb{S}^{d-1}} |u^\nu(x+\ell\hat{z},t) -u^\nu (x,t)|^p \ \rmd \sigma(\hat{z}) \rmd x \rmd t,
\ee
\end{myshade}
\noindent where  $\ell>0$ and $\sigma$ is the standard surface measure on the sphere.  These objects, in the turbulent regime, will be the primary subject of our investigations in these notes.

\begin{figure}[h!]
  \begin{center}
    \includegraphics[width=0.47\textwidth]{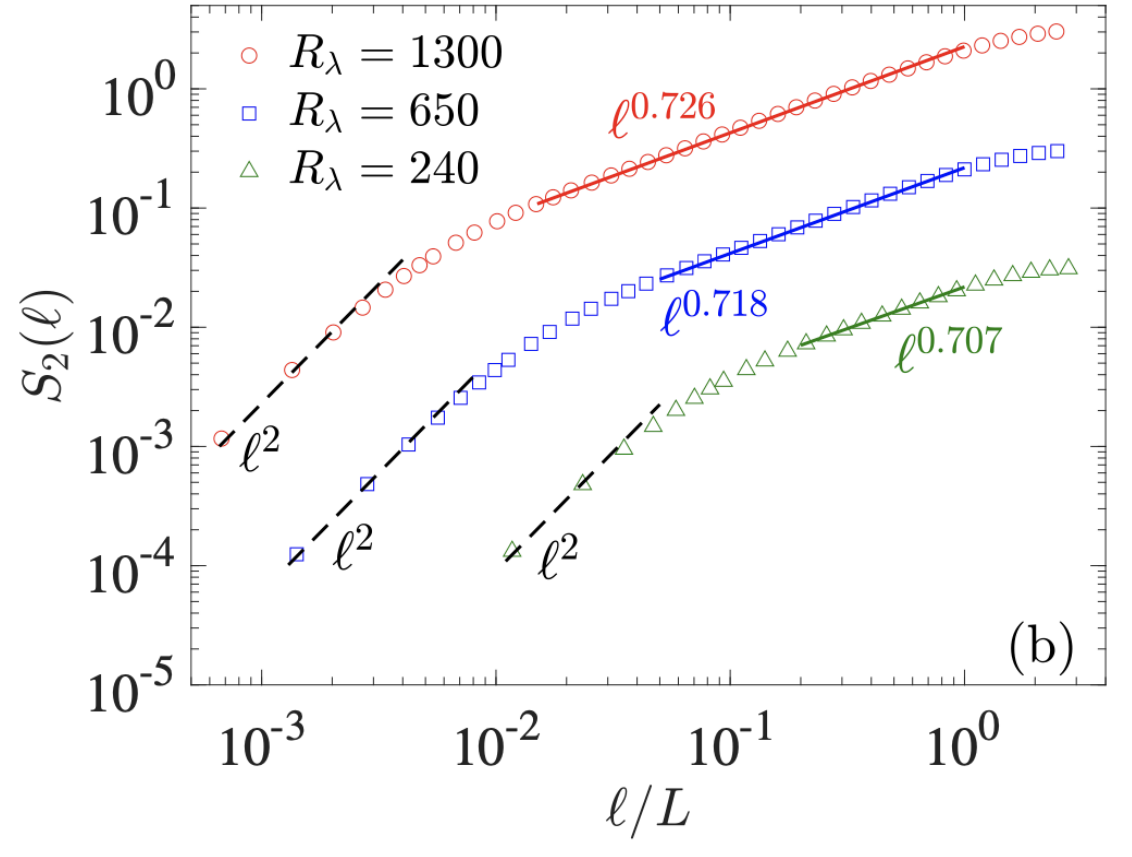}
       \hspace{2mm} 
       {
        \includegraphics[width=0.47\textwidth]{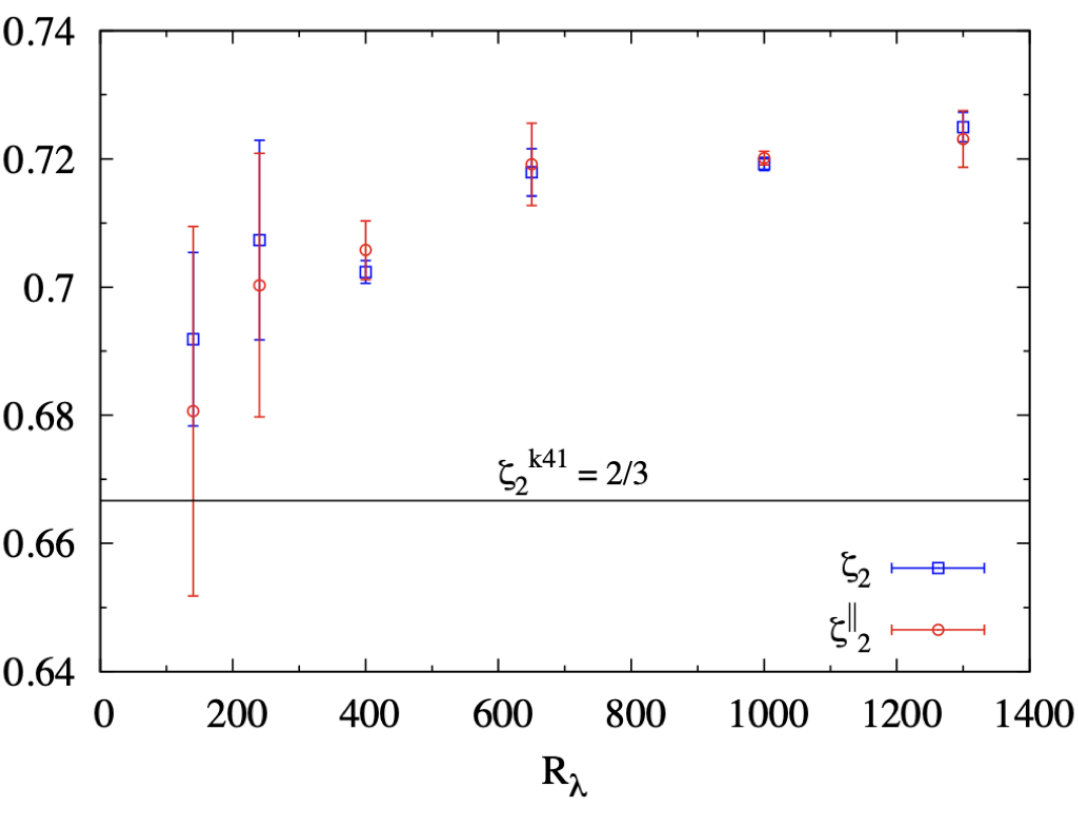}}
  \end{center}
  \caption{Measurements of second order structure functions and their exponents from numerical simulation \cite{D22}. Right panel is best fit measured exponent $\zeta_2^\nu:=\log S_2^{u^\nu}(\ell)/\log\ell$ for $\ell$ in the inertial range.}
  \label{S2}
\end{figure}

The following behavior of the $p$th--order structure functions is robustly observed: there exists a ``dissipative scale" $\ell_\nu/L\sim \nu^\gamma$ for some $\gamma>0$ and numbers $\zeta_p\in(0, p]$ such that
\be\label{Spbeh}
S_p^{u^\nu}(\ell)  \sim  \begin{cases} \nu^{-\gamma(p-\zeta_p)}  \left(\tfrac{\ell}{L}\right)^p &\text{for} \ \ 0 \leq \ell \ll \ell_\nu  \\   \left(\tfrac{\ell}{L}\right)^{\zeta_p} & \text{for} \  \ell_\nu \ll \ell \ll L
\end{cases},
\ee
where the implicit constants and $\zeta_p$ are \emph{independent of viscosity}.  In principle, all constants may depend on the time window of observation $T$ although in the "statistical steady state" this is not observed to be the case. The range of scales $\ell\in [0,\ell_\nu)$ is termed the \emph{dissipative range}, where diffusion becomes dominant and the vector field looks smooth.  The subrange of $(\ell_\nu, L)$ where the scaling $S_p^{u^\nu}(\ell) \sim \left({\ell}/{L}\right)^{\zeta_p} $ holds is called the inertial range, which is dominated by Euler-like behavior. See the right panel of Figure \ref{fixbox} for the behavior at fixed viscosity of moments $p=2,4,6$, and Figure \ref{S2} for the behavior as viscosity decreases for $p=2$.  It is a primary goal of many interested in the ``turbulence problem" to compute, based on a principled theory, what the spectrum of numbers $\{\zeta_p\}_{p\in \mathbb{N}}$ are in a "generic turbulent setting".  We shall comment more later about what can be said rigorously about their behavior. For now, we note that from the observed behavior \eqref{Spbeh} we may infer a uniform bound of 
\be\label{Spbnd}
S_p^{u^\nu}(\ell)  \lesssim \left(\tfrac{\ell}{L}\right)^{\zeta_p}
\ee
for \emph{some} $\zeta_p>0$, with the implicit constant being uniform in viscosity.  This behavior holds at least for moderately low moments $p$, say $p$ less than twelve as, in practice, moments higher than ten are difficult to stably compute.   Together with a bound on $L_t^pL^p_x$ norms, \eqref{Spbnd} gives a measure of compactness in $L^p$.  Indeed, such a vector field could be said to have $\sigma_p:=\sfrac{\zeta_p}{p}$ derivatives in $L^p$, and thus it sits in the ``time-averaged" Besov space $L^p(0,T; B_{p,\infty}^{\sigma_p}(M))$.  The following equivalence is shown in \cite{DGP25}:
\begin{lemma}\label{Slem}
Fix $p\geq 1$ and  $\zeta_p>0$.  The following two  are equivalent
\begin{enumerate}
\item $u\in L^p(0,T;L^p(\mathbb{T}^d))$ and $S_p^{u}(\ell)\lesssim \ell^{\zeta_p}$;
\item $u$ is finite in the norm $\|u\|_{L^p_t B_{p,\infty}^{\sigma_p}(\mathbb{T}^d)}:= \|u\|_{L^p_tL^p_x} + [u]_{L^p_t B_{p,\infty}^{\sigma_p}(\mathbb{T}^d)}$ with $\sigma_p:=\sfrac{\zeta_p}{p}$ and 
\be
[u]_{L^p_t B_{p,\infty}^{\sigma_p}(\mathbb{T}^d)}:= \sup_{z\neq 0} \frac{1}{|z|^{\sigma_p}} \left(\fint_0^T \| u(x+  z,t) - u(x,t)\|_{L^p(\mathbb{T}^d)}^p\rmd t\right)^{1/p}.
\ee
\end{enumerate}
\end{lemma}
We remark that, this definition of the Besov space differs slightly from the conventional one in that the sup and time integral are ordinarily reversed.  This space is slightly weaker, but in fact is more natural to measure and, as we will see from the example of the Burgers equation and Kraichnan model, also theoretically. We also remark that this result holds true locally for a bounded domain $M\subset \mathbb{R}^d$, provided the increment does not leave the domain.  We now sketch the proof of this fact, which is elementary.
\begin{proof}[Proof of Lemma \ref{Slem}]
Note first that one direction, $(2)\implies (1)$ holds trivially.  For the other direction,  note that $(1)$ easily implies   (by disintegrating the integral to shells) that
\be\label{ballavebnd}
 \fint_0^T \fint_M\fint_{B_\ell(0)} |u^\nu(x+z,t) -u^\nu (x,t)|^p \ \rmd z \rmd x \rmd t \lesssim \ell^{p\sigma_p}.
\ee
This, in turn, implies the following  mollification estimates:
\be\nonumber
\|u- \ol{u}_\ell\|_{L^p(0,T;L^p(\mathbb{T}^d))}\lesssim \ell^{\sigma_p} \qquad \text{and} \qquad \|\nabla \ol{u}_\ell\|_{L^p(0,T;L^p(\mathbb{T}^d))}\lesssim \ell^{\sigma_p-1},
\ee
where $\ol{u}_\ell(x,t) :=\int_{\mathbb{T}^d} G_\ell(r) u(x+r)\rmd r$ and $G_\ell(r) = \ell^{-d} G(r/\ell)$ with $G(z)$ a standard mollifier (positive, even, integral one, supported in the unit ball).  Note that these follow more or less immediately from the ball-averaged bound \eqref{ballavebnd}.  We conclude by noting
\begin{align*}
 \fint_0^T \| u(x+  z,t) - u(x,t)\|_{L^p(\mathbb{T}^d)}^p\rmd t &\lesssim   \fint_0^T \| u(x,t) - \ol{u}_\ell (x,t)\|_{L^p(\mathbb{T}^d)}^p \rmd t\\
 &\ \  +  \fint_0^T \| \ol{u}_\ell(x+  z,t) - \ol{u}_\ell(x,t)\|_{L^p(\mathbb{T}^d)}^p \rmd t \lesssim \ell^{p\sigma_p} (1 + \tfrac{|z|}{\ell})^p .
\end{align*}
Optimizing the bound by choosing $\ell = |z|$ gives the result.
\end{proof}

We remark further that, for Navier-Stokes solutions, there is correspondence (at order $p=2$) between the absolute structure function $S_p^{u^\nu}(\ell)$ and the \emph{longitudinal structure function}, which is an object that is far more frequently computed in experiment:
\begin{myshade}
\be\label{longsp}
S_{p,\|}^{u^\nu}(\ell) := \fint_0^T \fint_M\fint_{\mathbb{S}^{d-1}} \Big(\hat{z}\cdot(u^\nu(x+\ell\hat{z},t) -u^\nu (x,t))\Big)^p \ \rmd \sigma(\hat{z}) \rmd x \rmd t.
\ee
\end{myshade}
Specifically, it is proved that $S_2^{u^\nu}(\ell)\lesssim \ell^{\zeta_2}$ if and only if $S_{2,\|}^{u^\nu}(\ell)\lesssim \ell^{\zeta_2}$ \cite[Lemma 1.3]{D22}. 
This correspondence does not hold for a general vector field in $u\in L^2_t B_{2,\infty}^{\zeta_2/2}$ and crucially makes use of a structural property of the fluid nonlinearity $u\cdot \nabla u$ which holds for weak solutions.
 As such, all available measurements from experiment and numerical simulation (typically for longitudinal structure functions) of which the author is aware point to a uniform--in--Reynolds modulus of continuity in $L^2$, as in Figure \ref{S2}.  This, in turn, provides $L^2$ compactness, as we now describe.

Indeed,  by the Fréchet–Kolmogorov compactness theorem together with the  Aubin-Lions-Simon lemma (see \cite[Proposition 2.10]{LMP21} for a precise application), this bound suffices to ensure the existence of strong-$L_{x,t}^p$ sub-sequential limits $u^{\nu_n}\to u^\infty$.  This was observed in a number of works \cite{CG12,CV18,DE19,DN19,EP25} and is summarized by
\begin{theorem}[Emergence of weak Euler solutions] \label{cmpthm}
Fix $p\geq 2$ and suppose the $p$th--order structure function $S_p^{u^\nu}(\ell)$ satisfies the bound  \eqref{Spbnd} for some $\zeta_p>0$.  Then, the family $\{u^\nu\}_{\nu>0}$ is precompact in $L^p_{x,t}$ and all sub-sequential limits $u^{\nu_n}\to u^\infty$, possibly non-unique, are distributional (weak) solutions of the Euler equations, e.g. 
\begin{align}\label{we1}
 (u^\infty, \partial_t \varphi)_{L^2_{t,x}} +  (u^\infty\otimes u^\infty; \nabla  \varphi)_{L^2_{t,x}} &=0,\\
( u^\infty; \nabla  \psi)_{L^2_{t,x}} &=0, \label{we2}
\end{align}
for arbitrary solenoidal $\varphi\in C_0^\infty((0,T)\times\mathbb{T}^d;\mathbb{R}^d)$ and function $\psi \in C_0^\infty((0,T)\times \mathbb{T}^d;\mathbb{R})$.
\end{theorem}

So, keeping in mind the observations, it is in this weak, distributional, sense \eqref{we1}--\eqref{we2} that the ideal Euler equations emerge as the description of the infinite Reynolds number, zero viscosity, limit of non-ideal Navier-Stokes flows.  Of course, such objects maybe very irregular in space -- the compactness theorem tells that the limit object $u^\infty$ has just shy of $\sigma_p=\sfrac{\zeta_p}{p}$ derivatives in $L^p$, but need not have more.  Indeed, the behavior observed in Fig \ref{S2} indicates $\zeta_2^\infty \approx 0.72$ and not better.  

As a brief aside, one might wonder why it is reasonable to make an assumption about  space regularity and not time. At first sight, it appears unbalanced.  Interestingly,   for \emph{any} weak solution of the Euler equations,  space regularity  implies a similar type of temporal regularity (see e.g.  \cite{Isettreg}, \cite[Theorem 1.1]{CDRF} and also \cite[Theorem 4]{DE18}):

\begin{lemma}[Space regularity implies time regularity]
Let $u$ be a weak solution of the Euler equations.  For $p\geq 2$, if $u\in L^{p}(0,T; B_{p,\infty}^{\sigma_p}(\mathbb{T}^d))$ then  $u\in B_{p/2,\infty}^{\sigma_p}(0,T; L^{p/2}(\mathbb{T}^d))$, e.g. $u$ is finite in the norm $\|u\|_{B_{p/2,\infty}^{\sigma_p}(0,T; L^{p/2}(\mathbb{T}^d))}:= \|u\|_{L^{p/2}_tL^{p/2}_x} + [u]_{B_{p/2,\infty}^{\sigma_p}(0,T;L^{p/2}(\mathbb{T}^d))}$ where the seminorm is
\be
[u]_{B_{p,\infty}^{s}(0,T;L^{p}(\mathbb{T}^d))}:= \sup_{\tau\neq 0} \frac{1}{|\tau|^{s}} \left(\fint_\tau^{T-\tau} \| u(t+\tau,\cdot) - u(t,\cdot)\|_{L^p(\mathbb{T}^d)}^p \rmd t\right)^{1/p}.
\ee
\end{lemma}
Thus, up to changing the integrability exponents, turbulent flows look the same in time as they do in space.
Such a statement is important to build up our space-time picture of the structure of turbulent flow.

 Returning to our discussion of Euler, the distributional notion of solution builds in two global conserved quantities: momentum \eqref{we1}  and mass \eqref{we2} (the former holding provided the domain $M$ has symmetry).
 However, at this level of (ir)regularity, we cannot run the computations of \S \ref{sec:e} on such a weak solution to derive other conserved quantities.  
 
 Take the kinetic energy, for instance (let alone circulations, which are far more delicate \cite{E06,BDLW}).  The object itself makes sense for weak solutions, a.e. in time (being that they are at least space-time square integrable), but we are not, in general, allowed to insert $u^\infty$ as a test function into \eqref{we1} to derive a balance of energy, let alone say that 
\be\label{onethird}
``(u^\infty\otimes u^\infty; \nabla  u^\infty)_{L^2_{t,x}} =0"
\ee
which would clearly be the case for any \emph{smooth}, solenoidal $u^\infty$.  Given the form \eqref{onethird}, which is cubic in the velocity and has one derivative, one is already tempted to think that having $\sfrac{1}{3}$ of a derivative in $L^3$ should be relevant for conservation or lack-there-of. This we shall discuss in the next section, as well as the broader point that this failure of weak solutions to predict conservation of energy is a great merit of their descriptive power, rather than a weakness.
\vspace{2mm}

One standout features of the Navier-Stokes equations is that they \emph{dissipate} kinetic energy rather than conserve it.  The local form of the energy balance is
\begin{align}\label{energybal}
\partial_t \left(\tfrac{1}{2} |u^\nu|^2 \right) + \div \left((\tfrac{1}{2} |u^\nu|^2+p^\nu) u^\nu - \nu \nabla \tfrac{1}{2} |u^\nu|^2\right) &= -\varepsilon^\nu[u^\nu] + u^\nu \cdot f
\end{align}
where we have introduced the local measure of energy dissipation $\varepsilon^\nu[u^\nu]$ via
\begin{align}\label{locdissform}
\varepsilon^\nu[u^\nu] &:= \nu |\nabla u^\nu|^2.
\end{align}
Consequently,  strong solutions enjoy a \emph{global} kinetic energy balance
\begin{align}
\frac{\rmd}{\rmd t} \int_M \tfrac{1}{2} |u^\nu(x,t)|^2 \rmd x =-  \int_M \varepsilon^\nu[u^\nu](x,t) \rmd x + \int_M u^\nu(x,t) \cdot f(x,t)\rmd x.
\end{align}
What is the fate of this balance in the high Reynolds number limit $\nu \to 0$?  In a non-turbulent scenario where, say, the Euler solution starting from the same initial data is classical, then $\int_M \varepsilon^\nu[u^\nu](x,t) \rmd x\lesssim \nu$, and energy conservation is restored in the limit.  However,  a primary \emph{observed} feature of turbulence is the violation of this expectation.  Specifically, from numerous physical experiments and some numerical simulations (see e.g. Figure \ref{fig:ad}), it is observed that the turbulent regime displays \emph{anomalous dissipation}
\begin{myshade}
\be\label{anomalousdiss}
\liminf_{\nu \to 0} \fint_0^T  \fint_M\varepsilon^\nu[u^\nu](x,t) \rmd t\rmd x >0.
\ee
\end{myshade}
This observation appears so robust and primal to our understanding of turbulence (as we will see), that it has been termed the \emph{zeroth law}.  Nevertheless, we remark here that it may not be that \emph{all} regimes that could rightly be called "turbulent" behave in the same way, see \cite{IDES25} for a study of statistically stationary turbulence without walls.   Its status for real flows where turbulence is generated by boundaries (say, flow past a grid or in the wake of an airplane) seems far more clear \cite{S84,W18,NFS11}.  See \cite{Eyinkw}. It seems fitting to say that the zeroth law of turbulence is that there are zero laws of turbulence!

\begin{figure}[h!]
  \begin{center}
    \includegraphics[width=0.47\textwidth]{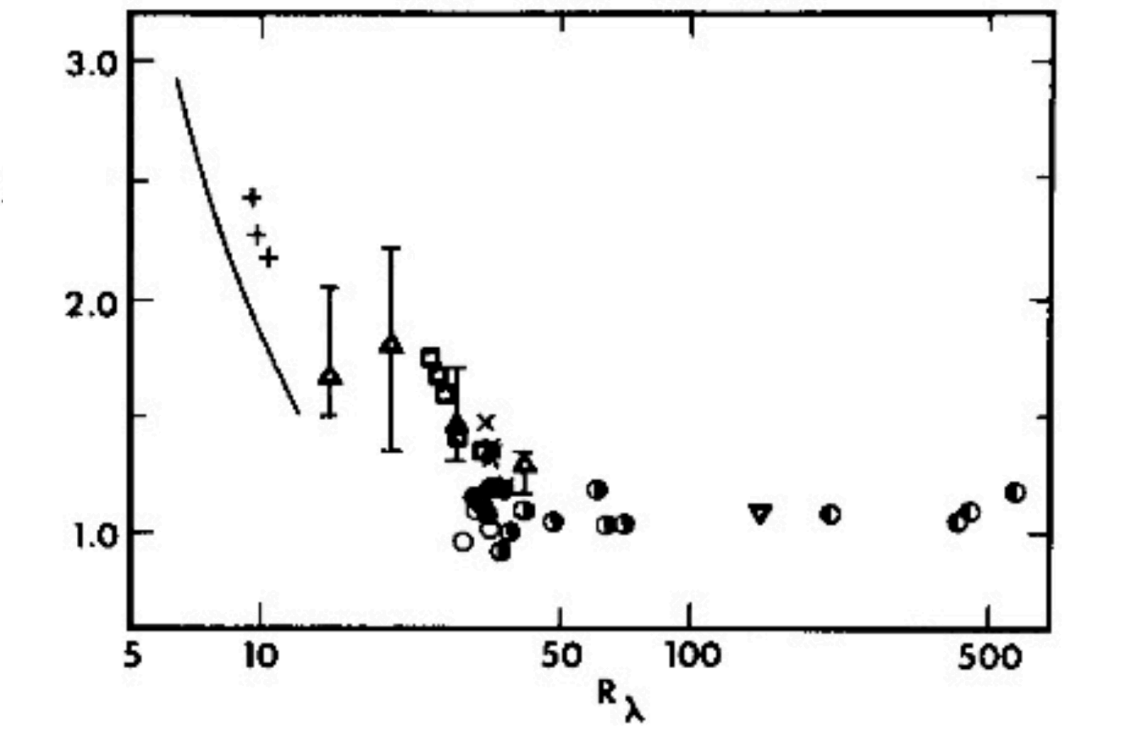}
       \hspace{-2mm} 
       {
        \includegraphics[width=0.38\textwidth]{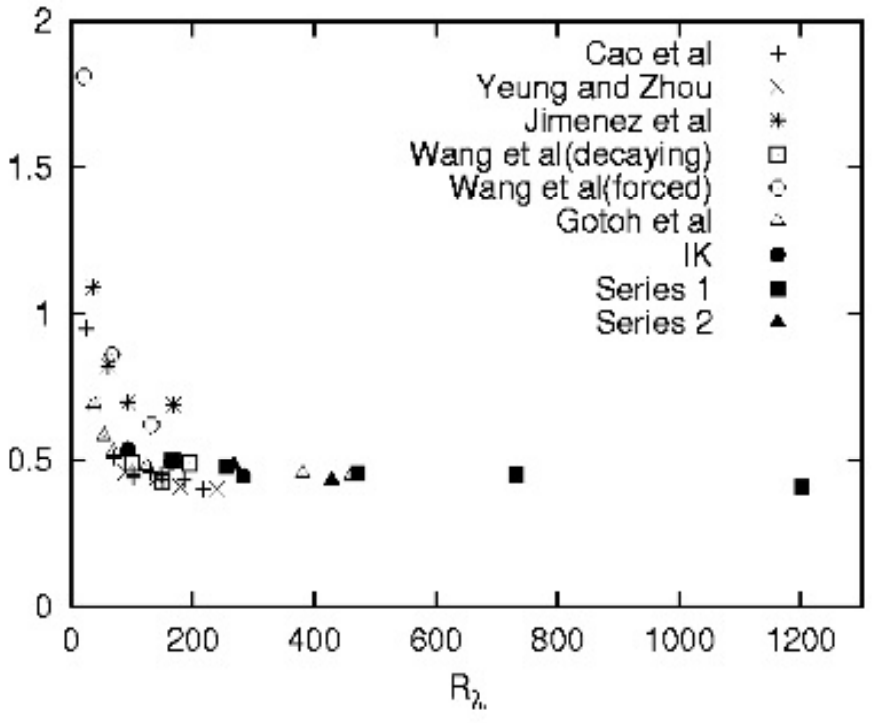}}
  \end{center}
  \caption{Dissipation anomaly from experiment \cite{S84} and numerics \cite{K03}}
  \label{fig:ad}
\end{figure}

For now, let us adopt this observation as an empirical fact, and try to understand more of its root cause.  Clearly, it requires that irregularities emerge in the limit $\nu\to 0$.  Specifically, 
$
\|\nabla u^\nu\|_{L^2_tL^2_x} \sim \tfrac{1}{\sqrt{\nu}}
$
 is required to sustain anomalous dissipation  \eqref{anomalousdiss} from the formula \eqref{locdissform}.  But can we say more about what singularities must develop?  

Let us begin by investigating the fate of the local energy balance for weak solutions.  The balance law is cubic, so to pass to the limit in the sense of distributions we should have strong compactness in $L^3_{t,x}$.  But this, according to Theorem \ref{cmpthm} results from a uniform bound on the third order structure functions $S_3^{u^\nu}(\ell)\lesssim \ell^{\zeta_3}$.  This is, again, readily observed in experiment and simulation.   See, for instance, \cite{S96,ISY20,G02,IDES25}. Thus, under this hypothesis, $u^\nu\to u$ strongly in $L^3_{t,x}$ and we may pass to the limit on the left hand side of the energy balance \eqref{energybal}:
\begin{align}\nonumber
\partial_t \left(\tfrac{1}{2} |u^\nu|^2 \right) + \div \left((\tfrac{1}{2} |u^\nu|^2+p^\nu) u^\nu - \nu \nabla \tfrac{1}{2} |u^\nu|^2\right) \quad \stackrel{\nu \to 0}{\longrightarrow} \quad \partial_t \left(\tfrac{1}{2} |u|^2 \right) + \div \left((\tfrac{1}{2} |u|^2+p) u \right)
\end{align}
where the limit is interpreted in the sense of distributions, e.g. holding upon integration against a smooth test function. As such, from the energy balance  \eqref{energybal} at finite viscosity, the right hand side must also converge in the sense of distributions
\be\label{loclimdis}
\varepsilon^\nu[u^\nu] \ \ \stackrel{\nu \to 0}{\longrightarrow}  \ \ \varepsilon[u] =:\partial_t \left(\tfrac{1}{2} |u|^2 \right) + \div \left((\tfrac{1}{2} |u|^2+p) u \right).
\ee
Note that the distribution $ \varepsilon[u]$ inherits  positivity from the approximation by $\varepsilon^\nu[u^\nu]$ and thus is, in fact,  a non-negative Radon measure.  The phenomenon of anomalous dissipation can now be phrased directly on the inviscid weak solution $u$ as
\be\label{andisse}
(\varepsilon[u] , \varphi)_{L^2_{t,x}} > 0 \qquad \text{for some smooth} \quad \varphi\geq 0.
\ee
Such weak solutions are termed  \emph{locally dissipative}. 
For the reasons discussed above, it is reasonable to believe that such dissipative weak solutions $u$ of the Euler equations \eqref{we1}--\eqref{we2} satisfying \eqref{andisse} provide  descriptions for turbulence in its Platonic form.

\section{Onsager's ideal turbulence}\label{onssec}
\vspace{2mm}

\begin{quotation}
\emph{“My tentative limiting formula for the correlation function in isotropic turbulence is not so obvious
that any one student could be expected to find it.
However, it seemed very probable to me that
somebody would have investigated the line of reasoning, which is not far fetched.”}\\
\phantom{adsf} \hfill L. Onsager in letter to C.C. Lin  \cite{ES06}
\end{quotation}
\vspace{2mm}

The result of the last section is that, under an observed uniform regularity assumption, the high Reynolds number limit of a Navier-Stokes solution is a locally dissipative weak solution of the Euler equation.  Anomalous dissipation is the phenomenon that this local dissipation measure is non-trivial.  Moreover, we see that it can be computed in two ways.  First, it can be computed as an inviscid limit of the \emph{viscous} local dissipation measure \eqref{locdissform}.  But, evidently, it can also be computed on the limiting weak solution itself, see \eqref{loclimdis}.  As such, we now seek another way to compute the dissipation on the inviscid weak solution $u$ which illuminates its ultimate cause. Recall from our earlier considerations around \eqref{onethird} that we expect having exactly $\sfrac{1}{3}$ of a derivative in $L^3_{t,x}$ to play a special role.  The person to envisage irregular weak Euler solutions as describing turbulent dynamics, and to realize this precise regularity threshold was the Nobel Laureate Lars Onsager, see left panel of  Figure \ref{fig:ons}.  He wrote a now famous statement to this effect in his 1949 paper on Statistical Hydrodynamics \cite{O49}.  There he stated that, in order for anomalous dissipation to take place, the fluid cannot satisfy a H\"{o}lderian condition of the form 
\be
|u(x+r) - u(x)|\lesssim |r|^h, \ \  \text{for} \ \ h>1/3
\ee
everywhere in the domain.
 In fact, Onsager knew much more than what he wrote, having effectively proved a strengthening of this statement in his unpublished notes, see right panel of  Figure \ref{fig:ons}.  See \cite{ES06} for an excellent discussion of this history, before which this fact was apparently unknown.  The first rigorous result in the literature towards Onsager's statement was Eyink's work \cite{E94}, followed by the work of Constantin, E and Titi \cite{CET94}.  Duchon and Robert \cite{DR00} framed the results in the context of the inviscid limit, as we have done here, and proved a stronger local statement.  Their approach is the most similar to Onsager's own unpublished work on the matter.
\begin{figure}[h!]
  \begin{center}
    \includegraphics[width=0.3\textwidth]{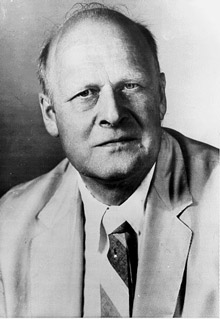}
       \hspace{5mm} 
       {
        \includegraphics[width=0.55\textwidth]{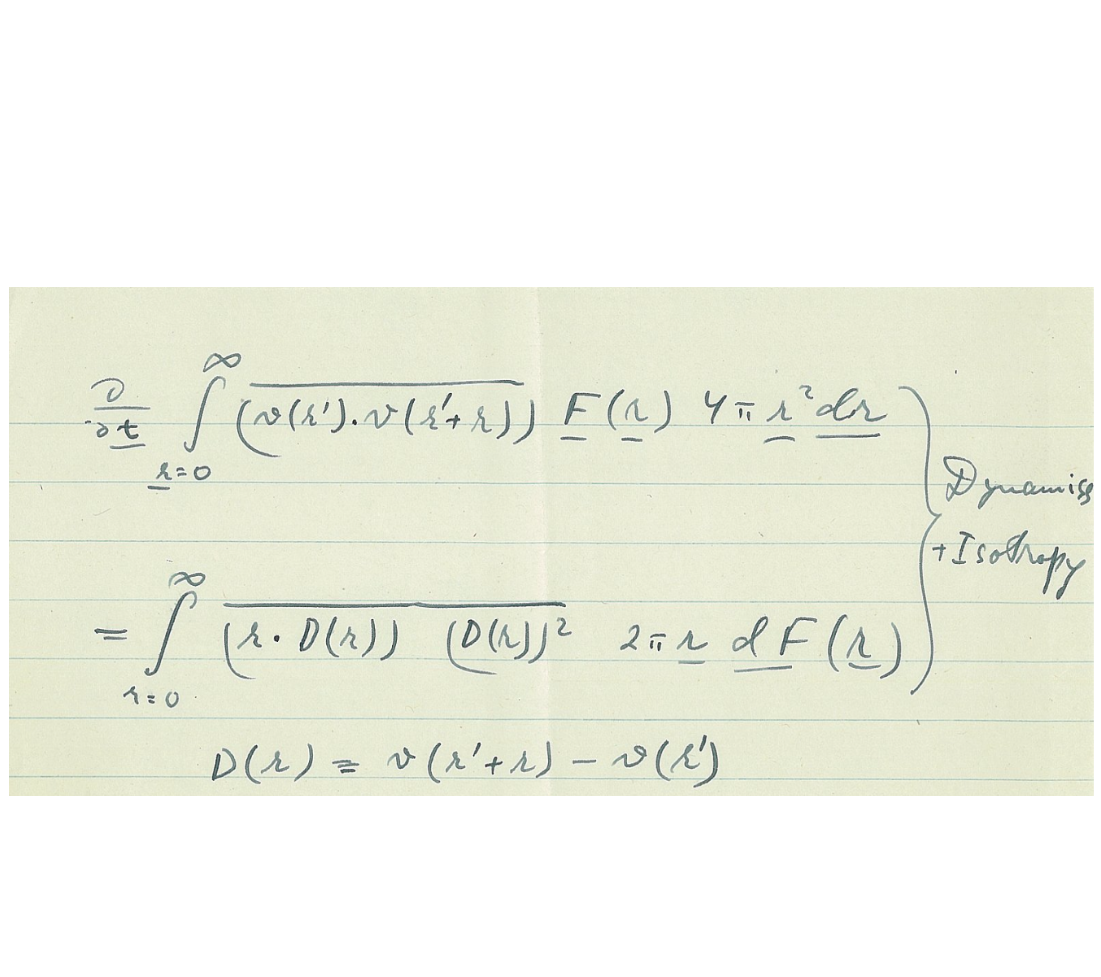}}
  \end{center}
  \caption{Lars Onsager and the anomaly \cite[Folder 11:129, p.14]{O75,ES06}}
  \label{fig:ons}
\end{figure}

Following Onsager, we first obtain a point-split energy balance (working on $M=\mathbb{T}^d$):
\begin{proposition}\label{invthm}
 Let $u^\nu\in L^3_{t,x}$ be a weak solution to \eqref{ns1}--\eqref{ns2} for $\nu\geq 0$. If $\nu>0$ we assume the solution is smooth. Denote by $\ol{u}^{\nu}_\ell =u^\nu *G_\ell$ the space mollification of $u^\nu$ with an even kernel $G$. Also, set $\delta_{\ell z} u^\nu(x,t):=u^\nu(x+\ell z,t)-u^\nu(x,t)$ and 
\be \label{Dell}
    \ve_\ell[u^\nu]:= \frac{1}{4}\int_{B_1(0)} \nabla G(z)\cdot \frac{\delta_{\ell z} u^\nu(x,t)}{\ell} |\delta_{\ell z} u^\nu(x,t)|^2\,\rmd z,
  \ee
    and 
    $$
    J_\ell^\nu[u^\nu]:=\tfrac{1}{2}(u^\nu\cdot \ol{u}^{\nu}_\ell) u +\tfrac{1}{2}(p^\nu \ol{u}^{\nu}_\ell +\ol{p}^{\nu}_\ell u)-\tfrac{\nu }{2}\nabla   (u^\nu\cdot \ol{u}^{\nu}_\ell)+ \tfrac14\big( \ol{(|u^\nu|^2 u^\nu)}_\ell - \ol{|u^\nu|^2}_\ell u^\nu \big)  .
    $$
    Then the identity 
   \begin{align}  \label{moll en identity NS}
       \partial_t \tfrac{1}{2}(u^\nu\cdot \ol{u}^{\nu}_\ell) &+ {\rm div} J_\ell^\nu [u^\nu] =-\nu\nabla u^\nu:\nabla \ol{u}^{\nu}_\ell -    \ve_\ell[u^\nu] + \tfrac{1}{2}( u^\nu\cdot \ol{f}_\ell  + \ol{u}_\ell^\nu\cdot {f}  )
   \end{align}
   holds in the sense of distributions.   
\end{proposition}
An immediate consequence of this result is the following integrated balance for $\nu,f=0$:
\be \nonumber 
\frac{\rmd}{\rmd t} \int_{\mathbb{T}^d} \int_{\mathbb{R}^d}  u(x)\cdot u(x+ z) {F}_\ell(|z|)\rmd z \rmd x= - \frac{1}{2} \int_{\mathbb{T}^d} \int_{\mathbb{R}^d} \tfrac{z}{|z|}\cdot \delta_{ z} u(x,t) |\delta_{ z} u(x,t)|^2\, {F}_\ell'(|z|)\rmd z \rmd x.
\ee
where we have taken $G(z) = F(|z|)$ to be spherically symmetric, and replace his isotropy assumption with integration over the fluid vessel (interpreting his overline as a spatial average, in this case over $\mathbb{T}^d$),  then one immediately recovers Onsager's formula for dimension $d=3$ from his notes (left panel of Figure \ref{fig:ons}).
\begin{proof}[Proof of Proposition \ref{invthm}]
We drop superscripts of $\nu$ and set $f=0$ in the course of the proof. Since  $u$ is a weak solution to \eqref{ns1} we have 
    \begin{equation}\label{NS_weak_sol}
( u ; \partial_t \varphi )_{L^2_{t,x}}+ (u \otimes u \ ;\ \nabla \varphi)_{L^2_{t,x}}-\nu  ( \nabla u \ ; \  \nabla  \varphi)_{L^2_{t,x}}+ (p ,\div \varphi)_{L^2_{t,x}} = 0 \qquad \forall \varphi \in C^2_0.
\end{equation}
We need to discuss the pressure; it can be recovered a posteriori as the unique zero average solution to the elliptic problem 
\begin{equation}\label{pressure eq}
-\Delta p =\div \div (u\otimes u) \ \  \text{on } \mathbb{T}^d.
\end{equation}
This Poisson equation for the  pressure  is formally obtained by taking the divergence of \eqref{ns1}.  This can be made more precise, but we omit this here. We simply require
\begin{lemma}[Pressure Regularity] \label{P: pressure CZ}
    If $u\in  L^q([0,T];L^p(\mathbb{T}^d))$ for some $q\in [1,\infty]$ and  $p\in (2,\infty)$, then the unique zero average weak solution to \eqref{pressure eq} satisfies 
    $$
     \|p\|_{L^{\frac{q}{2}}_t L^{\frac{p}{2}}_x}\leq C  \|u\|^2_{L^q_t L^p_x}.
    $$
\end{lemma}
\begin{proof}[Proof of Lemma \ref{P: pressure CZ}]
    We have $p=\mathcal K (u\otimes u)$ for a Calderón–Zygmund operator $\mathcal K$. Clearly $ \|u(t)\otimes u(t)\|_{L^{\frac{p}{2}}_x}\leq C \| u(t)\|^2_{L^{p}_x}.$
The lemma follows by continuity of Calderón–Zygmund operators in the spaces under consideration, together with an integration in time.
\end{proof}
Thus,  the pressure enjoys the regularity $p\in L^\frac{3}{2}_{x,t}$ by Lemma \ref{P: pressure CZ}.  We now introduce mollification of the weak solution $\ol{u}_\ell$, on which we may perform manipulations. The mollification procedure  was described already  in the proof of Lemma \ref{Slem}, see Figure \ref{fig:moll} for a visualization.  Specifically,  we choose $\ol{u}_\ell$ as our test function in weak formulation of the Euler/Navier-Stokes equations \eqref{we1}--\eqref{we2}. 
\begin{figure}[h!]
  \begin{center}
    \includegraphics[width=0.7\textwidth]{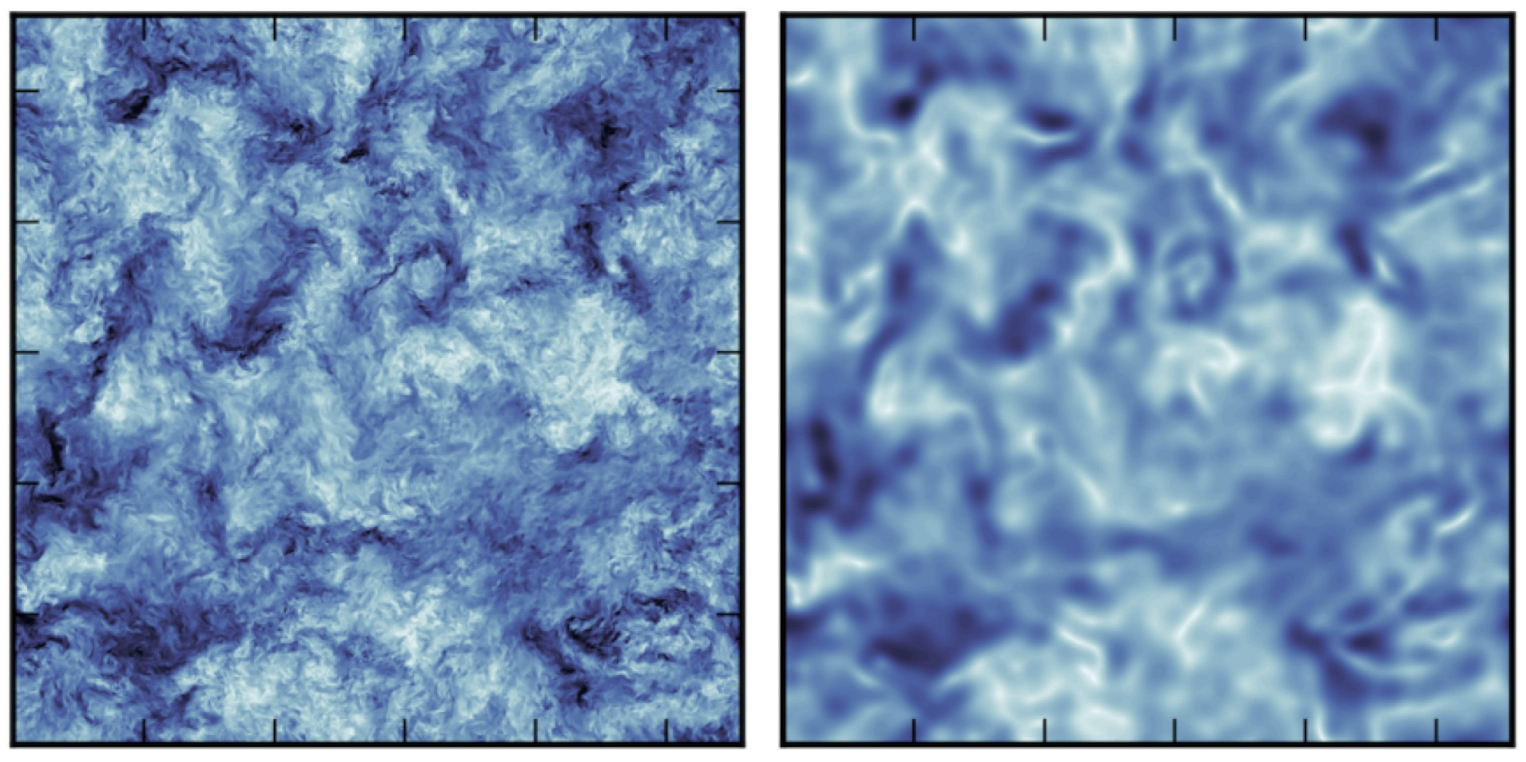}
  \end{center}
  \caption{Side by side comparison of fine grained velocity $u$, and its coarse grained counterpart $\ol{u}_\ell$ (taken from \cite{DJLW17}). }
  \label{fig:moll}
\end{figure}
Clearly $\ol{u}_\ell \in L^3_t C^\infty_x$.  Moreover, by
$$
\partial_t \ol{u}_\ell = - \div \ol{(u \otimes u)}_\ell - \nabla \ol{p}_\ell +\nu \Delta \ol{u}_\ell
$$
we also get $\partial_t \ol{u}_\ell \in L^{\frac{3}{2}}_{x,t}$  (better, in fact. See \cite[Proposition 2]{D19}). This allows to choose $\varphi = \psi \ol{u}_\ell$ as a test function in \eqref{NS_weak_sol}, for any arbitrary scalar $\psi\in C^\infty_0(\mathbb{R}\times \mathbb{T}^d)$. We get
\begin{align*}
0&=(  u \cdot \partial_t \ol{u}_\ell, \psi)_{L^2_{t,x}} +(  u \cdot \ol{u}_\ell, \partial_t\psi)_{L^2_{t,x}}  +( u \cdot \nabla \ol{u}_\ell \cdot  u ,  \psi)_{L^2_{t,x}} + (  (u\cdot \ol{u}_\ell)  u; \nabla \psi)_{L^2_{t,x}}   \\
&\quad+ ( p \ol{u}_\ell ; \nabla \psi)_{L^2_{t,x}} -\nu ( \nabla u \ ; \ \nabla \psi \otimes \ol{u}_\ell)_{L^2_{t,x}} - \nu ( \nabla u : \nabla \ol{u}_\ell, \psi)_{L^2_{t,x}}  \\
&=-( u\cdot \div \ol{( u \otimes u)}_\ell, \psi)_{L^2_{t,x}}+ ( \ol{p}_\ell u ; \nabla \psi)_{L^2_{t,x}} +(  u \cdot \ol{u}_\ell, \partial_t\psi)_{L^2_{t,x}} \\
&\quad+(  u \otimes u:\nabla \ol{u}_\ell ,\psi)_{L^2_{t,x}} +(  (u\cdot \ol{u}_\ell) u; \nabla \psi)_{L^2_{t,x}} +( p \ol{u}_\ell ; \nabla \psi)_{L^2_{t,x}} \\
&\qquad+\nu ( u\cdot  \ol{u}_\ell ,\Delta \psi)_{L^2_{t,x}}  -2 \nu ( \nabla u : \nabla \ol{u}_\ell, \psi)_{L^2_{t,x}} .
\end{align*}
where, by the incompressibility of $u$ we have used
$
- ( u \cdot \nabla \ol{p}_\ell, \psi)_{L^2_{t,x}} = ( \ol{p}_\ell u ; \nabla \psi)_{L^2_{t,x}}
$
and we integrated by parts in the viscous terms. Moreover, by expanding the cube we obtain 
\begin{align*}
    4\ve_\ell[u]&= -\div \ol{(|u|^2 u)}_\ell  + u \cdot \nabla \ol{(|u|^2)}_\ell  + 2 u\cdot \div \ol{(u\otimes u)}_\ell - 2 u\otimes u :\nabla \ol{u}_\ell.
\end{align*}
Thus, by reassembling the terms in the above identity we establish the identity. \end{proof}

\begin{remark}[Longitudinal flux]\label{45thlawrem}
A similar balance as the one established in  Proposition \ref{invthm}  holds for an important related quantity.   Specifically, one can prove that 
\be \label{Dell45}
    \ve_\ell^\|[u^\nu]:= \frac{d(d+2)}{12 } \int_{B_1(0)} \nabla G(z)\cdot \frac{\delta_{\ell z} u^\nu(x,t)}{\ell} |z \cdot \delta_{\ell z} u^\nu(x,t)|^2\,\rmd z,
  \ee
appears as the flux  in a similar distributional identity to \eqref{moll en identity NS}. 
 See Eyink \cite{E03} and  Novack \cite{N24}.  This form is related to the celebrated Kolmogorov $\sfrac{4}{5}$ law, as we shall see.
\end{remark}

An immediate corollary of Proposition \ref{invthm} captures the approach of Onsager and reproduces the result of  Duchon--Robert \cite{DR00}  (see \cite{Dubrulle18} for experimental confirmation):
\begin{corollary}[Local energy balance]\label{invthmcor4}
 Any $u\in L^3_{t,x}$ weak Euler solution satisfies 
\be\label{disenbal}
 \partial_t \left(\tfrac{1}{2} |u|^2 \right) + {\rm div} \left((\tfrac{1}{2} |u|^2+p) u \right) =-\ve [u]+ u \cdot f
\ee
in the sense of distributions,   where the defect term is given by the distributional limit 
\be\label{fluxbal}
\ve[u] =\lim_{\ell\rightarrow 0}\ve_\ell[u]
\ee
 with $\ve_\ell[u]$ defined by \eqref{Dell}.
    \end{corollary}
    \begin{remark}
     Recall from \S \ref{nssec}, if $u$ arises from a strong $L^3$ limit of Navier-Stokes solutions $u^\nu$ then $\lim_{\nu\to 0} \ve^\nu[u^\nu] = \ve[u]$ holds in the sense of distributions. As such, Corollary \ref{invthmcor4} identifies the non-linear Euler energy flux with the viscous dissipation:
     \be\label{fluxbalvisc}
   \lim_{\nu\to 0} \ve^\nu[u^\nu] = \lim_{\ell\rightarrow 0}\ve_\ell[u].
     \ee
    \end{remark}
\begin{proof}[Proof of Corollary \ref{invthmcor4}]
Since translations are strongly continuous in $L^p$, it is easy to see
\begin{align*}
  (  J_\ell[u] , \psi)_{L^2_{t,x}}  &\to (\tfrac{1}{2}|u|^2 ,\partial_t\psi )_{L^2_{t,x}}+ (\tfrac{1}{2}|u|^2+p )u ; \nabla \psi)_{L^2_{t,x}},\\
\left(\tfrac{1}{2}( u\cdot \ol{f}_\ell  + \ol{u}_\ell\cdot {f}  ), \psi\right)_{L^2_{t,x}} &\to (u \cdot f, \psi)_{L^2_{t,x}}.
\end{align*}
It follows that the distributional limit of $\ve_\ell[u]$ exists and must obey \eqref{disenbal}.  
\end{proof}

Thus, we return to Onsager's assertion about energy conservation.  The first published results are due to Eyink  \cite{E94} and Constantin, E and Titi \cite{CET94}.  The final result is
\begin{theorem}[Energy Conservation]\label{invthmcor3}
 Fix $p\geq 3$ and $\sigma_p>1/3$.  Any  weak Euler solution in $u\in  L^p(0,T; B_{p,\infty}^{\sigma_p}(\mathbb{T}^d))$   or, equivalently,
\be
S_p^u(\ell)\lesssim \ell^{\zeta_p} \ \text{ with }\  \zeta_p>\frac{p}{3},
\ee
has a trivial energy defect  $\ve[u]=0$. If also $u\in C(0,T; L^2(\mathbb{T}^d))$,  energy is conserved.
\end{theorem}
\begin{proof}[Proof of Theorem \ref{invthmcor3}]
In view of Corollary \ref{invthmcor4}, we bound $\ve_\ell[u]$. Since $p\geq 3$ we have
\be\label{fluxbound}
    \|\ve_\ell[u]\|_{L^{1}_{x,t}}\leq     \|\ve_\ell[u]\|_{L^{p/3}_{x,t}}\leq  \frac{C}{\ell}\int_{\mathbb{R}^d} |\nabla G(z)| \|\delta_{\ell z}u\|^3_{L^p_{x,t}}\,\rmd z\lesssim  \ell^{3\sigma_p -1}\|u\|^3_{L^p_tB^{\sigma_p}_{p,\infty}},
   \ee
    which vanishes if $\sigma_p>\frac13$. Thus   $\ve[u]=0$ is trivial for solutions $u\in L^p(0,T; B_{p,\infty}^{1/3+}(\mathbb{T}^d))$.
\end{proof}

Theorem \ref{invthmcor3} proves a strengthened version of Onsager's original 1949 assertion that the flow cannot possess more that $\sfrac{1}{3}$ of a derivative and be consistent with anomalous dissipation.  The "negative side" of this conjecture was finally proved by Isett \cite{Isett} where he exhibited a weak solution of Euler with slightly less than $\sfrac{1}{3}$ of a derivative that did not conserve energy (such solutions can, in fact, be dissipative \cite{BDSV}). This followed a remarkable series of works on the subject of constructing wild weak solutions via a technique called convex integration \cite{Sch,Shn,DL1,DL2}. See the comprehensive review by Buckmaster and Vicol \cite{BV20}.  The union of these two sides, positive and negative, is termed Onsager's conjecture.  However, an important question about the excluded middle remains open (and therefore, so does "Onsager's conjecture"):

\begin{question}\label{consques}
Suppose $u\in L^\infty(0,T; L^2(\mathbb{T}^d)) \cap L^p(0,T; B_{p,\infty}^{1/3}(\mathbb{T}^d))$ for any $p\geq 3$ is a weak solution of Euler.  Does it conserve energy?
\end{question}

We are mainly interested in the case without forcing $f=0$, but it's reasonable to consider it for $f\in L^1(0,T;L^2(\mathbb{T}^d))$.  Note, being finite energy $u\in L^\infty_t L_x^2$, one can choose a representative $u\in C_tL^2_w$ (namely, continuous in time, weakly in $L^2_x$).  It is simple to see that this object satisfies the integrated energy balance for times $0\leq s\leq t\leq T$:
\be
\tfrac{1}{2} \|u(t)\|_{L^2}^2 =\tfrac{1}{2} \|u(s)\|_{L^2}^2 - \int_s^t \int_{\mathbb{T}^d} \ve[u(x,\tau)]\rmd x \rmd \tau +  \int_s^t \int_{\mathbb{T}^d} u(x,\tau)\cdot f(x,\tau)\rmd x \rmd \tau
\ee
where $\ve[u] =\lim_{\ell\rightarrow 0}\ve_\ell[u]$
 with $\ve_\ell[u]$ defined by \eqref{Dell}.
Moreover $\int_{\mathbb{T}^d} \ve[u(x,t)]\rmd x $ is $L^1_t$ since $u\in L^3_tB_{3,\infty}^{1/3}$\footnote{This requires the true Besov space (not as in  Lemma \ref{Slem}), with supremum taken before time integral.}.  This means, that  $ \int_s^t \int_{\mathbb{T}^d} \ve[u(x,\tau)]\rmd x \rmd \tau$ is continuous in $s$ and $t$, giving zero mass to zero measure time intervals.

Question \ref{consques} is particularly important because these spaces, in particular $u\in L^3_tB_{3,\infty}^{1/3}$, are natural candidates for the regularity of \emph{actual} turbulent flows. We will discuss this further point  in \S \ref{K41sec}.  For the moment, I remark that there is only one ``critical space" (namely a space where the flux $\ve_\ell[u]$ is bounded but does not vanish for general vector fields in that class) for which the answer to Question \ref{consques} is known.  This critical space is $u\in L^\infty_{t,x}\cap L^1_t BV$, in which De Rosa and Inversi showed weak solutions \emph{conserve energy} \cite{DI24} (see also \cite{I25}).   This result is remarkable since  $L^\infty_{t,x}\cap L^1_t BV_x\subset L^3_t B_{3,\infty}^{1/3}$ sharply. Indeed, if $\|u(\cdot+ r)- u(\cdot)\|_{L^1}\lesssim |r|$ and $\|u\|_{L^\infty}<\infty$, then by interpolation $\|u(\cdot+ r)-u(\cdot)\|_{L^p}\lesssim |r|^{1/p}$ for all $p\geq 1$.  Moreover for compressible models such as Burgers, shocks live in  $u\in L^\infty_{t,x}\cap L^1_t BV_x$ and \emph{are} dissipative.  As for Question \ref{consques}, I would guess the answer is no, although recent work \cite{DI24} as well as \cite{IDES25} brings into question some of the intuition behind this guess. This will be discussed around Conjecture \ref{conj45}.

\begin{remark}[Vortex sheets conserve energy] In this remark, we will show that regular vortex sheets are conservative, despite living in the critical regularity class $L^\infty_tB_{3,\infty}^{1/3}$ and no better.  For the same reason as for \cite{DI24}, this result distinguishes the incompressible case from the compressible case of shocks. This statement was first proved by Shvydkoy \cite{Sh9} and then in \cite{DIN24} whose discussion we follow. It can also be understood as a consequence of the aforementioned result \cite{DI24}, but it is simple to prove directly. We do this here. We need a preliminary result about bounded measure divergence fields:
\begin{lemma}\label{sheetlem}
Let $v:\mathbb{R}^{N}\to \mathbb{R}^{N}$ be a bounded vector field. Suppose that $\nabla \cdot v= \mu$, where $\mu$ is a Radon measure.  Then, for  any smooth oriented codimension one surface $\Sigma$ in space-time with normal $\hat{n}$, we have
\be
\mu(\Sigma)=\int_\Sigma [\![ v ]\!]\cdot \hat{n} \ \rmd \sigma 
\ee
 where $v^{\mp}$ are the traces (provided they exist) of $v$ along $\Sigma$ and $ [\![ v ]\!]= v^--v^+$.
\end{lemma}
\begin{proof}[Proof of Lemma \ref{sheetlem}]
We do not want to get into the technicalities of when an appropriate trace exists.  Suffice it to say that having uniform $C^\alpha(\mathbb{R}^N)$, $\alpha\in (0,1)$ bounds up to $\Sigma$ implies they do; see  \cite{DIN24} for more general conditions.
Let $\varphi \in C_0^\infty(\mathbb{R}^{N})$ with ${\rm supp} \ \! ( \varphi )\subset B$ where $B$ separated by  $\Sigma$  into $B= B^-\cup B^+\cup (\Sigma\cap B)$. The claim follows from:
\begin{align*}
\int_{\mathbb{R}^{N}} \varphi \rmd \mu &=- \int_{B}v\cdot \nabla \varphi \rmd x 
=\int_\Sigma (v^+ -v^-)\cdot \hat{n} \ \rmd \sigma +  \int_{B^- \cup B^+} \varphi \rmd \mu.
\end{align*}
\end{proof}
With this in hand, we can establish energy conservation for vortex sheets. The key point is we have $d+2$ space-time vector fields which have measure divergences:
\begin{enumerate}
\item 1 mass $\nabla_{t,x}\cdot (1,u)=0$
\item $d$ momenta $\nabla_{t,x}\cdot (u,u\otimes u+p)=0$
\item 1 energy  $\nabla_{t,x}\cdot \big(\frac{1}{2}|u|^2,(\frac{1}{2}|u|^2+p)u\big)=-\ve[u]$
\end{enumerate}
Suppose that $(u,p)$ is a regular vortex sheet, namely a weak solution of Euler which is bounded and smooth away from a co-dimension one spacetime surface $S\subset \mathbb{R}^{d+1}$.  We also assume it is dissipative, namely that  $\ve[u]$ is a non-negative measure although this isn't truly necessary.  We aim to show $\ve[u]\equiv 0$. For this we need only prove that $\ve[u](S)=0$, since $(u,p)$ are smooth away from $S$.   Letting $n=(n_t,n_x)$, from (1) and (2) we deduce
\begin{enumerate}
\item[(1$'$)] $[\![ u \cdot n_x ]\!]=0$
\item[(2$'$)] $[\![ u n_t ]\!]+ [\![ u u_n + p n_x ]\!]=0$
\end{enumerate}
where $u_n= u \cdot n_x$. Dotting (2$'$) with $n_x$, we have $[\![ p ]\!]=0$,
consistent with the vortex sheet being a contact discontinuity.  Whence (2$'$)  becomes
\be\label{cons1}
[\![ u ]\!] (n_t + u_n)=0.
\ee
Finally note that,  by Lemma \ref{sheetlem}, for any codimension one surface $\Sigma$ we have
\be\label{cons2}
-\ve[u](\Sigma) = [\![ \tfrac{1}{2}|u|^2 ]\!] n_t + [\![ \tfrac{1}{2}|u|^2 u_n ]\!] = [\![ \tfrac{1}{2}|u|^2 ]\!](n_t+ u_n).
\ee
With $  u_{\rm ave}=\tfrac{1}{2}( u^-+ u^+)$ we have 
$
 [\![ \tfrac{1}{2}|u|^2 ]\!]=  [\![ u ]\!]  \cdot u_{\rm ave}.
$
Thus, by \eqref{cons1} and \eqref{cons2},
\be
-\ve[u](\Sigma) = u_{\rm ave}  [\![ u ]\!] (n_t+ u_n)=0.
\ee
This completes the proof that vortex sheets conserve energy.
\end{remark}

\begin{remark}[Wall bounded flows]
In our whole discussion, we have ignored the effect of solid walls on the flow.  Of course, this is a largely  unphysical assumption as most turbulence in most natural flows is triggered or even driven by shedding of vorticity from the wall.  The status of anomalous dissipation in wall bounded flows is, in some context, far more clear than it is for spatially periodic flows.  For instance, see \cite{W18} for a clean numerical example in two-dimensional flow.
We cannot hope to give a complete discussion of this very interesting and rich subject here, having decided to focus only on the simplest possible aspects of turbulent flow. Instead, we mention the great recent review paper by Eyink \cite{Eyinkw} and point to the works \cite{DN18, BTW19, EKQ22, EQ25} and references therein.
\end{remark}

To conclude our discussion, we  state one additional identity concerning the viscous dissipation  $\ve^\nu[u^\nu]$ (it can also apply directly to the Duchon-Robert distribution $\ve[u]$).  It was derived in \cite{DDII25} and is a local form of the integrated identity used in \cite{DE19}.  We will then apply it to give quick proofs of some known results. 
\begin{proposition}[Energy dissipation identity] 
    \label{P:decomposition_NS}
    Let $\nu \geq 0$. Let $u^\nu$ be a strong solution to Navier-Stokes. Letting $u^\nu_\ell$ be the space mollification of $u^\nu$, we set  
     \begin{align}
     E^{\ell,\nu}&:=\frac{|{u^\nu-\ol{u^\nu}_\ell}|^2}{2}, \\
     Q^{\ell,\nu} &:= \left(\frac{|{u^\nu-\ol{u^\nu}_\ell}|^2}{2} + (p^\nu-\ol{p^\nu}_\ell) \right) (u^\nu-\ol{u^\nu}_\ell) ,\\
     R^{\ell,\nu}&:= \ol{u^\nu}_\ell\otimes \ol{u^\nu}_\ell - \ol{(u^\nu\otimes u^\nu)}_\ell, \\
     C^{\ell,\nu} &:= (u^\nu-\ol{u^\nu}_\ell) \cdot {\rm div}  R^{\ell,\nu} + (u^\nu-\ol{u^\nu}_\ell)\otimes (u^\nu-\ol{u^\nu}_\ell):\nabla \ol{u^\nu}_\ell.
     \end{align}
For any $\ell>0$ we have the identity
     \begin{equation}
         \label{D_decomp_NS}
-\ve^\nu[u^\nu]= (\partial_t+u^\nu_\ell \cdot \nabla - \nu\Delta) E^{\ell,\nu} +  {\rm div}  Q^{\ell,\nu} + C^{\ell,\nu} + \nu |{\nabla \ol{u^\nu}_\ell}|^2 -2\nu  \nabla u^\nu:\nabla \ol{u^\nu}_\ell.
     \end{equation}
     When $\nu=0$, the  \eqref{D_decomp_NS} holds with $\ve^\nu[u^\nu]$ replaced by the defect distribution $\ve[u]$.
\end{proposition}

We don't prove Proposition \ref{P:decomposition_NS} here -- it is a computation (see \cite{DDII25}).  The utility of the identity \eqref{D_decomp_NS} is that $\ell$ on the right hand side is an absolutely free parameter which can be optimized in a variety of ways.  For example,   let us prove the result of  \cite{DE19}: 
\begin{corollary}[Vanishing of energy dissipation] \label{edisscor}
Let $\{u^\nu\}_{\nu>0}$ be a sequence of solutions  to Navier-Stokes. Suppose $\sup_{\nu>0} \|{u^\nu}\|_{L^3_t B^\sigma_{3,\infty}}<\infty$ for some $\sigma\geq \frac13$.  Then
$$
\langle \ve^\nu\rangle  \lesssim \nu^\frac{3\sigma-1}{1+\sigma}.
$$
\end{corollary}
\begin{proof}[Proof of Corollary \ref{edisscor}]
First, recall the elementary result
\begin{lemma}[Commutator Estimates] \label{comlem}
For any function $f: \R^d \to \R$ denote by $\ol{f}_\ell = f* G_\ell$, where $G$ is a Friedrichs' mollifier. Fix $\sigma, \alpha \in (0,1)$ and $p \in [1, \infty]$. There exist implicit constants independent on $\ell$ such that
\begin{align*}
    \|{f-\ol{f}_\ell}\|_{L^p}\leq  &  \ \ell^\sigma \|{ f}\|_{B^\sigma_{p,\infty}},\\
    \|{\nabla^n \ol{f}_\ell}\|_{L^p}\lesssim & \ \ell^{\sigma-n} \|{f}\|_{B^\sigma_{p,\infty}}\qquad\qquad\qquad\qquad n\geq 1,\\
    \|{\nabla^n (\ol{f}_\ell \ol{g}_\ell- \ol{(fg)}_\ell)}\|_{L^{p}}\lesssim& \  \ell^{\sigma+ \alpha-n}\|{ f}\|_{B^\sigma_{rp,\infty}} \|{ g}\|_{B^\alpha_{r'p,\infty}} \qquad\quad  n\geq 0, \ \ \ \,\frac{1}{r}+\frac{1}{r'}=1.
\end{align*}
\end{lemma}
\begin{proof}[Proof of Lemma \ref{comlem}]
These are not difficult to show, see \cite{Eyink} for a clear discussion. 
Let us prove the estimates for $p<\infty$ as $p=\infty$ is simpler. Since $\int_{\mathbb{R}^d}  G_\ell(z)\rmd z =1$, we have 
$
\ol{f}_\ell(x) -f (x)=\int_{\mathbb{R}^d}\left( f(x-z)-f(x)\right)G_\ell (z)\,\rmd z.
$
By Jensen's inequality, we have
\begin{align*}
    \|{f-\ol{f}_\ell}\|_{L^p}^p&\leq \int_{\mathbb{R}^d} \left(\int_{\mathbb{R}^d}  \left| f(x-z)-f(x)\right|^p\,\rmd x \right)G_\ell (z)\,\rmd z
    \leq \ell^{\sigma p} \|f\|^p_{B^\sigma_{p,\infty}}.
\end{align*}
Similarly, since $\int_{\mathbb{R}^d}  \nabla^n G_\ell(z)\rmd x=0$ for all $n\geq 1$, we have
\begin{align*}
\|{\nabla^n \ol{f}_\ell}\|_{L^p}^p
    &\lesssim  \ell^{n(1-p)} \iint_{\mathbb{R}^d\times \mathbb{R}^d }   |f(x-z)-f(x)|^p |\nabla^n G_\ell (z)|\,\rmd x\rmd z\lesssim  \ell^{(\sigma- n)p}\|f\|^p_{B^\sigma_{p,\infty}}.
\end{align*}
Finally, we treat the quadratic commutator. We rewrite it as 
$$
\ol{(fg)}_\ell (x)-\ol{f}_\ell(x) \ol{g}_\ell(x)=\int_{\mathbb{R}^d}  (f(x-z)-\ol{f}_\ell(x))(g(x-z)-g_\ell(x))G_\ell(z)\,\rmd z.
$$
Thus, by an application of H\"{o}lder's inequality, we have
\begin{align*}
   \|\ol{f}_\ell \ol{g}_\ell- &\ol{(fg)}_\ell\|_{L^{p}}^p \leq \int_{\mathbb{R}^d}  \left(\int_{\mathbb{R}^d}  |f(x-z)-\ol{f}_\ell(x)|^p |g(x-z)-\ol{g}_\ell(x)|^p\right)\, G_\ell(z)\,\rmd z\\
    &\leq \int_{\mathbb{R}^d}  \|f(\cdot -z)-\ol{f}_\ell (\cdot)\|_{L^{rp}}^p \|g(\cdot -z)-\ol{g}_\ell (\cdot)\|_{L^{r'p}}^pG_\ell(z)\,\rmd z\\
    &\lesssim  \|f\|^p_{B^\sigma_{rp,\infty}}\|g\|^p_{B^\alpha_{r'p,\infty}}\int_{\mathbb{R}^d}  \left(|z|^\sigma +\ell^\sigma\right)^p \left(|z|^\alpha +\ell^\alpha\right)^p G_\ell(z)\,\rmd z\lesssim \|f\|^p_{B^\sigma_{rp,\infty}}\|g\|^p_{B^\alpha_{r'p,\infty}} \ell^{(\sigma +\alpha)p}
\end{align*}
where we used that  ${\rm supp} (G_\ell) \subset B_\ell(0)$, as well as  the following  estimate
\begin{align*}
    \|f(\cdot -z)-\ol{f}_\ell (\cdot)\|_{L^{p}}&\leq \|f(\cdot -z)-f(\cdot)\|_{L^{p}}+ \|f-\ol{f}_\ell\|_{L^{p}}\leq \left(|z|^\sigma +\ell^\sigma\right) \|f\|_{B^\sigma_{p,\infty}}.
\end{align*}
Estimates on higher derivatives of the quadratic commutator follow similarly.
\end{proof}
Returning now to the proof of Corollary \ref{edisscor}, we apply  Lemma \ref{comlem} to show there exist implicit constants independent of $\ell$ such that
\begin{align}
    \|{E^\ell}\|_{L^\frac{p}{2}_{x,t}}&\lesssim \ell^{2\sigma} \|{u}\|_{L^p_t B^\sigma_{p,\infty}}^2 \label{est_Eell}, \\
     \|{Q^\ell}\|_{L^\frac{p}{3}_{x,t}}&\lesssim \ell^{3\sigma} \|{u}\|_{L^p_t B^\sigma_{p,\infty}}^3\label{est_Qell}, \\
         \|{R^\ell}\|_{L^\frac{p}{2}_{x,t}} + \ell  \|{\div R^\ell}\|_{L^\frac{p}{2}_{x,t}}&\lesssim \ell^{2\sigma} \|{u}\|_{L^p_t B^\sigma_{p,\infty}}^2\label{est_Rell}, \\
         \|{C^\ell}\|_{L^\frac{p}{3}_{x,t}}&\lesssim \ell^{3\sigma-1} \|{u}\|_{L^p_t B^\sigma_{p,\infty}}^3\label{est_Cell}.
\end{align}
 Since $\{u^\nu\}_{\nu>0}$ is bounded in $L^3_t B^\sigma_{3,\infty}$, by \eqref{est_Eell}, \eqref{est_Qell}, \eqref{est_Cell} we estimate from \eqref{D_decomp_NS}
    \begin{align} \nonumber
\langle \ve^\nu\rangle  & \lesssim \|{E^{\ell,\nu}}\|_{L^{\frac{3}{2}}_{x,t}} \left(1+ \|{ \ol{u^\nu}_\ell}\|_{L^3_{t,x}} \right) + \|{Q^{\ell,\nu}}\|_{L^{1}_{x,t}} + \|{C^{\ell,\nu}}\|_{L^{1}_{x,t}}\\ \nonumber
        &\quad + \nu \|{E^{\ell,\nu}}\|_{L^{\frac{3}{2}}_{x,t}} + \nu \left(\|{\nabla \ol{u^\nu}_\ell}\|_{L^2_{x,t}}^2 + \|{u^\nu-\ol{u^\nu}_\ell}\|_{L^2_{x,t}} \left(\|{\Delta \ol{u^\nu}_\ell}\|_{L^2_{x,t}}+ \|{\nabla \ol{u^\nu}_\ell}\|_{L^2_{x,t}}\right)\right)\\\nonumber 
        &\lesssim \ell^{2\sigma} + \ell^{3\sigma}+ \ell^{3\sigma-1}+\nu \ell^{2\sigma} + \nu \ell^{2(\sigma-1)} +\nu \ell^{2\sigma -1 }\lesssim \ell^{3\sigma-1}+ \nu \ell^{2(\sigma-1)}.
    \end{align}
    The proof is concluded by choosing $\ell^{1+\sigma}=\nu$, the (generalized) Kolmogorov scale.
\end{proof}
Meanwhile there is evidence, which should be further investigated, that Corollary \ref{edisscor} is sharp and dissipation vanishes in periodic turbulence \cite{IDES25}.  See Figure \ref{disanom}. 
\begin{figure}[h!]
  \begin{center}
    \includegraphics[width=0.95\textwidth]{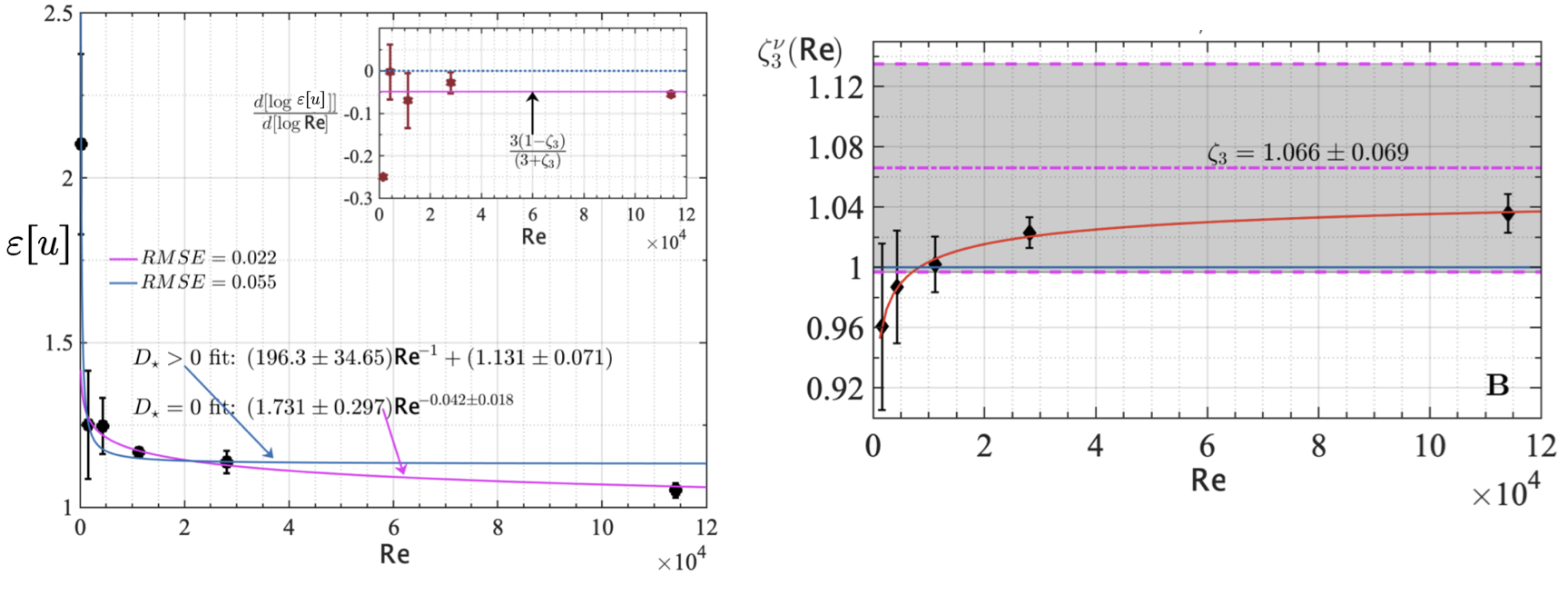}
  \end{center}
  \caption{Numerical evidence for vanishing dissipation and sub-critical Onsager regularity ($\sigma_3>1/3$ of $\zeta_3>1$) from  \cite{IDES25}.  }\label{disanom}
\end{figure}

As one last application of the identity \eqref{D_decomp_NS} for energy, we give a quick estimate of the regularity of the energy profile for weak solutions of Euler originally due to Isett \cite{Isettreg}.  

\begin{corollary}[Kinetic energy regularity]\label{C:en reg}
Let 
$
e(t):=\frac{1}{2}\int_{\mathbb{T}^d} |{u(x,t)}|^2 \,\rmd x
$
denote the kinetic energy.
If $u\in L^\infty_t B^{\sigma_p}_{p,\infty}$ for $p\geq 2$, $\sigma_p\in(0,1)$  is a weak solution of Euler, 
\begin{equation}\label{en_holder}
|{e(t)-e(s)}|\lesssim |{t-s}|^\frac{2\sigma_2}{1+2\sigma_2-3\sigma_3} \lesssim|{t-s}|^\frac{2\sigma_3}{1-\sigma_3}  \qquad \text{for a.e. } t,s.
\end{equation}
If $u\in L^p_t B^{1/3}_{3,\infty}$ for some $p\geq3$, then $e\in W^{1,p/3}(0,T)$.
\end{corollary}
\begin{proof}[Proof of Corollary \ref{C:en reg}]
    Testing \eqref{D_decomp_NS} for $\nu=0$ with $\eta\in C^\infty_t$ yields
    $
    \langle e, \eta' \rangle_{L^2_t}=\left\langle   \ve[u] ,\eta\right\rangle_{L^2_t}   = \left\langle  E^\ell, \eta' \right\rangle_{L^2_{t,x}} -  \left\langle  C^\ell, \eta \right\rangle_{L^2_{t,x}} .
        $
By  \eqref{est_Eell} and \eqref{est_Cell} we get 
    $
 |  \langle e, \eta' \rangle_{L^2}| \lesssim   \|{\eta}'\|_{L^1_t} \ell^{2\sigma_2} + \|{\eta}\|_{L^1_t} \ell^{3\sigma_3-1}.
    $
    For a.e. $t,s$ (say $t>s$), by keeping $\|\eta'\|_{L^1_t}\lesssim 1$, we let $\eta\rightarrow \mathbf{1}_{[s,t]}$ and obtain $|{e(t)-e(s)}|\lesssim \ell^{2\sigma_2} + |{t-s}| \ell^{3\sigma_3-1}.$ The choice $\ell^{1-\sigma_2}=|{t-s}|$ proves  \eqref{en_holder}.  By embedding, $\sigma_3\leq \sigma_2$.  The next claim follows by putting $\eta$ in the appropriate dual space, see \cite{Isettreg}.\footnote{We note that for this final statement, the Besov space used is slightly different from one defined in Lemma \ref{Slem}.  The supremum over increments must be taken before the time integration rather than after.}
\end{proof}

Let us make some remarks. First, the above result gives another proof of Onsager's conjecture by taking $\sigma_3>1/3$ and concluding $e(t)$ must be constant.  
Second, when $u$ is in the natural critical class $u\in L^3_t B^{1/3}_{3,\infty}$, we see that the energy profile is $W^{1,1}(0,T)$.  Being a function of a single variable, it is thus continuous in time. This excludes the possibility that the profile does not have bounded variation on any $I \subset (0,T)$ -- a wild feature that is common among all convex integration constructions, separating them from behavior forced in critical classes.

Finally, one implication of this result, pointed out by Isett, is that the quality of energy dissipation for Euler
flows is stable under perturbation, in the class  $u\in L^\infty_t B^{1/3}_{3,\infty}$.  Namely, in this class $\|e'\|_{L^\infty} \lesssim \|u\|_{L^\infty_t B^{1/3}_{3,\infty}}^3$ and $e'(t)$ varies continuously in $L^\infty$ as $u$ varies in the $L^\infty_t B^\sigma_{3,\infty}$ topology, namely $\|e'_1 - e'_2\|_{L^\infty} \lesssim \|u_1-u_2\|_{L^\infty_t B^{1/3}_{3,\infty}}$. This means that anomalous dissipation $e'(t)\leq -\ve<0$ is an open condition in the space of weak Euler solutions in the critical class endowed with the strong topology. By contrast, Isett conjectured that for weak solutions with regularity below $\sfrac{1}{3}$ of a derivative, or with time integrability below $q=\infty$, the  energy profile should generically fail to be of bounded variation on every time interval \cite{Isettreg,IsettOh}. In particular, a dissipative solution could be perturbed in this class to produce a nearby solution whose energy profile  fails to be monotonic on every time interval.  A version of this conjecture was proved by  De Rosa and Tione \cite{DT22} -- namely the set of weak Euler solutions with energy profile not of bounded variation in any open interval is Baire generic  in an appropriate complete metric space which is a strict subset of all $\alpha$--H\"{o}lder  Euler solutions with $\alpha<1/3$.

This conjecture is clearly very interesting and aims to capture a mechanism, or rather a cause, for regularization in turbulence.  However, it should be mentioned that any weak Euler solution which arises from a strong vanishing viscosity limit, regardless of its regularity, must have a non-increasing energy profile.  Such solutions, of course, should be meager in the space of all weak Euler solutions. Thus, the physical relevance of the stability of anomalous dissipation in any critical space $u\in L^p_t B^{1/3}_{p,\infty}$ for $p\in[3,\infty)$ is not completely clear.  We will issue Conjecture \ref{conj45} below which is similar in spirit to the above, but based on the $\sfrac{4}{5}$ law.

\section{Kolmogorov's 1941 theory}\label{K41sec}

\begin{quotation}
\vspace{2mm}

\emph{"I had heard several times Kolmogorov talking about turbulence and had always been given the impression that these were talks by a pure physicist. One could easily forget that Kolmogorov was a great mathematician. He could discuss concrete equations of state of real gases and liquids, the latest data of experiments, etc. When Kolmogorov was close to eighty I asked him about the history of his discoveries of the scaling laws. He gave me a very astonishing answer by saying that for half a year he studied the results of concrete measurements. In the late Sixties Kolmogorov undertook a trip on board a scientific ship participating in the experiments on oceanic turbulence. Kolmogorov was never seriously interested in the problem of existence and uniqueness of solutions of the Navier-Stokes system. He also considered his theory of turbulence as purely phenomenological and never believed that it would eventually have a mathematical framework."}\\
\phantom{adsf} \hfill  --  Y. Sinai \cite{VL}
\end{quotation}
\vspace{2mm}

\begin{figure}[h!]
  \begin{center}
    \includegraphics[width=0.41\textwidth]{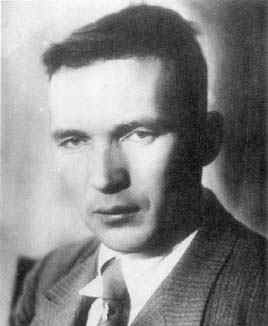}
    \hspace{5mm}
        \includegraphics[width=0.42\textwidth]{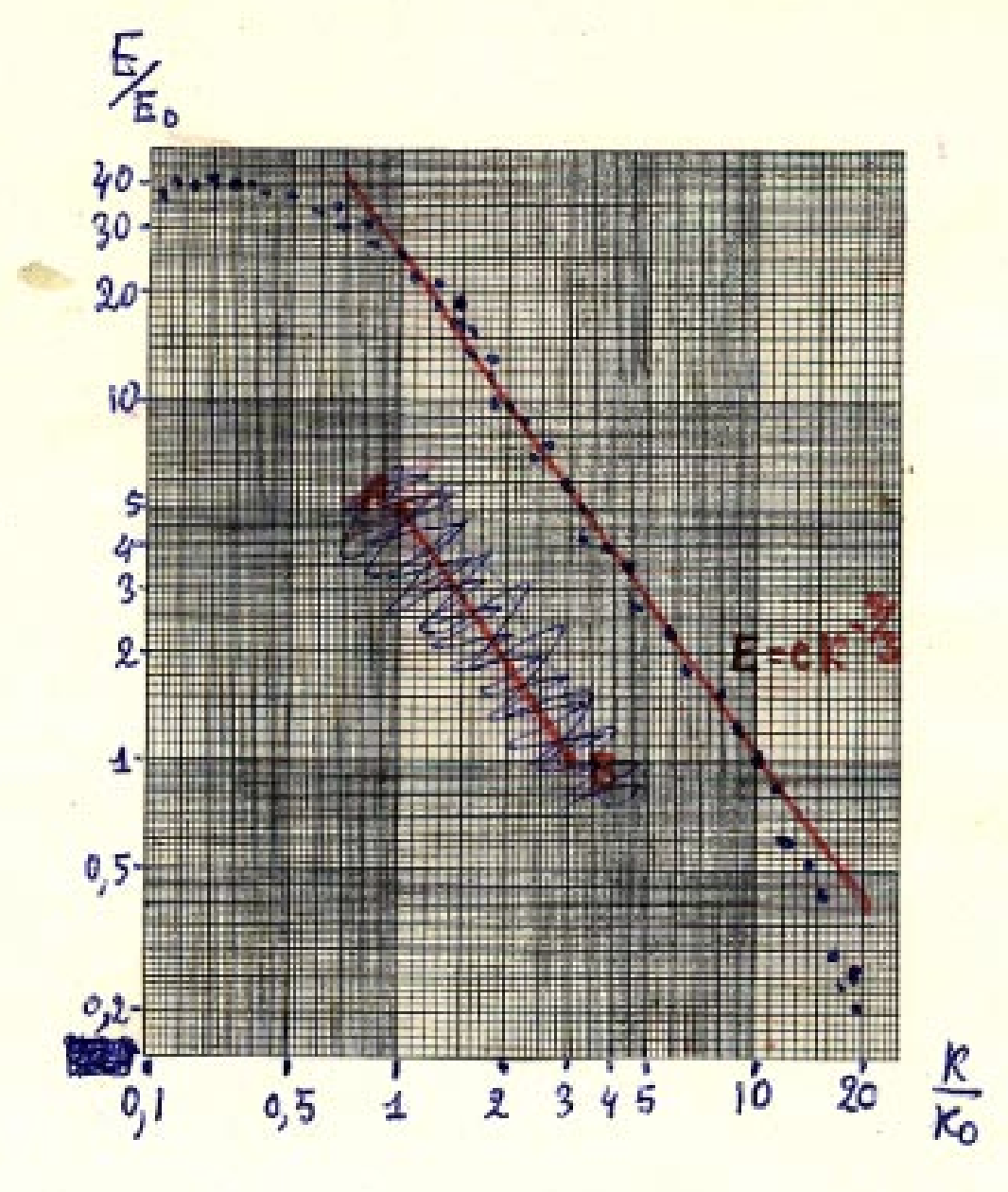}
  \end{center}
  \caption{Andrey N. Kolmogorov and his hand-drawn $\sfrac{5}{3}$--law \cite{Shiryaev}.}\label{Kfig}
\end{figure}

Kolmogorov, in his celebrated 1941 works, used the balance associated to the quantity \eqref{Dell45}, to argue that $\ve_\ell^\|[u]$, which represents a flux of energy through scale $\ell$,  should be effectively constant over a long range of scales $\ell_\nu \ll \ell \ll \ell_I$ called the inertial range.  The "dissipative" scale $\ell_\nu$ is where viscosity eventually removes the energy which cascades down, and the "integral" scale $\ell_I$ is the scale of the domain, or where the force is applied. In this scale range, Kolmogorov argued the flux would  balance  the anomalous rate of energy dissipation $\langle \ve_\ell^\|[u^\nu] \rangle \approx \langle \ve\rangle$, as in \eqref{fluxbalvisc}.  From \eqref{Dell45}, this results in his $\sfrac{4}{5}$ law
\be\label{k45law}
S_{3,\|}^{u^\nu}(\ell)\sim- \tfrac{4}{5} \langle \ve \rangle \ell, \qquad \ell_\nu \ll \ell \ll \ell_I.
\ee
Under the assumptions of \emph{statistical homogeneity, isotropy and no-intermittency}, Kolmogorov used the above  "law" to make predictions for any $p$th order structure functions $S_{p,\|}^{u^\nu}(\ell)$ in the inertial range.  This is often presented on strictly dimensional analysis grounds, assuming that $\nu, \langle \ve \rangle$ (and not $\ell_I$)  are the only relevant parameters. Here I present a slightly different "derivation" of the same, which highlights the role of the  assumptions concerning the dissipation rate, which we will reexamine in the next section.
\begin{align} \nonumber
S_{p,\|}^{u^\nu}(\ell)&:= \fint_0^T \fint_M\fint_{\mathbb{S}^{d-1}} (\hat{z}\cdot\delta_{\ell \hat{z}}u^\nu )^p \ \rmd \sigma(\hat{z}) \rmd x \rmd t\qquad\qquad\qquad \qquad\! \! \! \! \! \! \! \text{by definition \eqref{longsp}}\\ \nonumber
&\approx \fint_0^T \fint_M \left(\fint_{\mathbb{S}^{d-1}} (\hat{z}\cdot\delta_{\ell \hat{z}}u^\nu )^3 \ \rmd \sigma(\hat{z}) \right)^{p/3}\rmd x \rmd t\qquad \qquad\     \text{by isotropy}\\ \nonumber
&\approx \left(\tfrac{12}{d(d+2)}\right)^{p/3} \fint_0^T \fint_M \ ( \ve[u] \ell)^{p/3}  \rmd x \rmd t  \quad \ \quad \quad \quad \quad \quad \ \  \ \ \text{by identity \eqref{Dell45}/\eqref{k45law}} \\ \nonumber
&\approx C_p\ell^{p/3}\left(\fint_0^T \fint_M\ve[u]  \rmd x \rmd t\right)^{p/3} \qquad \qquad \quad \quad \quad \quad \quad\ \!  \ \text{by no "intermittency"} \\
 &\sim C_p (\langle \ve \rangle \ell)^{p/3} \qquad  \qquad \qquad\qquad\qquad\qquad\qquad\qquad \qquad\text{}\nonumber
\end{align}
with the approximations holding in the range $\ell_\nu \ll \ell \ll \ell_I$.
That is $\zeta_p^\| = \sfrac{p}{3}$ in \eqref{Spbeh}. This prediction is often referred to as the K41 theory of turbulence: 
\be
  \label{k41} 
S_{p,\|}^{u^\nu}(\ell)  \sim C_p (\langle \ve \rangle \ell)^{p/3} \qquad \text{for} \qquad \ell_\nu \ll \ell \ll \ell_I.
\ee
As one can see, the various assumptions such as no-intermittency are used to justify commuting averages with non-linearities (powers).  Such approximations would hold well if, for example, $p\approx 3$ or if the dissipation measure $\ve[u]$ does not vary too widely from point to point in space time.  However, if $\ve[u]$ charges some lower dimensional set, then the error made in commuting these operations can be great.  We shall return to this point in the next section.

For now, one famous consequence for $p=2$ --  the $\sfrac{2}{3}$ law --  makes also a prediction of a $-\sfrac{5}{3}$ decay of the time averaged $\langle \cdot \rangle$ Fourier energy spectrum:
\be\label{k41spec}
S_{2,\|}^{u^\nu}(\ell) \sim \ell^{2/3} \qquad \iff \qquad E(k) := \sum_{|k'|=k}\langle |\hat{u}(k')|^2\rangle \sim  k^{-5/3}.
\ee
Incidentally, Kolmogorov was the "somebody" to have arrived first at Onsager's formula for the correlation function in isotropic turbulence.
See Figure \ref{Kfig} for Kolmogorov's hand drawn prediction over experimental data!  Figure \ref{Kfig2} shows the strength of this prediction. 
The equivalence \eqref{k41spec} can be proved as follows
\begin{lemma}\label{eklem}
Let $u\in L_t^2 L_x^2$. Then $S_{2,\|}(\ell)\lesssim \ell^{s-}$ if $E(k)\leq k^{-(1+s)}$.
\end{lemma}
\begin{proof}[Proof of Lemma \ref{eklem}]
The result follows from the elementary identity 
\begin{align*}
\fint_0^T\|u(t)\|_{H^{s-}}^2\rmd t &:= \sum_{k'\in \mathbb{Z}^d} |k'|^{2s-}\fint_0^T |\hat{u}(k',t)|^2\rmd t =  \sum_{k\in \mathbb{N}} k^{2s-}  \sum_{|k'|=k}\fint_0^T |\hat{u}(k',t)|^2\rmd t \\
&=: \sum_{k\in \mathbb{N}} k^{2s-}E(k)\leq  \sum_{k\in \mathbb{N}} k^{-1-}<\infty.
\end{align*}
Assuming $E(k) \lesssim k^{-(1+s)}$, we conclude from the above that $S_2(\ell) \leq \ell^{s-}$.
\end{proof}
An equivalence holds for near-power laws (as in \eqref{eklem}) as a consequence of the Wiener-Khinchin theorem \cite{V19}.

\begin{figure}[h!]
  \begin{center}
    \includegraphics[width=0.47\textwidth]{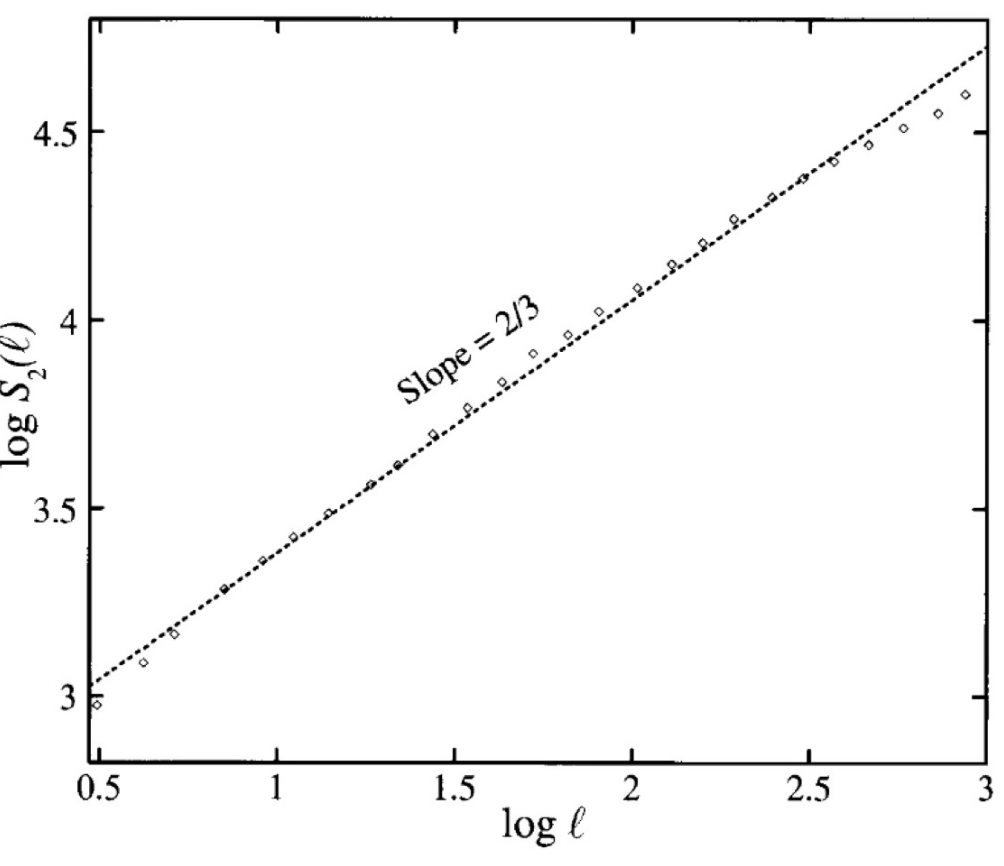}
       \hspace{2mm} 
       {
        \includegraphics[width=0.42\textwidth]{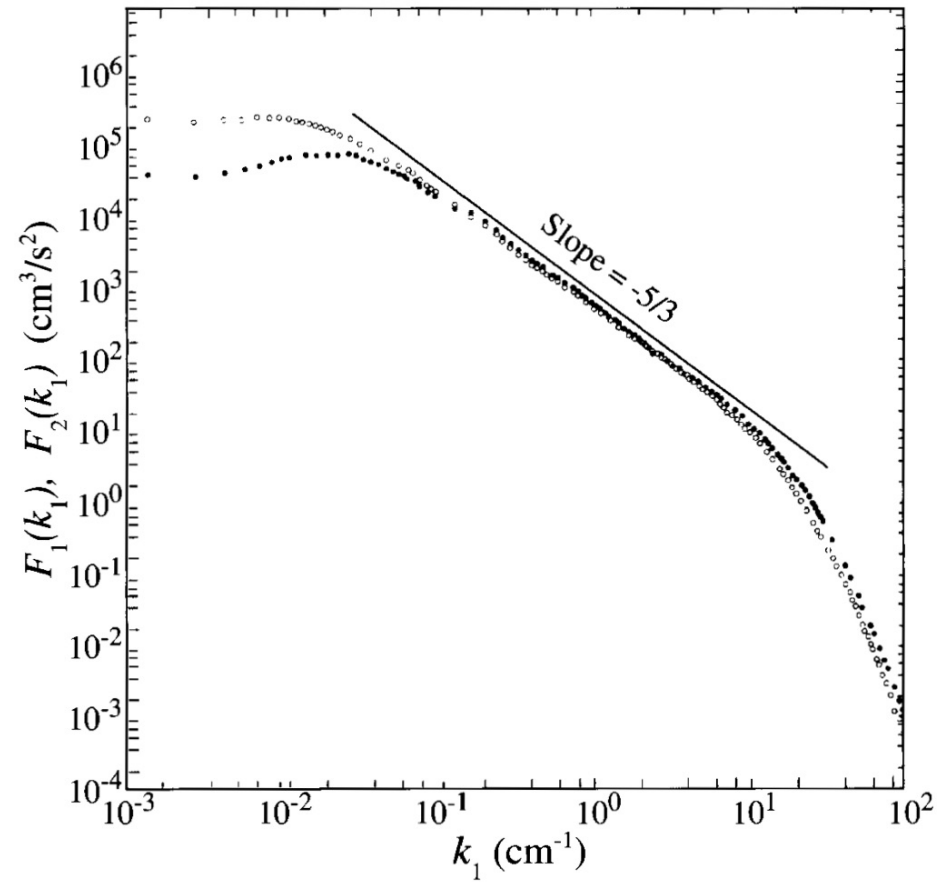}}
  \end{center}
  \caption{$S_2(\ell)$ and $E(k)$ plotted from data. \cite{C78, F95}  }\label{Kfig2}
\end{figure}

Kolmogorov's law \eqref{k45law} can, in fact, be made quite rigorous.  Following Duchon-Robert \cite{DR00}, this was done by Eyink \cite{E03} and  Novack \cite{N24} who show Kolmogorov's law holds in a space-time local and deterministic sense, at least asymptotically as $\ell \to 0$.  They prove
\begin{proposition}[Kolmogorov formulae for dissipation]\label{invthmcor2}
Let $u\in L^3(0,T; L^3(\mathbb{T}^d))$ be a locally dissipative weak solution of incompressible Euler as in Theorem \ref{invthm}. Then, the dissipation measure $\ve[u]$  in the distributional equality
\be
 \partial_t \left(\tfrac{1}{2} |u|^2 \right) + {\rm div} \left((\tfrac{1}{2} |u|^2+p) u \right) = -\ve[u] + u \cdot f
\ee
admits the following equivalent representations
\begin{align}\label{43ident}
\varepsilon[u]&= -\lim_{\ell \to 0} \frac{d}{4 } \langle (\hat{z} \cdot \delta_{\ell \hat{z}}u^\nu) |\delta_{\ell \hat{z}}u^\nu|^2 \rangle_{\rm ang}\\
&= -\lim_{\ell \to 0} \frac{d(d+2)}{12 } \langle (\hat{z} \cdot \delta_{\ell \hat{z}}u^\nu)^3\rangle_{\rm ang} \label{45ident}
\end{align}
in the sense of distributions, where
$\langle \cdot \rangle_{\rm ang}= \fint_{\mathbb{S}^{d-1}} \ \cdot \ \rmd\sigma(\hat{z})$
 is the angle average with respect to the usual spherical surface measure.
\end{proposition}
The identity \eqref{43ident} is Kolmogorov's  $\sfrac{4}{3}$ law, and \eqref{45ident} is the   $\sfrac{4}{5}$ law if these are interpreted as holding for asymptotically small scales.    It should be remarked that they are purely deterministic statements, without need to appeal to any notion of statistical ensembles (commonly thought necessary).  Moreover, they are local in space and time statements, illustrating the power of Kolmogorov's prediction.

We do not prove Proposition \ref{invthmcor2} here, as it unfortunately involves some cumbersome manipulations very similar to those which give rise to \eqref{moll en identity NS}.  Instead, we use Proposition \ref{invthm} to establish a stronger result -- the Kolmogorov  $\sfrac{4}{3}$ law in a \emph{finite}  range of scales
\begin{align}\nonumber
     S_{3,\rm mix}^{u^\nu}(\ell) &:=\fint_0^T\fint_M \langle (\hat{z} \cdot \delta_{\ell \hat{z}}u^\nu) |\delta_{\ell \hat{z}}u^\nu|^2 \rangle_{\rm ang} \rmd x\rmd t \sim - \frac{4}{3} \langle \ve \rangle \ell, \qquad \ell_\nu \ll \ell \ll \ell_I.
\end{align}
This identity follows by analyzing the balance \eqref{moll en identity NS}.  As such, one can understand the result a defining the inertial range of scales to be the one in which $     S_{3,\rm mix}^{u^\nu}(\ell)$ or $     S_{3,\|}^{u^\nu}(\ell)$ exhibit power law scalings. The size of these scales $ \ell_\nu $ and $ \ell_I$ can be mathematically controlled, under very mild assumptions on the velocity field. Indeed, we prove
\begin{theorem}[Kolmogorov's $\sfrac{4}{3}$ and $\sfrac{4}{5}$ laws]\label{45thlaw}
 Let $u^\nu$ be a smooth solution to \eqref{ns1}--\eqref{ns2} for $\nu> 0$ with forcing $f\in L^\infty_tC^2_{x}$.  Suppose, that $\{u^\nu\}_{\nu>0}$ has a (uniform in $\nu$) $L^3$ modulus $\phi_u(\ell):=\sup_{|r|\leq \ell}\sup_{\nu>0}\|u^\nu(\cdot +r,\cdot)-u^\nu(\cdot, \cdot)\|_{L^3_{t,x}}$. 
    If $u^\nu\rightarrow u$ in $L^3_{t,x}$, there exists a monotone decreasing sequence $\ell_\nu\rightarrow 0$ such that 
\begin{align}\label{finite43}
\limsup_{\nu\rightarrow 0}\sup_{\ell\in [\ell_\nu,\ell_I]} \left|  \left(  \frac{\langle (\hat{z} \cdot \delta_{\ell \hat{z}}u^\nu) |\delta_{\ell \hat{z}}u^\nu|^2 \rangle_{\rm ang}}{\ell}+\frac{4}{d}\varepsilon^\nu[u^\nu] ,\psi\right)_{L^2_{t,x}}  \right| &\lesssim \phi_u(\ell_I) \\
\limsup_{\nu\rightarrow 0}\sup_{\ell\in [\ell_\nu,\ell_I]} \left|  \left(  \frac{\langle (\hat{z} \cdot \delta_{\ell \hat{z}}u^\nu)^3 \rangle_{\rm ang}}{\ell}+\frac{12}{d(d+2)}\varepsilon^\nu[u^\nu] ,\psi\right)_{L^2_{t,x}}  \right| &\lesssim \phi_u(\ell_I) 
\end{align}
 for all space-time scalar functions $\psi \in C^\infty_0$, where   $\varepsilon^\nu[u^\nu] := \nu |\nabla u^\nu|^2$.  For space-time averages, introducing    $\langle \ve^\nu \rangle :=\fint_0^T \fint_{\mathbb{T}^d} \varepsilon^\nu[u^\nu]  \rmd x\rmd t$, we have
\be
\sup_{\ell\in [\bar{\ell}_\nu, \bar{\ell}_I]}  \left| \frac{ S_{3,\rm mix}^{u^\nu}(\ell) }{\frac{4}{d} \langle \ve^\nu \rangle \ell} +1\right| \ll 1, \qquad \sup_{\ell\in [\bar{\ell}_\nu, \bar{\ell}_I]}  \left| \frac{ S_{3,\|}^{u^\nu}(\ell) }{\frac{12}{d(d+2)} \langle \ve^\nu \rangle \ell} +1\right| \ll 1
\ee
where $\bar{\ell}_\nu:=  \left({\nu}/{\langle \ve^\nu \rangle}\right)^{{\frac{1}{2}-}}$  and $\bar{\ell}_I :=   \langle \ve^\nu \rangle^{\frac{1}{2}-}$, provided $T\gg \langle \ve^\nu \rangle^{-1}$.
\end{theorem}
\begin{figure}[h!]
  \begin{center}
    \includegraphics[width=\textwidth]{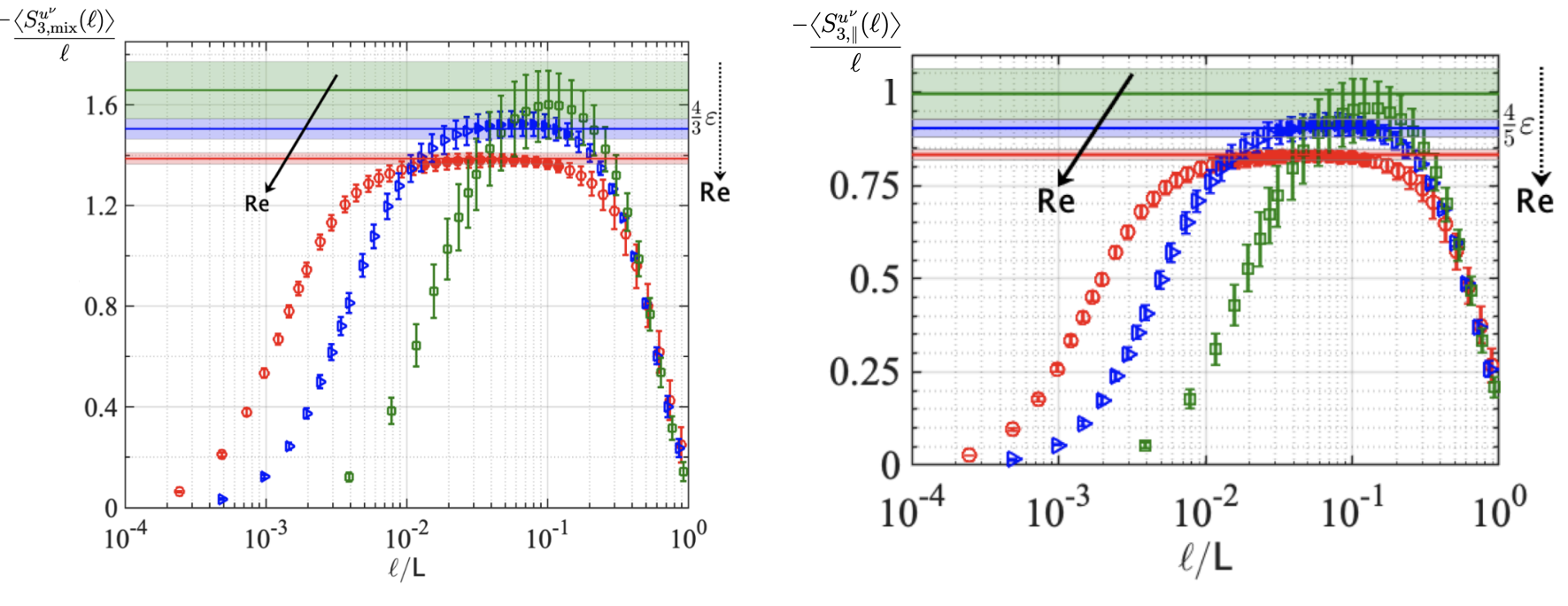}
  \end{center}
  \caption{ $\sfrac{4}{3}$  and  $\sfrac{4}{5}$ laws from direct numerical simulation  \cite{IDES25}}  \label{fig45}
\end{figure}
See Figure \ref{fig45} for a numerical demonstration of the $\sfrac{4}{3}$ and $\sfrac{4}{5}$ laws over finite ranges. A stochastic version of this theorem was proved in \cite{BCZPSW,BCZPSW2}.
 We remark that the information on the scales above can be sharpened with further regularity assumptions on the solution.  This is done in the proof below.
\begin{proof}[Proof of Theorem \ref{45thlaw}]
    Since the choice of the mollifying kernel $G$ in \eqref{moll en identity NS} is arbitrary, we let 
    $
    G\rightarrow \frac{1}{{\rm Vol}(B_1)} \mathbf{1}_{B_1}.
    $
 Since $\nabla G\overset{*}{\rightharpoonup} -  \frac{1}{{\rm Vol} (B_1)} \hat{z} d\mathcal \sigma(\hat{z}) $ in the sense of measures,     for all $x,t$
    \begin{align*}
    D^\nu_\ell(x,t)\rightarrow &-\frac{1}{4  \mathcal {\rm Vol}(B_1) }\int_{\mathbb S^{d-1}} \hat{z} \cdot \frac{\delta_{\ell  \hat{z} }u^\nu(x,t)}{\ell}|\delta_{\ell  \hat{z} }u^\nu(x,t)|^2 \, \rmd \sigma(\hat{z}) \\
    & =-\frac{d}{4}\fint_{\mathbb S^{d-1}}  \hat{z}  \cdot \frac{\delta_{\ell  \hat{z} }u^\nu(x,t)}{\ell}|\delta_{\ell  \hat{z} }u^\nu(x,t)|^2  \rmd \sigma(\hat{z})=-\frac{d}{4}\frac{\langle (\hat{z} \cdot \delta_{\ell \hat{z}}u^\nu) |\delta_{\ell \hat{z}}u^\nu|^2 \rangle_{\rm ang}}{\ell}.
    \end{align*}
Moreover,  $u^\nu$ is assumed smooth enough for the  Navier-Stokes energy balance to hold:
    $$
  \nu |\nabla u^\nu|^2=-  \partial_t \frac{|u^\nu|^2}{2} -\div  \left(\left(\frac{|u^\nu|^2}{2}+p^\nu\right)u^\nu\right)+ \nu \Delta \frac{|u^\nu|^2}{2} + f\cdot u^\nu
    $$
      in the sense of distributions.  Thus, we have 
    \begin{align} \nonumber
        \frac{d}{4}\frac{\langle (\hat{z} \cdot \delta_{\ell \hat{z}}u^\nu) |\delta_{\ell \hat{z}}u^\nu|^2 \rangle_{\rm ang} }{\ell} +\nu |\nabla u^\nu|^2 &=\partial_t \left(\frac{u^\nu\cdot (u^\nu_\ell - u^\nu)}{2}\right) +\div \left(\frac{u^\nu\cdot(u^\nu_\ell - u^\nu) }{2}u^\nu \right)\\ \nonumber
        &\quad\quad +\div \left(\frac{p^\nu (u^\nu_\ell - u^\nu) +(p^\nu_\ell-p^\nu) u^\nu}{2} \right)\\ \nonumber
        &\quad \quad \quad +\nu \Delta \left(\frac{|u^\nu|^2 -u^\nu\cdot u^\nu_\ell}{2}\right) +\nu\nabla u^\nu :\nabla u^\nu_\ell\\
        &\qquad \quad \quad - \tfrac{1}{2} u^\nu\cdot (\ol{f}_\ell -  f) - \tfrac{1}{2} (\ol{u}_\ell^\nu-u^\nu)\cdot {f},  \label{balanceS}
    \end{align}
    in the sense of distributions. Since $u^\nu$ is compact in $L^3_{t,x}$, by the Fr\`echet--Kolmogorov compactness criterion, the quantity 
    $
\phi_u(\ell):=\sup_{|y|\leq \ell}\sup_{\nu>0}\|u^\nu(\cdot +y,\cdot)-u^\nu(\cdot, \cdot)\|_{L^3_{t,x}}
    $
    defines a modulus of continuity, i.e. $\lim_{\ell\rightarrow 0}\phi_u(\ell)=0$. Moreover, a direct consequence of the Calderón–Zygmund estimate in Lemma \ref{P: pressure CZ} is
    $$
    \sup_{|y|\leq \ell}\sup_{\nu>0}\|p^\nu(\cdot +y,\cdot)-p^\nu(\cdot, \cdot)\|_{L^{\frac32}_{x,t}}\leq C \sup_{|y|\leq \ell}\sup_{\nu>0}\|u^\nu(\cdot +y,\cdot)-u^\nu(\cdot, \cdot)\|_{L^3_{t,x}}\lesssim \phi_u(\ell).
    $$
    The forcing terms are simple
        $
    \sup_{|y|\leq \ell}\sup_{\nu>0}| u^\nu\cdot (\ol{f}_\ell -  f) +  (\ol{u}_\ell^\nu-u^\nu)\cdot {f}  | \lesssim \ell  + \phi_u(\ell) \lesssim \phi_u(\ell).
    $
    We can thus estimate 
    \begin{equation*}
    \left|  \left( \frac{d}{4}\frac{\langle (\hat{z} \cdot \delta_{\ell \hat{z}}u^\nu) |\delta_{\ell \hat{z}}u^\nu|^2 \rangle_{\rm ang} }{\ell} +\varepsilon^\nu[u^\nu] , \psi \right)_{L^2_{x,t}} \right| \leq C \left( \phi_u(\ell)  +\nu \|\nabla u^\nu\|_{L^2_{x,t}}\|\nabla u^\nu_\ell\|_{L^2_{x,t}}\right).
    \end{equation*}
    Since $u^\nu$ is a sequence of Navier-Stokes solutions with $L^2_x$ bounded initial data we have  
    $
    \|\nabla u^\nu\|_{L^2_{x,t}}=  \left(\frac{\langle \ve^\nu \rangle}{\nu}\right)^{\frac12}.
    $
Together with $\|\nabla u^\nu_\ell\|_{L^2_{x,t}}\leq C  \phi_u(\ell)  \ell^{-1}$, this yields 
$$
     \left|  \left( \frac{\langle (\hat{z} \cdot \delta_{\ell \hat{z}}u^\nu) |\delta_{\ell \hat{z}}u^\nu|^2 \rangle_{\rm ang} }{ \frac{4}{d}   \langle \ve^\nu \rangle \ell} +\frac{\varepsilon^\nu[u^\nu] }{\langle \ve^\nu \rangle}, \psi \right)_{L^2_{x,t}} \right|  \lesssim   \frac{\phi_u(\ell)}{ \langle \ve^\nu \rangle} \left(1  +\ell^{-1} \left(\nu {\langle \ve^\nu \rangle}\right)^{\frac12}\right).
$$
Whence \eqref{finite43} holds if $\ell_\nu \geq   \left(\nu {\langle \ve^\nu \rangle}\right)^{\frac12}$.  Provided $\ell_I\ll  \phi_u^{-1}( \langle \ve^\nu \rangle)$ we get 
$
\text{(left hand side)} \ll   1.
$
This holds even if $\langle \ve^\nu \rangle\to 0$ as $\nu\to 0$, but it means the scale range shrinks.

 For the  global balance, there are improvements that one can have if one considers only space averages.  Specifically, integrating \eqref{balanceS} yields
\be
\left| \frac{ S_{3,\rm mix}^{u^\nu}(\ell) }{ \frac{4}{d}   \langle \ve^\nu \rangle \ell} +1\right| \lesssim \frac{S_2^{u^\nu}(\ell)}{ \langle \ve^\nu \rangle T}  +  \frac{\ell^2}{ \langle \ve^\nu \rangle } +  \frac{S_2^{u^\nu}(\ell) }{\ell} \left(\nu {\langle \ve^\nu \rangle}\right)^{\frac12},
\ee
since $ \|f-\ol{f}_\ell \|_{L^\infty}\lesssim \|f\|_{C^2}  \ell^2$ for even mollification kernels, see \cite[Lemma 14.1]{Isetthesis}. Keep in mind that everything has been non-dimensionalized by the $L^2$ norm $\mathsf{U} = \|u^\nu\|_{L^2}$, so the constants are indeed viscosity independent. Moreover, in the space-integrated identity, the $S_2^{u^\nu}(\ell)$  is  a $L_t^\infty L^2_x$ modulus, not $L_{t,x}^3$.    Suppose that $S_2(\ell)^{1/2} \leq  \ell^{\sigma_2}$, as observed from data (Figure \ref{S2}). 
Then, taking $\bar{\ell}_\nu:=  \left({\nu}{\langle \ve^\nu \rangle}\right)^{{\frac{1}{2(1-\sigma_2)}-}}$,  $\bar{\ell}_I \ll \sqrt{  \langle \ve^\nu \rangle}$ and $T\gg \langle \ve^\nu \rangle^{\frac{\sigma_2}{2}-1}$, we have
$
\sup_{\ell\in [\bar{\ell}_\nu, \bar{\ell}_I]}  \left| \frac{ S_{3,\rm mix}^{u^\nu}(\ell) }{\frac{4}{d} \langle \ve^\nu \rangle \ell} +1\right| \ll 1.
$
 We treated only the $\sfrac{4}{3}$ law here. The $\sfrac{4}{5}$ law version is slightly more involved computationally, but it follows similarly.  See, e.g.  Remark \ref{45thlawrem}.
\end{proof}

\begin{remark}[Kolmogorov laws without the zeroth law]
It is worth considering the case in which the dissipation vanishes, perhaps  slowly, as a power of inverse Reynolds:
\be\label{vanish}
\langle \ve^\nu \rangle\sim \nu^\alpha, \qquad \alpha>0.
\ee
For example, in \cite{DE19} it was proved that if $u^\nu\in L^3(0,T; B_{3,\infty}^{\sigma_3}(\mathbb{T}^d))$ uniformly, then $(\varepsilon^\nu[u^\nu] , \psi) \leq \nu^{\frac{3\sigma_3-1}{\sigma_3+1}}$ (Corollary \ref{edisscor}).  When $\sigma_3>1/3$, this gives a bound on the rate $\alpha$. Our theorem  shows a long scale range where the  $\sfrac{4}{3}$ law $S_{3,\rm mix}^{u^\nu}(\ell) \approx  -\frac{4}{d} \varepsilon^\nu[u^\nu] \ell$ (similarly the  $\sfrac{4}{5}$ law $S_{3,\|}^{u^\nu}(\ell) \approx  -\frac{12 }{d(d+2)}  \varepsilon^\nu[u^\nu] \ell$)  is observed despite formally there is not existing ``anomalous dissipation" in the limit $\nu \to 0$.   	This is our interpretation of recent numerical data \cite{IDES25}, from which Figure \ref{fig45} is drawn. The relevant  \emph{inertial range} scales are
$$
 \bar{\ell}_\nu\sim   {\nu}^{{\frac{1+\alpha}{2(1-\sigma_2)}-}}, \qquad \bar{\ell}_I \sim \nu^{\alpha/2}, \qquad T\gg \nu^{\alpha(\frac{\sigma_2}{2}-1)}.
$$
Indeed, by eye one can see a trend in the figure for the lower limit of the scaling range to decrease faster than the upper (which does look to be slowly decreasing).  
\end{remark}

With Theorem \ref{45thlaw} in sight, it is tempting to make the following conjecture
\begin{conjecture}\label{conj45}
For generic initial conditions $u_0\in L^2$, inviscid limits of sequences of Leray-Hopf weak solutions of the Navier-Stokes remain bounded uniformly in $L_t^3 B_{3,\infty}^{1/3}$. 
\end{conjecture}
The intuition behind this conjecture is the following. Assume anomalous dissipation occurs $\langle \varepsilon \rangle >0$.  Then, Theorem \ref{45thlaw} says that $\langle S_{3,\|}^{u^\nu}(\ell)\rangle /\ell$ is finite  over a finite (inertial) scale range.  This range extends to zero length as $\nu\to 0$, and $\langle S_{3,\|}^{u^\nu}(\ell)\rangle /\ell$ is non-vanishing as $\ell \to 0$ . Now, for any vector field in $L_t^3 B_{3,\infty}^{s}$ with $s\geq 1/3$,  the $\langle S_{3,\rm mix}^{u^\nu}(\ell)\rangle /\ell $ is bounded uniform in $\ell$ but if $s>1/3$ then the ratio clearly vanishes as $\ell\to 0$.  Therefore we should have $s\leq 1/3$.  But    "generic" vector fields with $s<1/3$ \emph{should} have $|\langle S_{3,\|}^{u}(\ell)\rangle /\ell \to \infty|$ as $\ell \to 0$, since $\langle S_{3,\|}^{u}(\ell)\rangle /\ell\sim (\delta_\ell u)^3/\ell\sim \ell^{3s-1}$. Turbulence should place the fluid velocity in such a generic category, why not?  It thus seems that $L_t^3 B_{3,\infty}^{1/3}$  is the only reasonable space, at least in the Besov scale. See also \cite{D22}.

Indeed, this conjecture ends up being true for some model problems.  For the Burgers equation it is known by works \cite{JKM,GP,TT}.  We offer a proof here in Theorem \ref{thm}, following \cite{GJO15}. For the Kraichnan model of scalar turbulence (scalar is transported by a Gaussian random velocity field with H\"older regularity $\alpha$), it has been proved in \cite{GGM24,DGP25} that the scalar field retains $1-\alpha$ derivatives in $L^2$, at least after averaging. Both these results follow from the same principle: the flux of energy is limited by the input of energy (initial conditions), meanwhile it takes a coercive or nearly coercive form and this limits the possible (ir)regularity of the solution uniform in viscosity.

Aside from the clearly shaky reasoning behind Conjecture \ref{conj45}, the recent numerical results obtained in \cite{IDES25} call into question some components of the intuition, namely  the answer to Question \ref{consques} being negative.  In particular, \cite{IDES25}  gives evidence that  the dissipation slowly vanishing in the inviscid limit.  In this case, the information in the $\sfrac{4}{5}$ law degenerates, and does not shed any particular light on the regularity of the limiting object.  In fact, those numerics indicate the $L^3$ Besov regularity tends towards a value larger that $\sfrac{1}{3}$ asymptotically which, via Corollary \ref{edisscor}, is consistent with the perceived rate of decay of energy dissipation. See Figure \ref{disanom}. This may be taken as an (admittedly very weak) point in favor of the answer to Question \ref{consques} being positive, since why would the turbulent solution land on any particular $\sigma_3>1/3$?

\section{Landau's objection and intermittency}\label{landausec}
\vspace{2mm}

\begin{quotation}
\emph{"...the energy associated with large wave-numbers is very unevenly distributed in space. There appear to be isolated regions in which the large wave-numbers are ‘activated’, separated by regions of comparative quiescence."} \hfill Batchelor and Townsend \cite{BT}\\
\end{quotation}
\vspace{2mm}

Let us return to Kolmogorov's prediction \eqref{k41} concerning the scaling laws for 
\be\label{Sppred}
S_p^\| (\ell) \sim C_p  \ell^{\zeta_p}, \qquad \ell_\nu \ll \ell \ll \ell_I.
\ee
Kolmogorov predicted $\zeta_p=p/3$, which agrees with the  "rigorous" prediction from the $\sfrac{4}{5}$ law (see Theorem \ref{45thlaw}) at order $p=3$.
How does it hold up?  One can see in Figure \ref{figexp} a plot of the exponents $\zeta_p$ vs $p$.  Nearby $p=3$, Kolmogorov's prediction does remarkably well (in particular, at order $p=2$ it is only very slightly an underestimate, see also Figure \ref{Kfig2}).  However, for higher moments $p \gg 3$ the data clearly does not agree with the linear-in-$p$ prediction. Instead, $\zeta_p$ appears to be some non-linear convex function of $p$.  Such behavior is termed \emph{intermittency} since it necessitates non-uniform irregularities in space and time.  The flow field is spotty, with some regions being markedly smoother than others, beautifully captured by the quote of Batchelor and Townsend.  All this happening, in principle, on fractal sets of lower dimension.

Predicting $\zeta_p$ has been one of the major goals of physicists working on turbulence \cite{SL,K62,FSN,Y01}.  Constructing solutions which display intermittency \cite{NV23,GKN}, as well imposing fundamental constraints the exponents $\zeta_p$, is a goal of mathematicians. We will focus on the issue of imposing constraints on $\zeta_p$ here. 
\begin{figure}[h!]
  \begin{center}
    \includegraphics[width=0.4\textwidth]{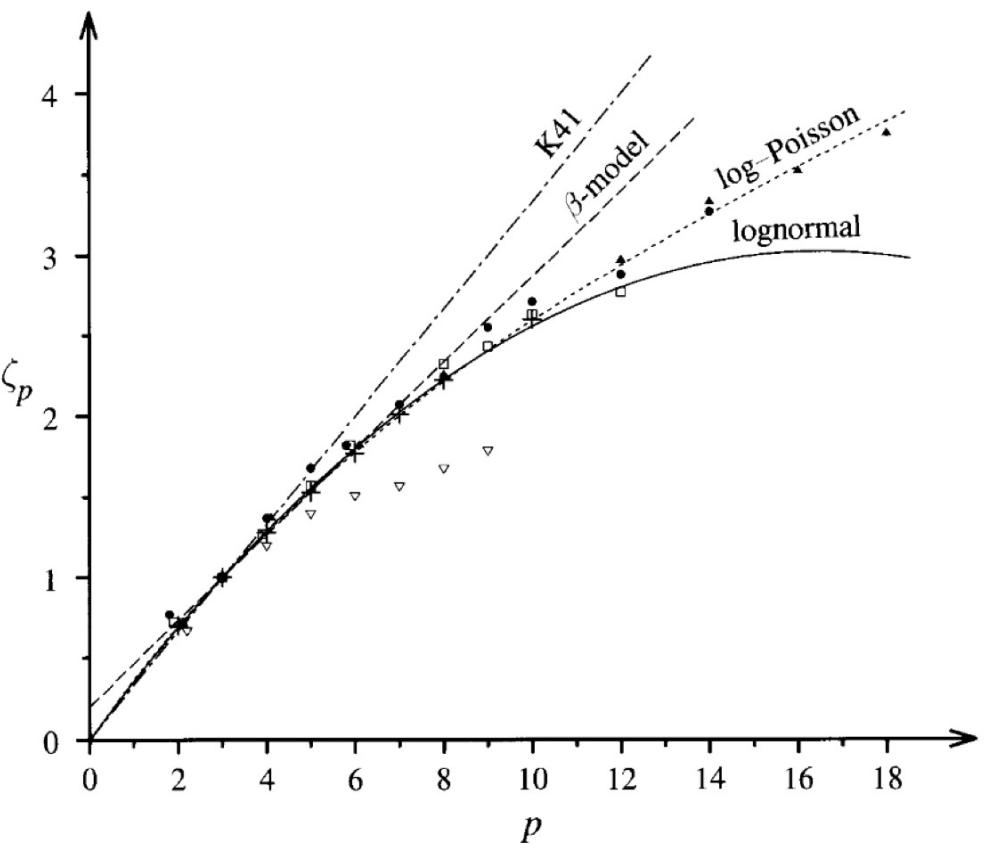}
       \hspace{2mm} 
       {
        \includegraphics[width=0.4\textwidth]{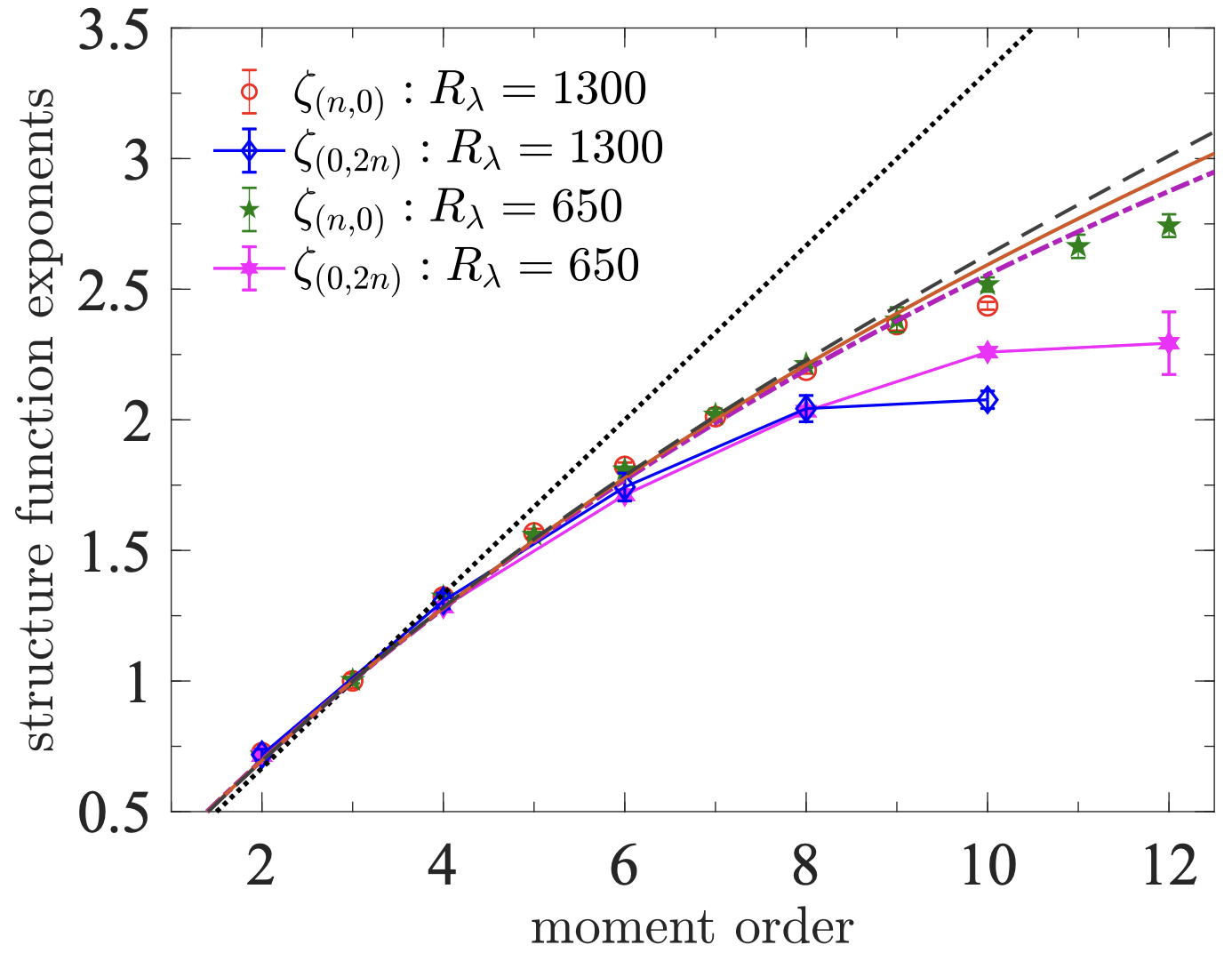}}
  \end{center}
  \caption{Anomalous scaling exponents from  \cite[Chapter 8]{F95} and \cite{ISY20} }\label{figexp}
\end{figure}

Where is the weak point in the argument we sketch for the behavior  \eqref{Sppred}?  As we have seen from Theorem \ref{45thlaw}, the application of the $\sfrac{4}{5}$ law is quite justified. This theorem also, in part, justifies the constancy of the flux in the inertial range since $\langle\varepsilon\rangle$ directly appears there.  Well then, what is the evidence towards this non-intermittent dissipation postulate?  As it turns out, this seems to be the biggest gap in Kolmogorov's picture since in the typically 3D turbulence situation, it appears badly violated. Lev Landau was apparently the first to recognize this\footnote{Onsager also knew of intermittency. In his 1945 letter to C.C. Lin stated that energy spectral exponent might more negative than $-\sfrac{5}{3}$ 
because of “spottiness” of velocity-increments in space. \cite{ES06}}, and pointed it out to Kolmogorov after a seminar in Kazan in 1942 (see discussion in Frisch \cite{F95} for a deeper historical analysis). In their text, Landau and Lifshitz \cite{LL}.
\vspace{2mm}

\begin{quotation}
\emph{"It might be thought that the possibility exists in principle of obtaining a universal formula, applicable to any turbulent flow, which should give $S_2(\ell)$ for all distances $\ell$ that are small compared to $\ell_I$.   In fact, however, there can be no such formula as we see from the following argument.  The instantaneous value of $(\hat{z}\cdot \delta_{\ell\hat{z}} u)^2$ might, in principle, be expressed as a universal function of the dissipation $\ve[u]$ at the instant considered.  When we average these expressions, however, an important part will be played by the manner of variation of $\ve[u]$ over times of order of the periods of the large eddies (with size $\sim \ell_I$), and this variation is different for different flows. The result of the averaging therefore cannot be universal."}\\
\phantom{adsf} \hfill  --  Landau and Lifshitz,  \emph{Fluid Mechanics} \cite{LL}
\end{quotation}

\begin{figure}[h!]
  \begin{center}
    \includegraphics[width=0.33\textwidth]{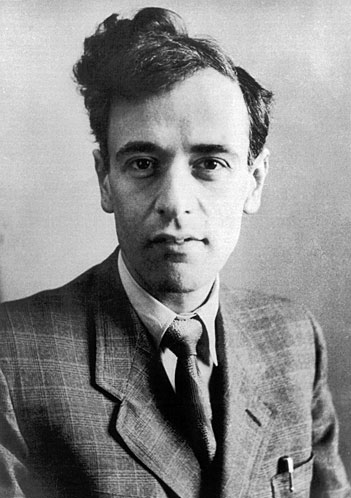}
     \hspace{5mm} 
        \includegraphics[width=0.55\textwidth]{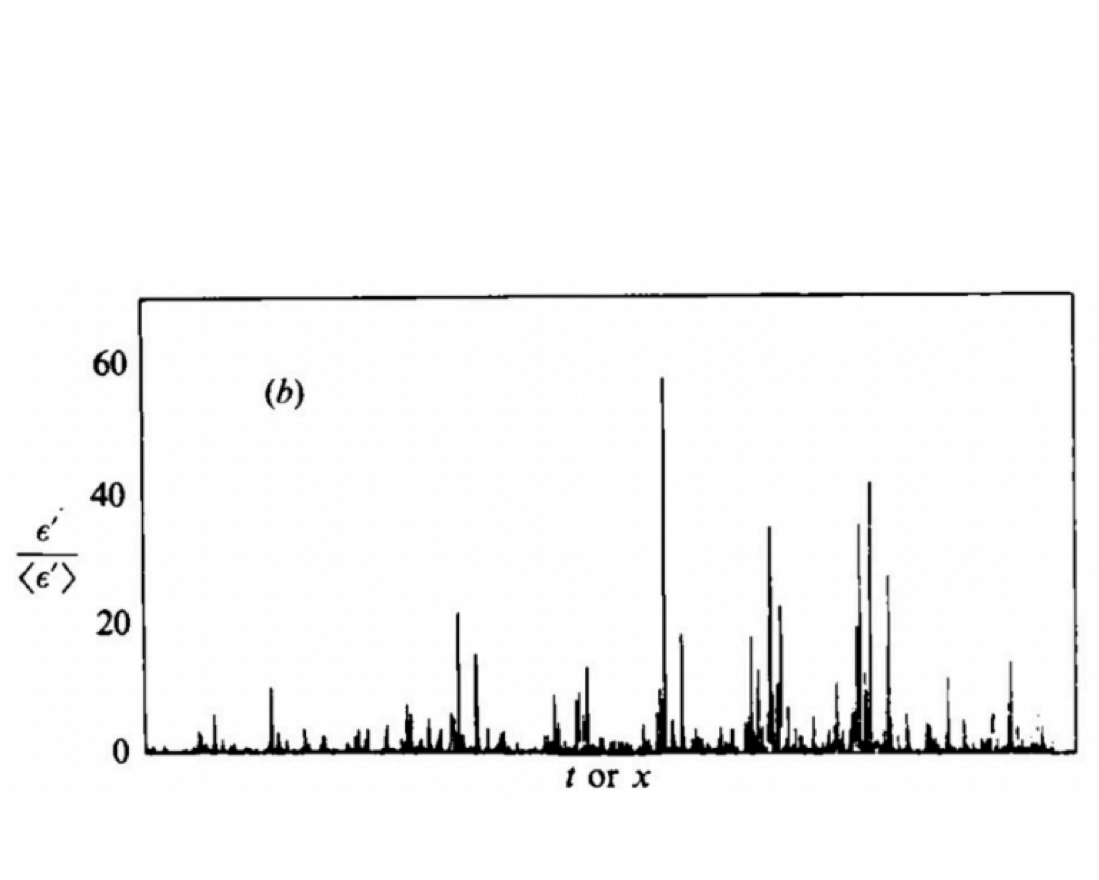}
  \end{center}
  \caption{Lev D. Landau and intermittent energy dissipation \cite{MS91} } \label{Lanfig}
\end{figure}

Landau's text concerns the temporal variability of the energy dissipation rate, see Figure \ref{Lanfig}.  More generally, the dissipation thought of as a measure on spacetime should be highly non-uniform/intermittent, giving positive mass to lower dimensional sets (think of, e.g. a shock wave). See Figure \ref{figdiss} for a rendering of volumes with extreme energy dissipation from numerical simulation of turbulence in a periodic box. The upshot of charging a lower dimensional set is that the dissipation measure must \emph{diverge} at some spacetime locations.  Let us reexamine Kolmogorov's argument for  \eqref{Sppred}, in view of the rigorous Theorem \ref{45thlaw}.  This theorem says that, in the sense of distributions, Kolmogorov's $\sfrac{4}{5}$ law holds. Specifically, it implies the following coarse grained identity should hold in the inertial range $\ell_\nu \ll \ell \ll \ell_I$:
\be\label{mol45}
\frac{\ol{( \langle  (\hat{z}\cdot\delta_{\ell \hat{z}}u^\nu )^3 \rangle_{\rm ang})}_{\ell}}{\ell}\approx - \frac{12}{d(d+2)} \ol{\ve^\nu[u^\nu]}_\ell,
\ee
where, recall $\ol{f}_\ell = f* G_\ell$ for some \emph{space-time} mollifier $G_\ell$.  
Thus, repeating Kolmogorov's  argument slightly more precisely, we have
\begin{align} \nonumber
S_{p,\|}^{u^\nu}(\ell)&:= \fint_0^T \fint_M  \langle  (\hat{z}\cdot\delta_{\ell \hat{z}}u^\nu )^p \rangle_{\rm ang} \rmd x \rmd t\qquad\qquad\qquad \ \ \quad\! \! \! \text{by definition \eqref{longsp}}\\ \nonumber
&\approx \fint_0^T \fint_M \langle  (\hat{z}\cdot\delta_{\ell \hat{z}}u^\nu )^3 \rangle_{\rm ang}^{p/3} \ \rmd x \rmd t\qquad \qquad \qquad  \ \    \ \text{by isotropy}\\ \nonumber
&= \fint_0^T \fint_M\ol{ ((\hat{z}\cdot\delta_{\ell \hat{z}}u^\nu )^3 \rangle_{\rm ang}^{p/3}) }_\ell \rmd x \rmd t\\ \nonumber
&\approx \fint_0^T \fint_M\ol{( \langle  (\hat{z}\cdot\delta_{\ell \hat{z}}u^\nu )^3 \rangle_{\rm ang})}_{\ell}^{p/3}\rmd x \rmd t\\ \nonumber
&\approx \fint_0^T \fint_M \ ( \tfrac{12}{d(d+2)} \ol{\ve[u^\nu]}_\ell \ell)^{p/3}  \rmd x \rmd t  \quad \ \quad\quad\quad \quad \ \   \    \text{by identity \eqref{Dell45}/\eqref{k45law}} 
\end{align}
In the above, we have introduced a  mollification so as to make use of the  exact identity \eqref{mol45} in a more precise fashion. As a result, the mollified energy dissipation makes its way into the picture (this is also the basis for Kolmogorov's own refinement of his theory from 1962 \cite{K62},  which features  ball averaged dissipation). See also Kraichnan \cite{Kr94}.
At this point, Kolmgorov's original argument relates the $\sfrac{p}{3}$ moment of the dissipation simply to its mean.  Instead, following Landau's suggestion, a reasonable implication of the observed spottiness of the dissipation would be power-law divergences of higher moments
\be\label{landaubw}
\fint_0^T \fint_M \  \ol{\ve[u]}_\ell ^{p} \ \rmd x \rmd t  \sim  \langle \varepsilon \rangle^{p} \left(\frac{\ell}{\ell_I}\right)^{-\alpha_{p}}.
\ee
See, e.g., the discussion in \cite{MS87,MS91,AFLV}.
Here, we must have $\alpha_1=0$,  and based on the above, we might expect $\alpha_1\geq 0$ for $p> 1$ and   $\alpha_p\leq  0$  for $p\in[0,1)$ which relates to the size of the sets that $\ve[u]$ charges, and the rate at which $\ve[u]$ diverges.  
Indeed, consider the case of the dissipation being concentrated on a smooth, co-dimension one hypersurface.  This essentially reduces to thinking of $\ve[u](x,t)=  \langle\varepsilon \rangle \delta_{x=0}$ in one dimension, so that   $ \ol{\ve[u]}_\ell (x,t) = \langle\varepsilon \rangle G_\ell(x)$ as an approximation to the identity on the real line $\mathbb{R}$ at each time $t$, where $G$ is smooth, positive, even and ${\rm supp} \ \!G \subset [-\ell_I, \ell_I]$.  Then
\begin{align}\label{shockbd}
\fint_0^T \fint_M\  \ol{\ve[u]}_\ell^{p} \ \rmd x \rmd t  &=  \langle\varepsilon \rangle^{p}  \fint_\mathbb{R} G_\ell^{p}(x)  \rmd x =\|G\|_{L^{p}}^{p}\langle\varepsilon \rangle^{p}  \left(\frac{\ell}{\ell_I}\right)^{-(p-1)}.
\end{align}
So $\alpha_p=p-1$ in this case.
Using \eqref{landaubw} inside Kolmogorov's argument, we arrive at 
\be
  \label{k41a} 
S_{p,\|}^{u^\nu}(\ell)  \sim C_p (\langle \ve \rangle \ell)^{p/3}  \left(\tfrac{\ell}{\ell_I}\right)^{-\alpha_{p/3}} \sim   \left(\tfrac{\ell}{\ell_I}\right)^{ \sfrac{p}{3} - \alpha_{p/3}}  \qquad \text{for} \qquad \ell_\nu \ll \ell \ll \ell_I.
\ee
This formula relates two unknown scalings for the structure function \eqref{Sppred}, $\zeta_p= \sfrac{p}{3} - \alpha_{p/3}$.
For the dissipation charging codimension one sets (e.g. shocks) this predicts $\zeta_p=1$, as $\alpha_1=0$. This indeed holds, at least for $p\geq 3$ in models like the Burgers equation (see, e.g. Theorem \ref{thm} and discussion herein) and compressible Euler \cite{Daf83,DE18,BDSV22}.

We remark that the above argument can be made tighter by actually controlling the error made from exchanging mollification with the $\sfrac{p}{3}$ power.  Potentially, one could show that the error from doing this is the same, or lower, order in $\ell$ than \eqref{k41a} by bootstrapping this information.  On the other hand, I don't see a clean way to quantify the error incurred in the first step of the argument which appeals to "isotropy".

We remark also that the longitudinal structure function $S_{p,\|}^{u^\nu}(\ell)$ is not coercive in that it does not imply regularization in the Besov-type space introduced in Lemma \ref{Slem}. Nontrivial scaling exponents (or bounds) on absolute structure functions  $S_{p}^{u^\nu}(\ell)$ \eqref{SP} are. This can be related to multifractal picture of the turbulence velocity field \cite{FP,Eyink93,Eyink95}.  It is the \emph{absolute} structure functions that will be the subject of the results  to follow.

\begin{figure}[h!]
  \begin{center}
    \includegraphics[width=0.47\textwidth]{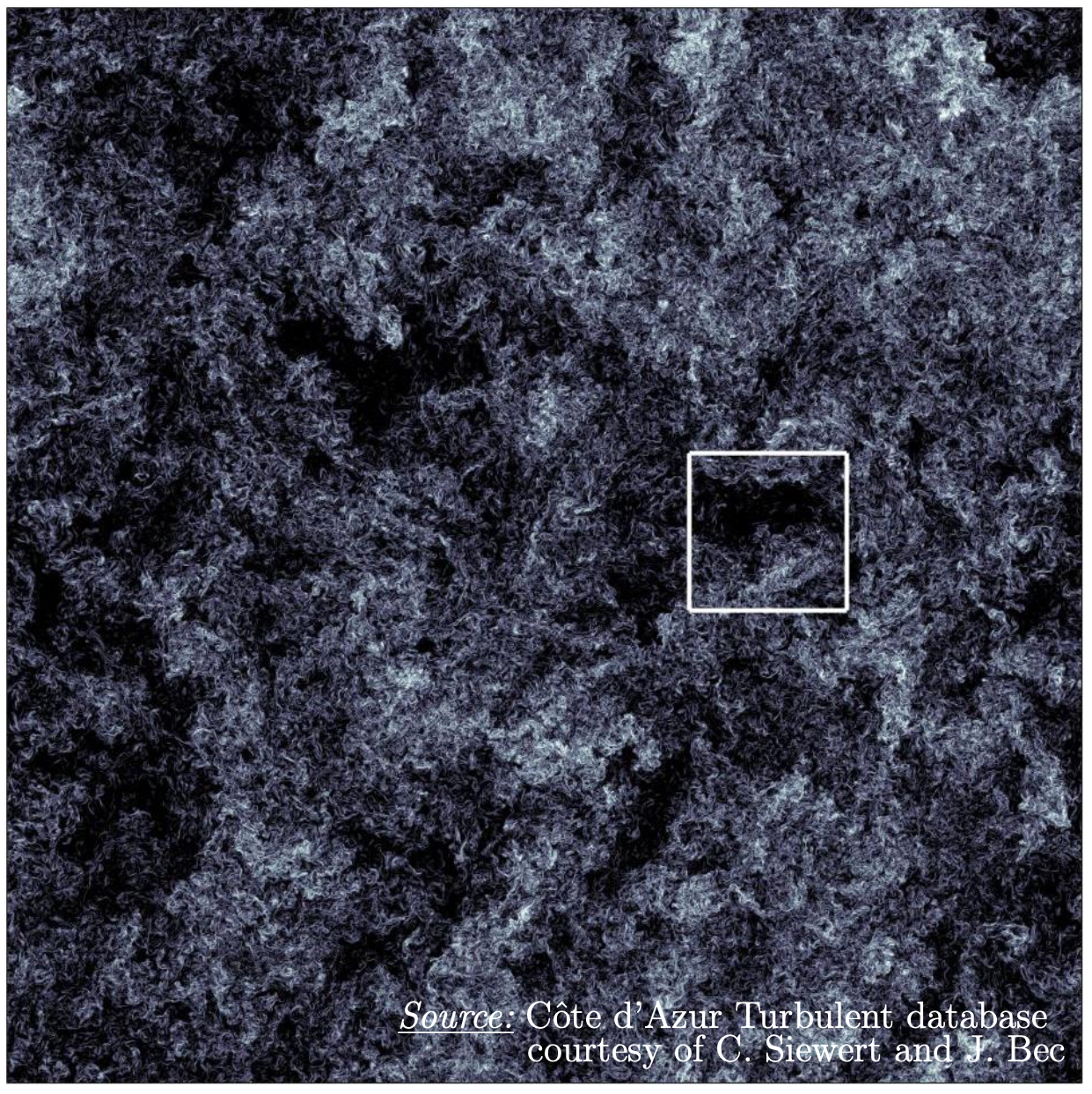}
       \hspace{2mm} 
       {
        \includegraphics[width=0.47\textwidth]{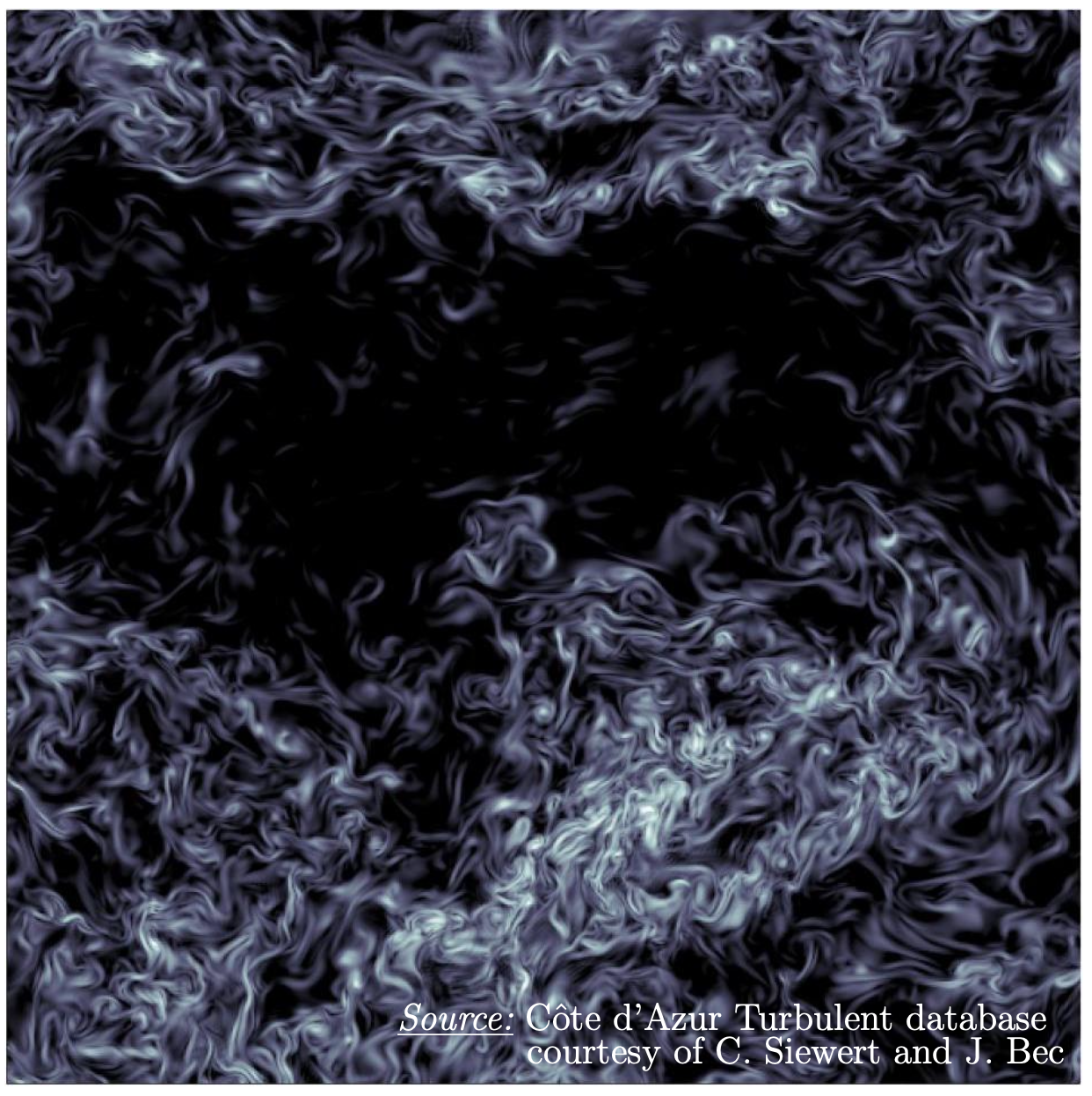}}
  \end{center}
  \caption{Dissipation field in homogeneous isotropic turbulence. }\label{figdiss}
\end{figure}

Following \cite{DDII24,DDII25}, to relate dissipation with regularity and (absolute) structure functions, we take a slightly different perspective based on the distributional energy balance which holds for $L^3$ weak solutions.  Specifically, it will be based on the  balance law in  Proposition \ref{P:decomposition_NS},   which relates the dissipation measure to, effectively, increments of the weak solution.

To exploit this relation, we must quickly introduce a framework which allows us to study a class of lower dimensional measures on $\mathbb{R}^d$ that allow measurement of fractals (sets with  non-integer dimensions). Such measures are called Hausdorff measures, and they will be denoted by $\mathcal{H}^\gamma$. The idea is then to say that a set $A \subset \mathbb{R}^d$ has dimension  $\gamma\in [0,d]$ if $0<\mathcal{H}^\gamma(A)<\infty$, regardless of how complicated the set might be geometrically. To do this, given any $A\subset \R^N$ and $\gamma \geq 0$, for any $\delta >0$ we set 
$$\mathcal{H}^\gamma_\delta (A) := \inf \left\{ \sum_i r_i^\gamma \, :\, A\subset \bigcup_i B_{r_i}, \, r_i<\delta \text{ for all } i \right\}.$$
Namely, it is a sort of $\gamma$--dimensional volume of an optimal covering of the set $A$ by balls of radius no bigger than $\delta$ but may be not uniform.
Then, the $\gamma$-dimensional Hausdorff measure is defined as
\begin{equation}\nonumber
    \mathcal{H}^\gamma(A) := \sup_{\delta > 0} \mathcal H^{\gamma}_\delta(A). 
\end{equation}
This is a non-negative Borel measure on $\R^N$  and $\mathcal H^d$ is equivalent to the $d$-dimensional Lebesgue measure for all $d\in \mathbb{N}$, $d\leq N$. The Hausdorff dimension is obtained as
\begin{equation}\nonumber
    \dim_{\mathcal{H}}A: = \inf \{ \gamma\geq 0\,:\,\mathcal{H}^\gamma(A) = 0 \}. 
\end{equation}
A simpler notion of fractal dimension can be made by demanding that all the balls covering the set $A$ in the definition $\mathcal{H}^\gamma_\delta (A)$ should be of fixed radius $\delta$.  This results in the Minkowski or Box counting dimension.  In particular, the $\gamma$-dimensional upper Minkowski content is 
$$\overline{\mathcal{M}}^{\gamma}(A) := \limsup_{r \to 0} \frac{\mathcal{H}^N([A]_r)}{r^{N-\gamma}} \qquad \text{with }[A]_r = \left\{ x \in \R^N \colon \dist(x,A) < r\right\}.$$
Then, the corresponding upper Minkowski dimension is given by
$$\overline{\dim }_{\mathcal{M}} A := \inf \{\gamma \geq 0\,:\,\overline{ \mathcal{M}}^\gamma(A) = 0\}.  $$
The  lower Minkowski content $\underline{\mathcal{M}}^{\gamma}(A)$ is defined instead by a liminf, and this is used to define the  lower Minkowski dimension $\underline{\dim }_{\mathcal{M}} A$.  Of course, there is the following relationship between the dimensions,
\be
    \dim_{\mathcal{H}}A\leq \underline{\dim }_{\mathcal{M}} A\leq \overline{\dim }_{\mathcal{M}} A,
\ee
with Hausdorff being the sharpest notion among them.

Our first result will be to give \emph{lower bounds} on the smallest dimension of a set the dissipation measure can charge. It is a special case of results from  \cite{DDII24}.
 \begin{theorem}[Lower bounds on the support of dissipation ]\label{lddisst}
Bounded weak solutions of Euler on  $(0,T) \times \mathbb{T}^d$ with non-trivial energy dissipation measure $\ve[u]$  have
$$
    \dim_{\mathcal{H}} ({\rm supp} \ \! \ve[u]) \geq d.
$$
 \end{theorem}
 \begin{remark}[Sharpness of Theorem  \ref{lddisst}]
In fact, this theorem applies essentially to any conservation law equation, such as compressible fluids and transported scalars.  In fact, examples from those other models show the result of the theorem is, in fact, sharp.  For instance, the one-dimensional Burgers equation exhibits shock discontinuities that move along space-time rectifiable curves and support the energy dissipation measure.  As such, the dimension of the support is codimension one in the two-dimensional spacetime \cite{Daf83}. Likewise, regular shock interfaces give examples where the \emph{entropy production} measure is supported on codimension one surfaces \cite{BDSV22}.  Finally, examples of anomalous dissipation for passive scalar transport can be built in such a way that dissipation takes place at a fixed instant of time and therefore is again co-dimension one \cite{DEJI22}.  In all these situations, the solution is bounded and the theorem guarantees the dissipation cannot concentrate on a set of any smaller dimension.
\end{remark}
\begin{proof}[Proof of Theorem \ref{lddisst}]
The proof is given in \cite{DDII24}.
We prove the result here for the upper box counting dimension instead of Hausdorff, so that we do not need to appeal to Frostman's lemma. 
Denote by $[{\rm supp} \ \! \ve]_\delta$ the space-time $\delta$-neighbourhood of ${\rm supp} \ \! \ve$. Let $\chi^\delta\in C^\infty_{x,t}$ be such that
\be \label{chiinfo}
    0\leq \chi^\delta\leq 1,\qquad \chi^\delta \big|_{[{\rm supp} \ \! \ve]_\delta}\equiv 1,\qquad \chi^\delta \big|_{[{\rm supp} \ \!\ve]_{2\delta}^c}\equiv 0 \qquad \text{and} \qquad |\nabla_{x,t}\chi^\delta| \leq 4\delta^{-1}.
\ee
Then, using the fact that $\ve[u]=-\nabla_{t,x} \big(\frac{1}{2}|u|^2, (\frac{1}{2}|u|^2+p)u\big)=: \nabla_{t,x} \cdot V$ is assumed to be a positive measure (this is not essential) we have
\begin{align*}
\ve[u]([0,T]\times \mathbb{T}^d) = \ve[u]([{\rm supp} \ \! \ve]_\delta) &\leq \int_0^T \int_{ \mathbb{T}^d} \chi^\delta \rmd \ve[u]= -\int_0^T \int_{\mathbb{T}^d} V\cdot \nabla_{t,x} \chi^\delta \rmd x \rmd t \\
& \leq \|V\|_{L^\infty_{t,x}} \|\nabla_{t,x} \chi^\delta\|_{L^1} \lesssim \frac{1}{\delta} \mathcal{H}^{d+1}([{\rm supp} \ \! \ve]_{2\delta}) \lesssim \delta^{d-\gamma}
\end{align*}
where $\gamma =  \overline{\dim }_{\mathcal{M}} [{\rm supp} \ \! \ve]$. This vanishes if $\gamma<d$, contradicting the non-triviality of the dissipation measure. The conclusion follows.
\end{proof}

\begin{remark}[The Meneveau--Sreenivasan estimate of the fractal dimension]\label{remms}
The celebrated works of Meneveau and Sreenivasan studied the relationship between properties of the energy dissipation measure to intermittency 
\cite{MS87,MS91,SM88}. These papers  suggest from experiments that, in the infinite Reynolds number limit, the anomalous energy dissipation measure at fixed time is concentrated on a fractal subset of dimension less than the space dimension 3, about 2.87, and has volume zero \cite{MS87}. Moreover, based on the data, it is supposed that this fractal dimension is roughly constant in time, making the inferred space-time support of the dissipation measure to be of dimension 3.87. 
\end{remark}

 \begin{theorem}[Lower dimensional dissipation implies intermittency]\label{lddiss}
     On $(0,T) \times \mathbb{T}^d$, let $\{u^\nu\}_{\nu>0}$ be a sequence of smooth solutions to Navier-Stokes. Assume that, in the limit as $\nu\rightarrow 0$, $\ve^\nu[u^\nu] :=\nu |\nabla u^\nu|^2$ converges in $\mathcal D'_{x,t}$ to a measure $ \ve[u]$ whose singular part with respect to the Lebesgue measure is non-trivial and concentrated on a set $S$ with $\dim_{\mathcal H} S=\gamma \in [ 1,d+1]$.  For all  $p\in [3,\infty]$ for which there exists $\zeta_p\in (0,p)$ such that $\{u^\nu\}_{\nu>0}$ stays bounded in $ L^p_t B^{\frac{\zeta_p}{p}}_{p,\infty}$, it must hold that $S_p(\ell)\lesssim \ell^{\zeta_p}$ with
    \begin{equation}\label{eq: intermittenc viscos}
    \zeta_p \leq \frac{p}{3} -  \frac{2(d+1-\gamma)(p-3)p}{ 9p -3(d+1-\gamma) (p-3)}.
            \end{equation}
 \end{theorem}

\begin{proof}[Proof of Theorem \ref{lddiss}].
The proof is given in \cite{DDII25}.  As we did for Theorem \ref{lddisst},  we give a self-contained proof of this statement replacing the Hausdorff dimensions assumption with an upper Minkowski dimension (the result of the prior work \cite{DI24}). We also prove it directly on the level of weak solutions of Euler which, under the stated assumptions of the theorem, arise as zero viscosity limits. 

Specifically, we show the following.  Let $p\in [3,\infty]$, $\sigma\in (0,1)$ and let $u\in L^p_tB^\sigma_{p,\infty}$ be a weak solution to Euler with dissipation measure $ \ve[u]$ such that $\overline \dim_{\mathcal{M}}\left( {\rm supp} \ \! \ve\right)\leq \gamma$. Then 
$$
\frac{2\sigma}{1-\sigma}>1-\frac{p-3}{p}(d+1-\gamma) \qquad \Longrightarrow \qquad  \ve[u]\equiv 0.
$$
One can easily check the above condition reduces to that reported on $\zeta_p$ above.
To prove our claim, let $\varphi\in C^\infty_{x,t}$ be any compactly supported test function. Denote by $[{\rm supp} \ \! \ve]_\delta$ the space-time $\delta$-neighborhood of ${\rm supp} \ \! \ve$. Let $\chi^\delta\in C^\infty_{x,t}$ be such that \eqref{chiinfo} holds.    Then\footnote{For simplicity, let us assume that there is no Lebesgue part of the dissipation measure. Indeed, if there were we could split $\ve[u]= \ve_{\rm sing}[u]+ \ve_{\rm leb}[u]$.  Our argument would apply to $\langle \ve_{\rm sing}[u], \varphi\rangle=\langle \ve_{\rm sing}[u], \varphi\chi_\delta \rangle =\langle \ve[u], \varphi \chi_\delta\rangle- \langle \ve_{\rm leb}[u], \varphi \chi_\delta\rangle$. For the first term, we proceed as in the proof.  For the second it suffices to say that it tends to zero as $\delta\to 0$, as  $ \ve_{\rm leb}[u]$ has an $L^1$ density and the set $S$ has dimension strictly less that $d+1$ (otherwise the statement is trivial).} $\left\langle   \ve[u],\varphi\right\rangle  =\left\langle    \ve[u],\varphi \chi^\delta\right\rangle  $, and  by the  decomposition \eqref{D_decomp_NS} we infer
    \begin{align*}
        |{\left\langle    \ve[u],\varphi\right\rangle}|&\lesssim \left( \|E^\ell\|_{L^{\frac{p}{2}}_{x,t}} \left(1+ \|{u_\ell}\|_{L^p_{x,t}} \right) + \|{Q^\ell}\|_{L^{\frac{p}{3}}_{x,t}} + \|{C^\ell}\|_{L^{\frac{p}{3}}_{x,t}} \right) \|{\chi^\delta}\|_{L^{\frac{p}{p-3}}_{x,t}}\\
        &\qquad +\left( \|{E^\ell}\|_{L^{\frac{p}{2}}_{x,t}} \left(1+ \|{u_\ell}\|_{L^p_{x,t}} \right)+ \|{Q^\ell}\|_{L^{\frac{p}{3}}_{x,t}}  \right) \|{\nabla_{x,t}\chi^\delta}\|_{L^{\frac{p}{p-3}}_{x,t}}.
    \end{align*}
 Let $\alpha>0$ be a small parameter. The assumption $\overline \dim_{\mathcal{M}}\left( {\rm supp} \ \! \ve\right)\leq \gamma$ implies 
    $$
   \|{\chi^\delta}\|_{L^{\frac{p}{p-3}}_{x,t}}+\delta\|{\nabla_{x,t}\chi^\delta}\|_{L^{\frac{p}{p-3}}_{x,t}}\lesssim \delta^{\frac{p-3}{p}(d+1-\gamma- \alpha)},
    $$
    if $\delta$ is sufficiently small. Thus, by \eqref{est_Eell}, \eqref{est_Qell} and \eqref{est_Cell} we deduce 
    \begin{align*}
         |{\left\langle    \ve[u],\varphi\right\rangle}| &\lesssim \left( \ell^{2\sigma} + \ell^{3\sigma} + \ell^{3\sigma-1}\right) \delta^{\frac{p-3}{p}(d+1-\gamma -\alpha)} + \left( \ell^{2\sigma} + \ell^{3\sigma} \right) \delta^{\frac{p-3}{p}(d+1-\gamma-\alpha)-1}\\
         &\lesssim \ell^{3\sigma-1}\delta^{\frac{p-3}{p}(d+1-\gamma-\alpha)} + \ell^{2\sigma}\delta^{\frac{p-3}{p}(d+1-\gamma-\alpha)-1},
    \end{align*}
    where the constant depends only on norms of $\varphi$ and $u$. The choice $\ell^{1-\sigma}=\delta$ yields to
     \begin{align*}
         |{\left\langle    \ve[u],\varphi\right\rangle  } 
       |& \lesssim \delta^{\frac{2\sigma}{1-\sigma}-1+\frac{p-3}{p}(d+1-\gamma-\alpha)}.
    \end{align*}
    Since $\sigma,p,\gamma$ satisfy an open condition, we find $\alpha>0$ sufficiently small so that the exponent of $\delta$ in the above inequality is positive. The proof follows by contradiction of the non-triviality of $ \ve[u]$ by letting $\delta\rightarrow 0$.
\end{proof}

\begin{remark}[Implication for intermittency]
Note that if $\gamma$ is less than $d + 1$, then for $p > 3$ this bound implies the structure function exponent $\zeta_p$ must be strictly below $\sfrac{p}{3}$. In particular, non-trivial lower dimensional dissipation necessarily results in a quantitative downward deviation from the Besov $\sfrac13$ regularity for all $p>3$. Recall also from Remark \ref{remms} that Meneveau and Sreenivasan estimated that  the  fractal dimension of the dissipation support to be $\gamma \approx  3.87$. 
We now give an argument for this value based on Theorem \ref{lddiss}.
Suppose Theorem \ref{lddiss} is sharp locally\footnote{We are mainly appealing to two facts for $p\approx 3$: the dissipation is what constraints the regularity the most and almost saturates the upper bound. Modulo the possible gap of longitudinal vs. absolute increments, the exactness of the $\sfrac{4}{5}$ law makes both claims valid at $p=3$.}  
at $p=3$.  Set then
$$
\zeta_p^* := \tfrac{p}{3} -  \tfrac{2(d+1-\gamma)(p-3)p}{ 9p -3(d+1-\gamma) (p-3)},
$$
i.e. the right hand side in $\eqref{eq: intermittenc viscos}$.
In dimension $d=3$, 
$$
\tfrac{d \zeta_p^*}{dp}\Big|_{p=3}=\tfrac{2\gamma -5}{9}.
$$
Numerical simulations of incompressible turbulence  \cite{ISY20} indicate $\frac{d \zeta_p}{dp}\Big|_{p=3} = 0.303   \pm 6.4\times 10^{-4}$.  This corresponds to $\gamma = 3.85$, remarkably close to the observations of  Meneveau and Sreenivasan.  See Figure \ref{fig1} for an inspection of numerical data and this bound. 
		\begin{figure}[h!]
		\centering
			\includegraphics[width=0.5\textwidth]{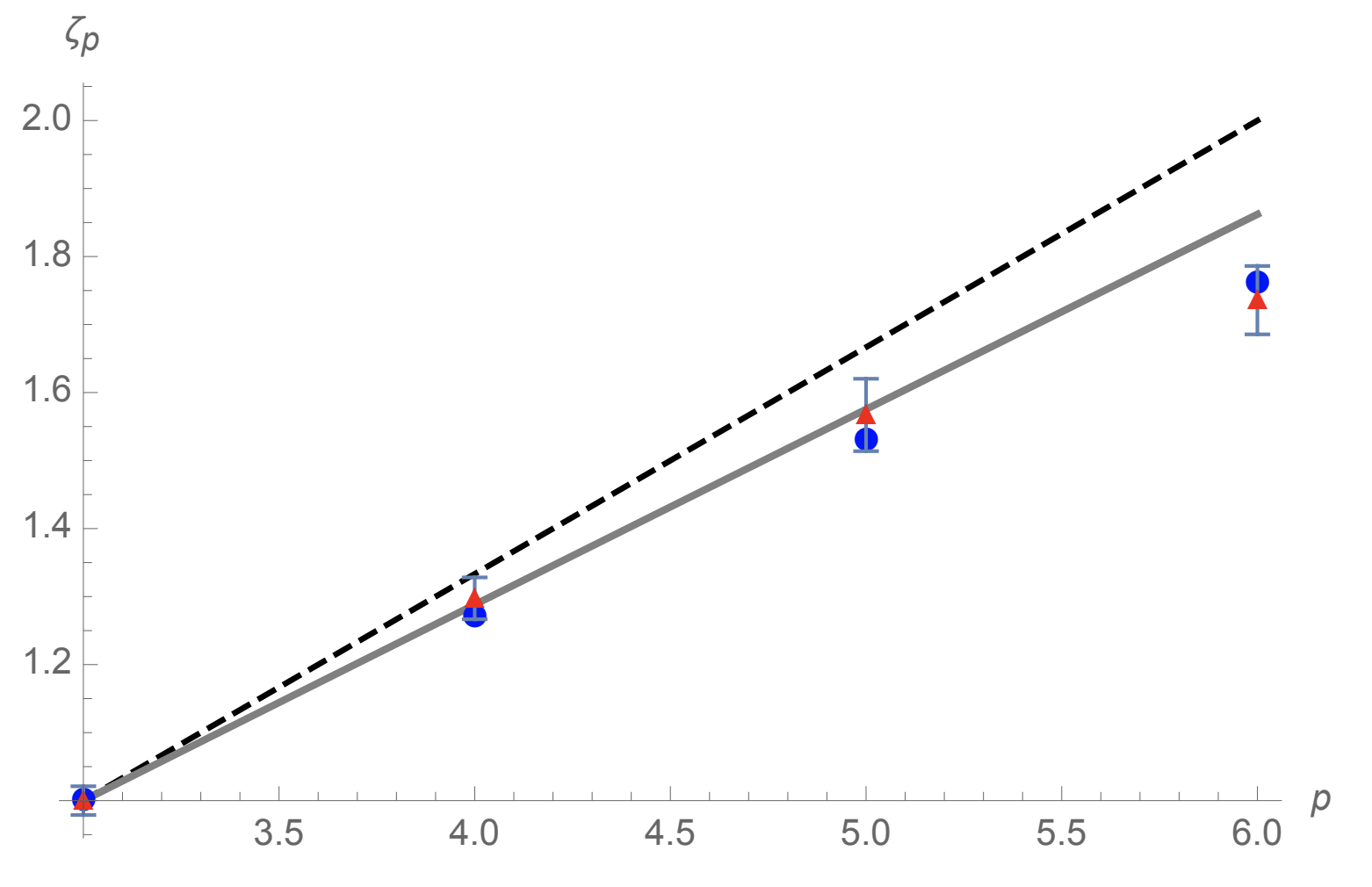} 
					\caption{Structure function exponents for $p\in[3,6]$ \cite{DDII25}.  Blue dots are absolute structure function exponents measured from the JHU turbulence database \cite{JHUDB}. Red triangles are transverse exponents reported in \cite{ISY20}.  Dashed black line corresponds to the Kolmogorov prediction of $\frac{p}{3}$.  Solid grey line corresponds to $\zeta_p^*$ with $\gamma = 3.85$.}\label{fig1}
	\end{figure}
\end{remark}

Finally, we report one additional result from \cite{DDII25} which limits the rate at which moments of ball-averaged dissipation can diverge.
 \begin{theorem}[Bounds on ball-averaged dissipation]\label{badiss}
     On $(0,T) \times \mathbb{T}^d$, let $u$ be a  weak Euler solution  with anomalous dissipation measure $\ve[u]$. Assume $u\in L^p_t B^{\sigma_p}_{p,\infty}$ for some $p\in [3,\infty]$ and $\sigma_p \in \left(0,1\right)$. Then there exists an $\ell_0$ such that for all $\ell<\ell_0$
            \begin{equation}
                \ve[u](B_\ell (x,t))\lesssim \ell^{\frac{2\sigma_p}{1-\sigma_p}-1 +\frac{p-3}{p}(d+1)} \qquad \forall (x,t) \in (0,T) \times \mathbb{T}^d.
                \end{equation}
  \end{theorem}
This bound follows from placing the distribution $\ve[u]$ in an appropriate Besov space with negative regularity index, $B_{p/3,\infty}^{\frac{2\sigma_p}{1-\sigma_p}-1}$, locally in space-time. Let us inspect the implications of this last result to our heuristic Kolmogorov/Landau argument at the start of the section.  Conflating $\frac{\ve[u](B_\ell (x,t))}{|B_\ell(x,t)|}$ with $\ol{\ve[u]}_\ell$, we obtain a constraint on \eqref{landaubw}, namely
  \be\nonumber
\fint_0^T \fint_M \  \ol{\ve[u]}_\ell ^{p} \ \rmd x \rmd t  \gtrsim  \langle \varepsilon \rangle^{p} \left(\frac{\ell}{\ell_I}\right)^{-\alpha_{p}} \quad \implies \quad \alpha_p\leq p\left(1  +\frac{3}{p}(d+1)- \frac{2\sigma_p}{1-\sigma_p}\right).
\ee              
In principle, such information can be used to further constrain the argument of Kolmogorov and Landau.

            In fact, this quality of the dissipation measure was also studied by Meneveau and Sreenivasan. Based on a multifractal model \cite{FP}, they given evidence for the prediction  
            \be
\fint_0^T \fint_M \  \ve[u](B_\ell (x,t))^p \ \rmd x \rmd t  \sim  \langle \varepsilon \rangle^{p}  \left(\frac{\ell}{\ell_I}\right)^{\tau_p}
            \ee
              where the spectrum of exponents $\tau_p$ are determined by the formula
              $\tau_p = \inf_\alpha [\alpha p - f(\alpha)]$,
              which requires as input a spectrum of ``dimensions" $f(\alpha)$ 
              \be\nonumber
               f(\alpha) = \text{codimension on the set of $(x,t)$ such that $\ve[u](B_\ell (x,t))\sim \ell^\alpha$ as $\ell\to 0$}.
              \ee
              
Let us conclude the section with some discussion of the results.
All the results of this section \emph{assume} the dissipation measure charges some lower dimensional set, and derives implications of this to the structure of the velocity field.  Specifically, we show the velocity cannot be too regular for this to happen, and the regularity cannot be uniform in $p$, at least for $p>3$. This implies that some form of intermittency must occur.  On the other hand, we do not have any way of deducing that the dissipation measure \emph{must} charge lower dimensional sets from measurements of the velocity field itself.  One must instead measure the dissipation field directly, as was done by Meneveau and Sreenivasan.  However, the multifractal model picture for Frisch and Parisi \cite{FP} (see also lecture notes of Eyink \cite{Eyink} and review of DuBrulle \cite{Dubrulle19}) predicts that the slope of the curve of $\zeta_p$ at $p=0$, namely $h_*:= \frac{\rmd \zeta_p}{\rmd p}\Big|_{p=0}$ relates to the "most probable" H\"{o}lder exponent in the flow. Namely the one occupying a full space-time volume, with other local exponents $h\neq h_*$ occupying fractal sets of lesser dimension.  Since $\zeta_p$ is a concave function of $p$, if any value of $\zeta_p>p/3$ for $p\in (0,3)$, then it follows that this exponent $h_*>1/3$.  Thus, in light of Onsager's theorem about  $\sfrac{1}{3}$ being the requisite singularity for dissipativity,  one might guess that $\ve[u]$ will charge lower dimensional if $\zeta_p>p/3$ for $p\in (0,3)$.  We therefore end with the following

\begin{question}\label{smallp}
Give sufficient conditions, perhaps involving intermittency of exponents $\zeta_p$ for $p\in (0,3)$, to ensure lower dimensionality of the dissipation measure.
\end{question}

\section{Lessons from model problems}\label{models}
\vspace{2mm}

\begin{quotation} 
\emph{"It might be worth while to investigate the properties of certain systems of
mathematical equations, which, although much simpler in structure than
the equations of hydrodynamics, nevertheless show features which can be
considered as the analogues of typical properties of the hydrodynamic
equations."} \hfill J.M. Burgers \cite{B39}
\end{quotation}
\vspace{2mm}

In this section, we will give details on two model problems that shed light on a variety of aspects of the theory explained above.  In these models, one can make simple, rigorous and unconditional statements.
The first model is the Burgers equation, the second is passive scalar transport.

\subsection{Burgers' equation for pressureless gas}

Consider the viscous Burgers equation  \cite{B48} on $x\in \mathbb{T}_L :=[0,L)$ with periodic boundary conditions for $t\geq 0$
\begin{align}\label{Burgers}
\partial_t u^\nu + u^\nu \partial_x u^\nu &= \nu \partial_x^2 u^\nu + f,\\
u|_{t=0}&=u_0.\nonumber
\end{align}
		\begin{figure}[h!]
		\centering
			\includegraphics[width=0.9\textwidth]{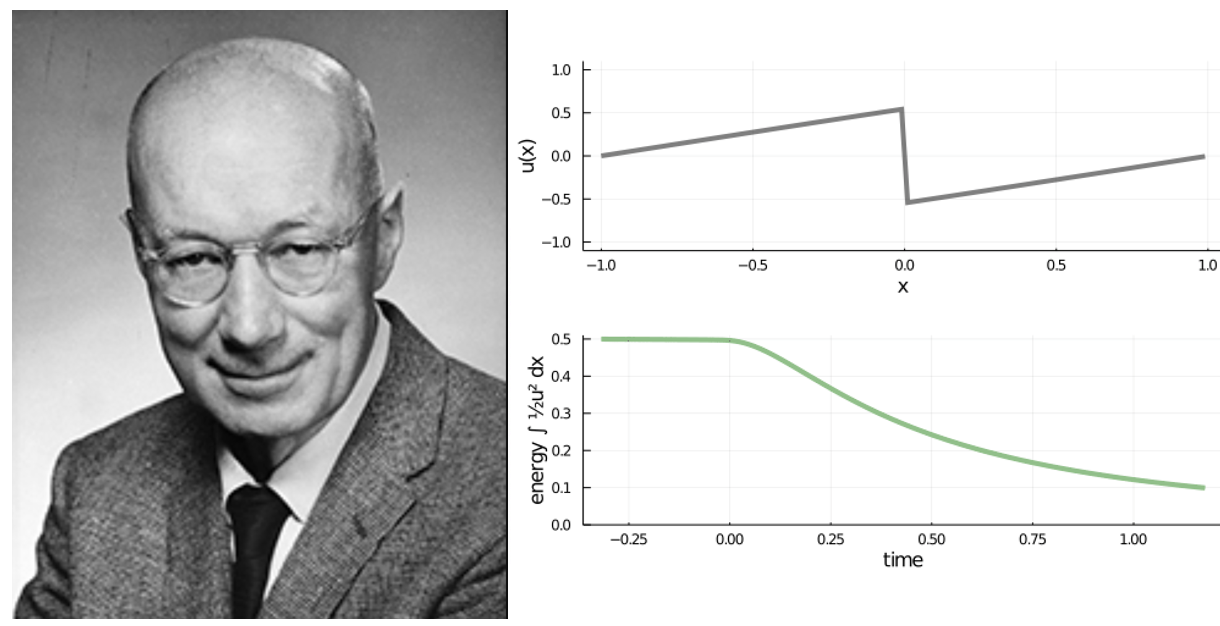} 
					\caption{J.M. Burgers and a dissipative shock.} \label{burgfig}
	\end{figure}
This equation was introduced by Burgers as a simple model displaying a number of the salient features of turbulence, see \cite{FB02, BK} for excellent surveys.  It is also a model of one-dimensional pressureless gas dynamics, and its multidimensional analogue has been used to great effect in cosmology by Zeldovich \cite{Z70}.   In what follows, we aim to see
\begin{itemize}
\item singularity formation from smooth data in the inviscid case,
\item formation of discontinuities in  inviscid weak solutions,
\item anomalous dissipation in the limit of zero viscosity,
\item self-regularization of dissipative inviscid weak solutions, and intermittency.
\end{itemize}
We tackle the first two points above at once, following the proof of Alinhac \cite{A10} (it becomes a statement about singularity formation  as soon as $u_0$ is in a strong enough class, having a local well-posedness theory, e.g. $C^k$ with $k\geq 1$).
\begin{theorem}[Shock formation]\label{thmB1}
\textit{For any non-constant  $u_0\in C(\mathbb{T}_L)$, there is no  global-in-time \underline{continuous} weak solution $u\in C(\mathbb{R}^+\times \mathbb{T}_L)$ of the $\nu=0$ Burgers equation.}
\end{theorem}
\begin{proof}[Proof of Theorem \ref{thmB1}]
Set $f=0$ for simplicity - the proof can easily be modified to accommodate external forcing.
Suppose there existed a continuous weak solution $u\in C(\mathbb{R}^+\times \mathbb{T}_L)$ of the Burgers equation starting from data $u_0$: for all $\varphi\in C_0^\infty((0,\infty)\times \mathbb{T}_L)$ the following holds
\be
(u, \partial_t \varphi)_{L^2(\mathbb{R}^+\times \mathbb{T}_L)}+ (\tfrac{1}{2} u^2, \partial_x \varphi)_{L^2(\mathbb{R}^+\times \mathbb{T}_L)} =0.
\ee
Let $\varphi(x,t)= \phi(x)\eta(t)$ where $\phi\in C_0^\infty(0,L)$ and $\eta\in C_0^\infty(\mathbb{R})$.  
Letting $\phi\rightarrow 1$ and $\eta(s)\rightarrow \chi_{[0,t]}(s)$ for any $t>0$, we see that momentum is conserved:
\be
\int_{ \mathbb{T}_L} u(x,t) \rmd x = \int_{ \mathbb{T}_L} u_0(x) \rmd x:= \mathsf{m} \qquad \forall \ t\geq 0.
\ee
Thus, we may consider the data $u_0$ to be mean zero (e.g. take $ \mathsf{m}=0$) without loss of generality, since otherwise we may generate such a new weak solution ${u}(x,t) = u(x+\mathsf{m}t,t)- \mathsf{m}$  by subtracting off the momentum and moving to the appropriate Galilean reference frame.  Now, given that  $u_0$ is continuous, mean-zero and $L$--periodic, it must have at least two zeros in $[0,L)$, call them $a$ and $b$. Without loss of generality, suppose $a<b$ and that $u_0|_{(a,b)}>0$. By virtue of being continuous, these zeros can be seen to propagate for any such weak solution \cite{Daf}. Let now $\psi$ be a smooth  function having  the property that $ \psi'|_{[a,b]} >0$ and $\psi|_{[a,b]}\geq 0$.  Choose $\phi \to \frac{1}{b-a}\chi_{[a,b]}\psi$ and $\eta\wc   \delta_t$.  This results in the following information on the integral $U_\psi(t) = \frac{1}{b-a}\int_a^b u(x, t) \psi(x) \rmd x$ (which, as $u$ is a continuous weak solution of Burgers, is a differentiable function of time): 
\be
\frac{\rmd}{\rmd t} U_\psi(t) = \frac{1}{2(b-a)}\int_a^b u^2(t,x)  \psi'(x)\rmd x \geq  \frac{U_\psi^2(t)}{(b-a)\mathsf{K}_\psi } 
\ee
where $K_\psi:= \frac{2}{b-a}\int_a^b \frac{\psi^2(x)}{\psi'(x)} \rmd x$.  The above  inequality follows by  Cauchy–Schwarz. As, by our assumptions, we have that $U_\psi(0) >0$ it follows that as $t$ approaches the time $\frac{(b-a)K_\psi}{ U_\psi(0)}>0$, the integral $U_\psi(t)$ becomes unbounded.  This contradicts $u$ itself being continuous and hence bounded.  
\end{proof}

From the above, we know any weak solution must develop a discontinuity of some kind.  In fact, entropy solutions of Burgers develop shocks which propagate along rectifiable curves in space time. What should we expect, then, about anomalous dissipation and scaling of structure functions? There is a simple exact weak solution of Burgers that can give us an expectation: the Khokhlov sawtooth solution on $\mathbb{T}_L= [-L,L)$ and $t>0$
\be\label{ksol}
u(x,t) =\begin{cases} \frac{x+L}{t}  & -L\leq x\leq 0 \\
\frac{x-L}{t}  &\ \ 0\leq x\leq L
\end{cases}.
\ee
This solution has one downwards jump  $\Delta u = u^--u^+= \frac{2L}{t}$, where $u^\mp$ are the left/right traces of $u$ at zero.   
One immediately sees this solution is dissipative; for $t_0<t$ one has
\be
\frac{\rmd}{\rmd t} \frac{1}{2} \|u(t)\|_{L^2}^2  = -\frac{1}{3} \frac{L^2}{t^3} = - \frac{1}{12 } \frac{(\Delta u)^3}{L}<0.
\ee
Indeed, the inviscid limit of the dissipation measure can be studied  for this data (see, e.g. \cite{ED15}) and it converges to a Dirac mass at the shock.
\be\label{dissrelburg}
\ve^\nu[u^\nu] := \nu |\partial_x u^\nu|^2 \to \frac{1}{12 } \frac{(\Delta u)^3}{L}\delta_0, \qquad \langle \ve \rangle := \langle \ve^\nu[u^\nu]  \rangle  = \frac{1}{12 } \frac{(\Delta u)^3}{L}.
\ee
Structure functions can also easily be computed
\begin{align} \nonumber
S_p(\ell) &:=\frac{1}{2L} \int_{-L}^L |u(x+\ell)-u(x)|^p \rmd x\\ \nonumber
&= (1-\frac{\ell}{2L}) \left(\frac{\ell}{t}\right)^p + \frac{\ell}{2L} \left(\frac{2L+\ell}{t}\right)^p \sim_p (\Delta u)^p \begin{cases} \left(\frac{\ell}{2L}\right)^p & 0<p<1\\
 \frac{\ell}{2L} & \ \phantom{0<}p>1
\end{cases}.
\end{align}
Using the relationship with total dissipation \eqref{dissrelburg}, we find for $p>1$ that
\be
S_p(\ell)\sim (\ell  \langle \ve \rangle)^{p/3} \left(\frac{\ell}{L} \right)^{-(\frac{p}{3}-1)}.
\ee
This fully agrees with  our heuristic prediction following Kolmogorov and Landau \eqref{shockbd}--\eqref{k41a}.
In contrast with $p\geq 1$, the behavior for $p\in (0,1)$ does not fall into this prediction and in fact may not be fully universal in decaying Burgulence \cite{SAF92, MHPF}. The significance of $p=1$ here is that it is the dimension of space, which limits the dimension of the set of shock points.   For instance, for bounded data, the entropy dissipation is a (space-time) measure concentrated on countably many Lipschitz curves \cite{BM17}.

Let us prove some rigorous and general results in this direction, tackling the second two points at once.  First we remark about two important apriori estimates for Burgers solutions.
The most basic estimates for the solution $u^\nu$ are derived from energy balance
\be\label{energybal}
\nu \int_0^T \int_{\mathbb{T}} | \partial_x u^\nu |^2\rmd x\rmd t = \tfrac{1}{2} \|u_0\|_{L^2}^2- \tfrac{1}{2} \|u^\nu(\cdot,T)\|_{L^2}^2 + \int_0^T \int_{\mathbb{T}}u^\nu f \rmd x \rmd t,
\ee
which confers $L^2(0,T;H^1(\mathbb{T}))$ regularity to the solution, but  not uniformly in the viscosity.  The strongest \textit{uniform}  estimate at first sight is $L^\infty(0,T;L^\infty(\mathbb{T}))$, following from the maximum principle 
\be\label{maximumprinciple}
\sup_{t\in [0,T]}\|u^\nu(t)\|_{L^\infty(\mathbb{T})} \leq \|u_0\|_{L^\infty(\mathbb{T})} + T\|f\|_{L^\infty(\mathbb{T})}.
\ee
In fact, these two bounds are essentially enough to prove global regularity of the viscous Burgers equation.  Starting from any bounded data $u_0\in L^\infty$, the solution immediately becomes real analytic, is unique, and exists for all times.  So from now on, any manipulation with the viscous Burgers solution will be justified.

Our main result is a \emph{uniform} fractional regularity of $u^\nu$ related to the $p$th--order structure functions
\be \nonumber
S_p^\nu(\ell,T) :=\frac{1}{T} \int_0^T \int_{\mathbb{T}} |u^\nu(x+\ell,t) - u^\nu(x,t)|^p \rmd x \rmd t.
\ee
As discussed at length in \S \ref{nssec}, for turbulent flows these objects develop scaling ranges
\be\nonumber
S_p^\nu(\ell,T) \sim |\ell|^{\zeta_p} \qquad \text{for} \quad \ell_\nu \ll \ell \ll L.
\ee
Our aim is to establish uniform upper bounds $S_p^\nu(\ell,T) \leq C_p|\ell|^{\zeta_p}$ for all $|\ell|>0$.
\begin{theorem}[Self-regularization and Intermittency]\label{thm}
\textit{For each $T>0$ and $p\geq 3$ and data $u_0\in L^\infty(\mathbb{T})$, $f\in L^\infty(\mathbb{T})$, there is  a  constant $C_p:=C_p(\|u_0\|_{L^\infty},\|f\|_{L^\infty},T)$ so that $\zeta_p \geq 1$ for all $p\geq 3$, that is for  all $|\ell|>0$ and $T>0$, we have
\be\label{unifbnd}
S_p^\nu(\ell,T)  \leq C_p |\ell|.
\ee
If $f$ is time-periodic, then  $C_p$ is independent of $T$.
If $f=0$, for $c_p:= c_p(\|u_0\|_{L^\infty})$ we have uniform decay
\be\label{unifbnd2}
S_p^\nu(\ell,T)  \leq c_p \frac{|\ell|}{T}.
\ee}
\end{theorem}
		\begin{figure}[h!]
		\centering
			\includegraphics[width=\textwidth]{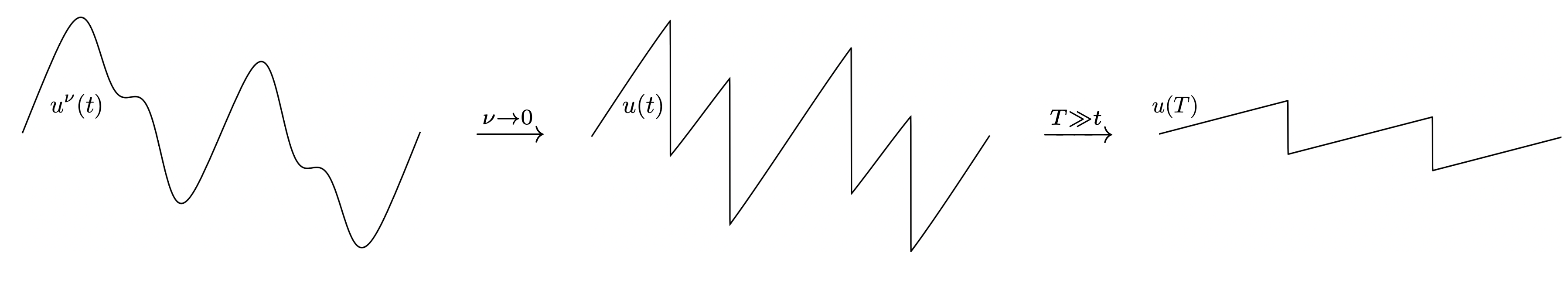} 
					\caption{Burgers solution with $\nu>0$, $f=0$, first for small viscosity and subsequently for long times. } \label{burgfig}
	\end{figure}

Some remarks are in order.  First, the estimate \eqref{unifbnd2} clearly shows anomalous dissipation, since the decay of the structure function is uniform in $\nu$.  Indeed, assuming without loss of generality that the mean is trivial $\mathsf{m}:=\fint_{\mathbb{T}} u^\nu(y,t)\rmd y =0$, then
\be\nonumber
\|u^\nu(\cdot, t)\|_{L^2}^2 = \left\|u^\nu(\cdot, t)-\fint_{\mathbb{T}} u^\nu(y,t)\rmd y \right\|_{L^2}^2 \leq \int_{\mathbb{T}^2} |u(x,t) - u(y,t)|^2 \rmd x \rmd y \lesssim \sup_{|r|\leq L} S_2^\nu(r,t).
\ee
The rate of decay in time may not be sharp, as can be seen by inspection of the exact solution \eqref{ksol}. We remark that the origin of anomalous dissipation in this model can be interpreted as non-uniqueness of particle trajectories backwards in time \cite{ED15} (more on this perspective at the end of the next subsection). 
Secondly, condition \eqref{unifbnd} is equivalent to  $\dashint_0^T\|u^\nu(t)\|_{\dot{B}^{1/p}_{p,\infty}(\mathbb{T})}^p \rmd t \leq C_p,$ with the Besov semi-norm introduced in Lemma \ref{Slem}.
As such, we have a self-regularization result for viscous solutions -- from bounded data, they immediately enter into the space  $X=\bigcap_{p\geq 3} L^\infty(0,T; B^{1/p}_{p,\infty}(\mathbb{T}))$ \textit{uniformly} in the viscosity.
Entropic shocks solutions saturate this regularity in that they live in $X$ and no better space within the Besov scale. 
 Combining Theorems \ref{thm} and \ref{thmB1}, a cartoon of the evolution of the inviscid limit emerges in the form of Figure \ref{burgfig}.

The regularization result is well known \cite{JKM,GP,TT}; in fact there is a stronger uniform estimate for the total variation $L^\infty(0,T;(L^\infty\cap BV)(\mathbb{T}))\subset L^\infty(0,T;B_{p,\infty}^{1/p}(\mathbb{T}))$ for all $p\geq 1$. Our primary interest is in giving a short, intrinsically physical space, argument based on the following principle: nonlinear ideal conserved quantities that may be anomalously dissipated limit the degree to which the solution may suffer irregularities.  A bound of this type was established by Goldman, Josien and Otto \cite{GJO15} using a modified energy flux. For Burgers, there are infinitely many such quantities (any convex function of the solution), suggesting there should be infinitely many such bounds \eqref{unifbnd}, which we here show.  These estimates show that $p$th-order absolute structure functions obey a uniform bound with $\zeta_p=1$ for $p\geq 3$.  This fact underlies the rigidity of the so-called multifractal spectrum of anomalous exponents in Burgulence.  For Navier-Stokes, the only known inviscid invariant that is dissipated is the kinetic energy.  Its flux, related to the dissipation by the Kolmogorov $\frac{4}{5}$--law \cite{K}, is not coercive unlike the present setting.  With an additional alignment hypothesis, the law does confer limited regularity \cite{D22}. Whether this is a generic feature of turbulence is open.

	Finally, we remark that much of the analysis here could carry over to general one-dimensional conservation laws, even those with non-degenerate, nonlinear viscosity
\be\nonumber
\partial_t u^{\ve} + \partial_x h(u^\ve) = \ve \partial_x (\nu(u^\ve) \partial_x u^\ve),
\ee
provided, at least, that the Hamiltonian function $h(u)$ is sufficiently close to quadratic.  For simplicity, we choose to focus here on the viscous Burgers equation. We require the identity (similar to that appearing in \cite{O}):
\begin{lemma}[Burgers increment balance laws]\label{lem}
\textit{Let $\varphi:\mathbb{R}\to \mathbb{R}$ have Lipschitz first derivative. Let $\Phi$ be the primitive of $\varphi$ and $\tilde{\Phi}$ of  $x \varphi'(x)$. Then
\be\label{balance}
\partial_t \varphi(\delta_\ell u) + \partial_x J_\ell[u] + \partial_\ell \Pi[\delta_\ell u] = -\nu \varphi''(\delta_\ell u) |\partial_x \delta_\ell u|^2+ \varphi'(\delta_\ell u)  \delta_\ell f
\ee
where we defined
\vspace{-6mm}
\begin{align*}
J_\ell[ u]& := u'\varphi(\delta_\ell u) - \tilde{\Phi}(\delta_\ell u) - \nu  \partial_x \varphi(\delta_\ell u) , \\
\Pi[\delta_\ell u] & :=  \tilde{\Phi}(\delta_\ell u)-\Phi(\delta_\ell u). 
\end{align*}
}
\end{lemma}
\begin{proof}[Proof of Lemma \ref{lem}]
To see this, let $u'=u(\cdot+\ell)$, $u=u(\cdot)$ and $\delta_\ell u = u'-u$.  Then
\be\nonumber
\partial_t \delta_\ell u + u' \partial_x \delta_\ell u+  \delta_\ell u \partial_\ell \delta_\ell u - \delta_\ell u \partial_x \delta_\ell u= \nu \Delta \delta_\ell u  + \delta_\ell f,
\ee
since 
$
u'\partial_x u' - u\partial_x u  = \delta_\ell u \partial_x u' + u \partial_x \delta_\ell u= \delta_\ell u \partial_\ell \delta_\ell u + u \partial_x \delta_\ell u= \delta_\ell u \partial_\ell \delta_\ell u + u' \partial_x \delta_\ell u- \delta_\ell u \partial_x \delta_\ell u.
$
Multiplying the above by $\varphi'(\delta_\ell u)$, we obtain the evolution, which is equivalent to \eqref{balance}
\begin{align*}
\partial_t \varphi(\delta_\ell u) +\partial_x( u'  \varphi(\delta_\ell u))- \varphi'(\delta_\ell u) \delta_\ell u \partial_x \delta_\ell u  &+\Big(  \varphi'(\delta_\ell u) \delta_\ell u -  \varphi(\delta_\ell u)\Big)  \partial_\ell \delta_\ell u\\
&\qquad  = \nu \varphi'(\delta_\ell u)  \Delta \delta_\ell u +\varphi'(\delta_\ell u)  \delta_\ell f.
\end{align*}
\vspace{-2mm}
\end{proof}

\begin{proof}[Proof of Theorem \ref{thm}]
For $p\geq3$ and $\alpha \in \mathbb{R}$, let $\varphi(x) = \alpha x |x|^{p-2}$. We compute $\varphi'(x) = \alpha (p-1) |x|^{p-2}$ and $\varphi''(x) =  \alpha (p-1) (p-2)x |x|^{p-4}\in L^\infty$. The primitive of $\varphi$ is $\Phi=\frac{\alpha}{p}|x|^{p}$ and for $x \varphi'(x) =  (p-1)   \varphi(x)$, it is $\tilde{\Phi}=(p-1)\Phi$. Thus $\Pi[x] = (p-2)\Phi(x) = \alpha \frac{p-2}{p}|x|^p$. Letting $\alpha= \frac{p}{p-2}$, we obtain the balance
\begin{align*}
\frac{1}{\ell} \  \dashint_0^T \!  \! \! \int_{\mathbb{T}} |\delta_\ell u|^p \rmd x \rmd t &= \frac{1}{T}\dashint_0^\ell \!  \! \! \int_{\mathbb{T}} \varphi(\delta_\ell u_0) \rmd x - \frac{1}{T}\dashint_0^\ell\!  \! \!  \int_{\mathbb{T}} \varphi(\delta_\ell u_T) \rmd x + \dashint_0^\ell \dashint_0^T\!  \! \!  \int_{\mathbb{T}}  \varphi'(\delta_{\ell'} u) \delta_{\ell'} f \rmd x \rmd t \rmd \ell' \\
&\qquad -\nu \ \dashint_0^\ell \dashint_0^T\!  \! \! \int_{\mathbb{T}}  \varphi''(\delta_{\ell'} u) |\partial_x \delta_{\ell'} u|^2\rmd x \rmd t \rmd \ell' .
\end{align*}
Using energy balance \eqref{energybal} and the maximum principle \eqref{maximumprinciple} to bound each term on the right-hand-side, we obtain a uniform-in-viscosity bound.    This yields the claimed estimates.   If $f$ is time-periodic, the uniform-in-time estimate follows from \cite{JKM} which established $\|u(t)\|_{L^\infty}\leq C_0$ for a $t$ independent constant $C_0$, improving the bound \eqref{maximumprinciple}.
 When $f=0$, by \eqref{energybal} we have the decaying estimate 
  \be\nonumber
  \left|\nu \ \dashint_0^\ell \dashint_0^T\!  \! \! \int_{\mathbb{T}}  \varphi''(\delta_{\ell'} u) |\partial_x \delta_{\ell'} u|^2\rmd x \rmd t \rmd \ell' \right| \leq \frac{2}{T} \|\varphi''\|_{L^\infty( 0,\|u_0\|_{L^\infty})} \|u_0\|_{L^2}^2,
  \ee
  which gives \eqref{unifbnd2}.
  This completes the proof.
\end{proof}

\subsection{Obukhov, Corrsin, and Kraichnan's scalar turbulence}

Next we discuss passive advection and diffusion of scalar fields. Let $M\subset \R^d$ be open. Given a divergence-free vector field $u:(0,T)\times M\rightarrow \R^d$, consider the advection-diffusion  equation governing the motion of a tracer dye or temperature 
\begin{equation}\label{T} 
\partial_t \theta^\kappa +\div ( u \theta^\kappa )  =\kappa \Delta \theta^\kappa.
\end{equation}  
Just as for the velocity field modeled by the incompressible Navier-Stokes equations, the tracer energy is decaying.  The local form of the balance is
$$
\partial_t \left(\tfrac{1}{2}|{\theta}^\kappa|^2\right) +\div \left(\tfrac{1}{2}|{\theta}^\kappa|^2u  -\kappa \nabla \tfrac{1}{2}|{\theta}^\kappa|^2 \right)=-\chi^\kappa[\theta^\kappa]
$$
where we have introduced the local measure of scalar dissipation $\chi^\kappa[\theta^\kappa]$ via
\begin{align}\label{chilocdissform}
\chi^\kappa[\theta^\kappa]&:= \kappa |\nabla \theta^\kappa|^2.
\end{align}
See Figure \ref{scalfig} for some snapshots of scalar motion in Navier-Stokes turbulence and the Kraichnan model.
Just as in hydrodynamic turbulence,  it is observed that, provided $u$ is rough enough, \emph{anomalous diffusion} occurs  \cite{DSY,S19}:
\begin{myshade}
\be\label{anomalousdiff}
\liminf_{\kappa \to 0} \fint_0^T  \fint_M\chi^\kappa[\theta^\kappa](x,t) \rmd t\rmd x >0.
\ee
\end{myshade}

This balance holds directly for non-diffusive weak solutions $\theta$ upon replacing $\chi^\kappa[\theta^\kappa]$ with $\chi[\theta]$,
as soon as $u\in L^p_{x,t}$, $\theta\in L^q_{x,t}$ with $\frac{1}{p}+\frac{2}{q}\leq 1$.  Note that, being a linear equation, weak solutions can be obtained from weak compactness of vanishing diffusivity approximations. Provided that the limit is achieved strongly, it follows that $\chi[\theta]$ is actually a non-negative Radon measure.

\begin{figure}[h!]
  \begin{center}
              \includegraphics[width=0.45\textwidth]{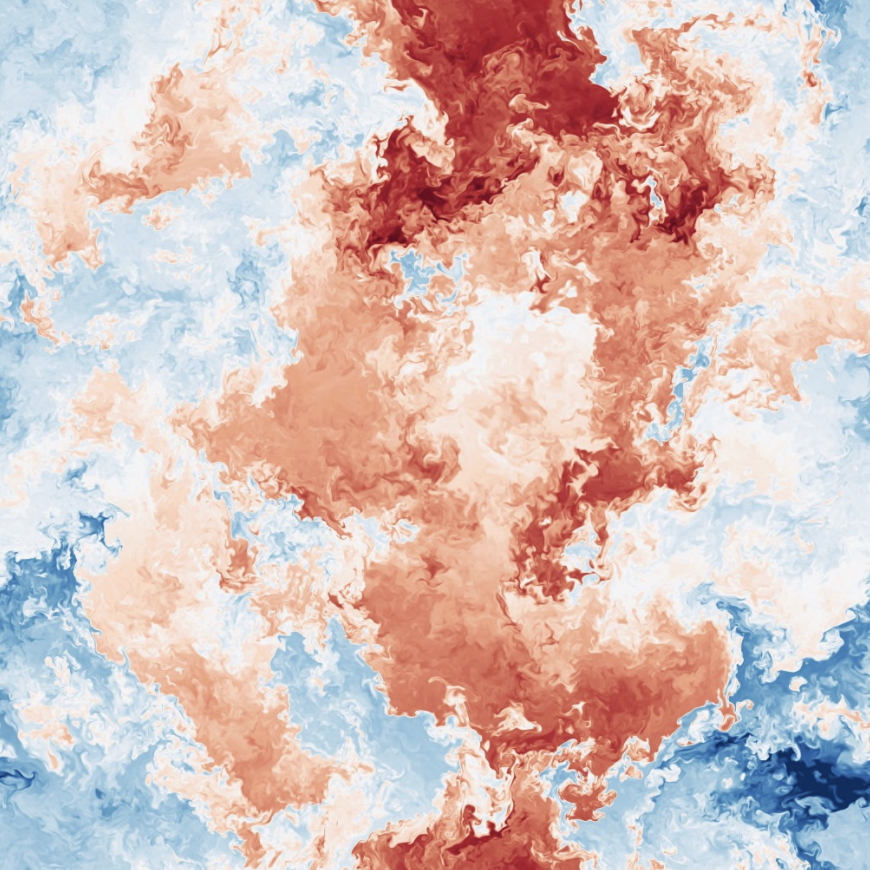}
                     \hspace{2mm} 
                           \includegraphics[width=0.45\textwidth]{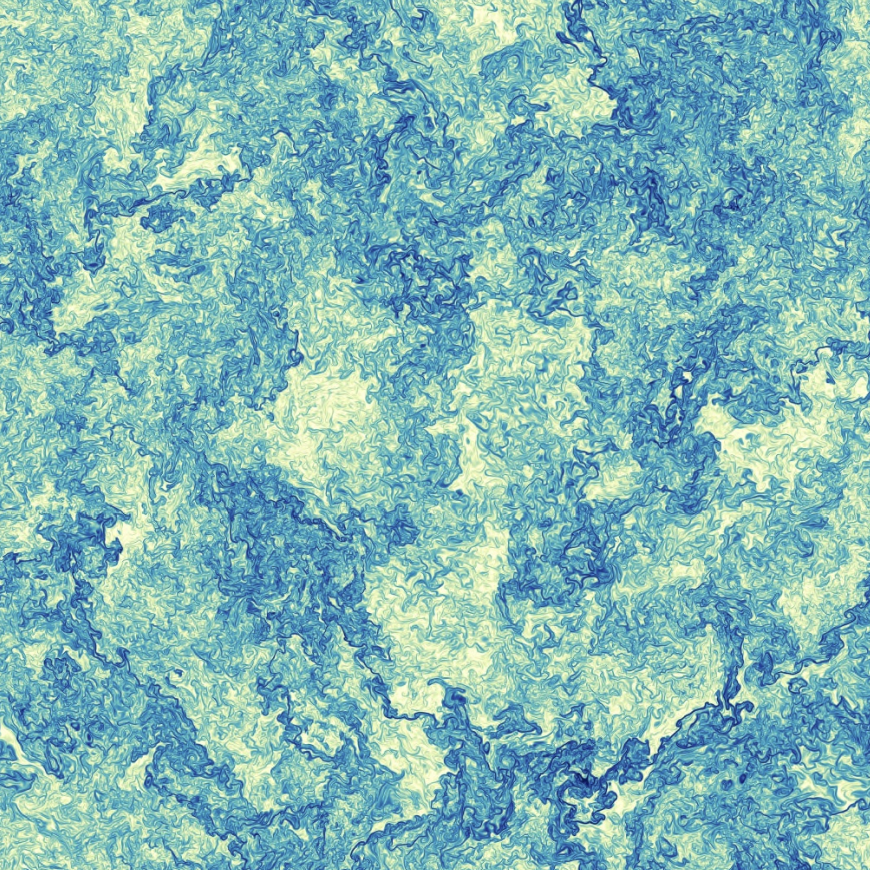}
  \end{center}
  \caption{Advected scalar by Navier-Stokes  and Kraichnan velocities.}\label{scalfig}
\end{figure}

\begin{figure}[h!]
  \begin{center}
    \includegraphics[width= 0.3\textwidth]{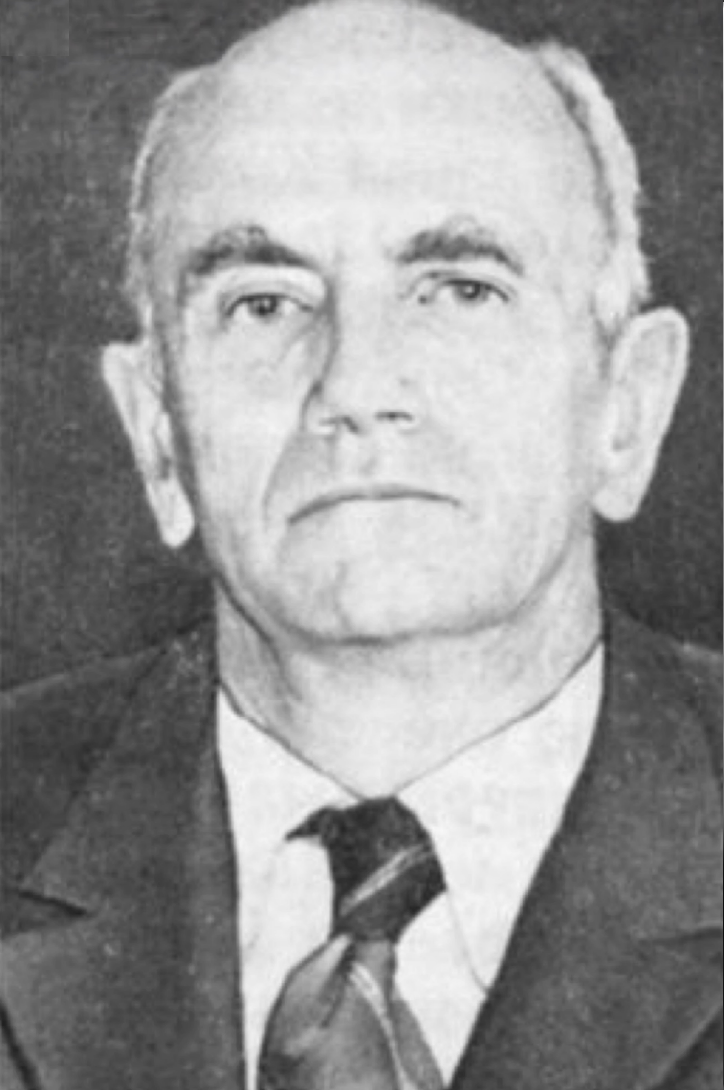}
       \hspace{2mm} 
       {
        \includegraphics[width= 0.3\textwidth]{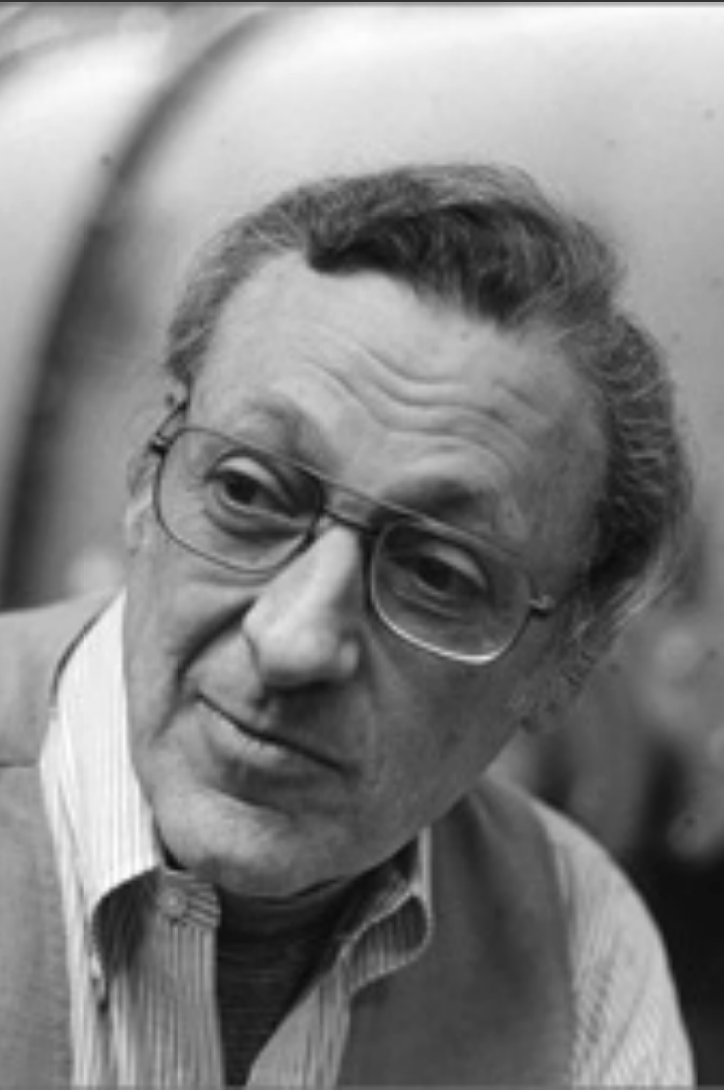}}
  \end{center}
  \caption{Alexander Obukhov and Stanley Corrsin }
\end{figure}

Obukhov \cite{Ob49}, followed by Corrsin \cite{C51}, developed a theory for passive scalar turbulence.\footnote{Obukhov and Corrsin studied the case where the advecting velocity itself was turbulent.  A complementary regime where the advecting velocity is smooth was studied first by Batchelor - the main phenomenon is \emph{mixing}.   There is a huge literature in this regime, both on the finite time dynamics and statistically steady state of the scalar.   I just mention some results \cite{CL, EJ22, CDG} and references therein, and highlight the great recent works of Bedrossian, Blumenthal and Punshon-Smith \cite{BBPS1, BBPS2, BBPS3}.  } Obukhov's motivation was to understand the entropy (and therefore temperature) cascade in the incompressible limit.   Both derived Onsager-type predictions on the requisite degree of singularity required to see anomalous diffusion in this context. The result is, roughly, that if $\sigma$ represents the fractional regularity of the velocity, the scalar cannot have regularity $\beta$ greater than $\frac{1-\sigma}{2}$. This was made rigorous by Eyink \cite{Eyink96} and works of  Constantin and Procaccia \cite{CP1,CP2}.  See the discussion in \cite{DEJI22}. Recently, there have been mathematical constructions of passive scalars exhibiting anomalous diffusion \cite{DEJI22,AV25,BSW23,HS}, even some that nearly exhibit the sharpness of the Obukhov--Corrsin theory in H\"{o}lder spaces \cite{colombo2023anomalous,elgindi2024norm}. As an example, one has
\begin{theorem}[Example of anomalous diffusion]
Fix $T>0$, $d\geq 2$, $\alpha\in (0,1)$ and $\theta_0\in L^\infty(\mathbb{T}^d)$.  There exists a divergence-free vector field $u\in L^1(0,T; C^\alpha)$ such that anomalous diffusion \eqref{anomalousdiff} occurs.  Moreover, $\theta=\lim_{\kappa\to 0} \theta^\kappa$ enjoys (nearly) the dual Obukhov--Corrsin regularity $\theta\in L^\infty(0,T; C^\beta)$ for every $\beta<\frac{1}{2} (1-\alpha)$.
\end{theorem}
Observations indicate, however, that typically the scalar field will be highly intermittent (even more so than the fluid velocity).  This motivates obtaining constraints on intermittency of the scalar field, much the same way we did for the fluid velocity. 
In fact, all the analysis we preformed for the Euler equations can be carried over to the setting of passive scalars and gives precise constraints on the relation between intermittency and dissipation measure.  The main implications of our framework in this situation are

\begin{theorem}[Intermittency and anomalous diffusion]
 Let $u\in L^p_{x,t}$ be a given vector field for some $p\in [1,\infty]$. Let $\theta\in L^q_{x,t}$ be a weak solution to \eqref{T} with $\frac{1}{p}+\frac{2}{q}\leq 1$ and with a  local dissipation measure $\chi[\theta]$ with a non-trivial singular part concentrated on a space-time set $S$ with $\dim_{\mathcal H} S=\gamma$. For all such $p$ and $q$  for which there exist $\sigma_p, \beta_{q} \in (0,1)$ such that $u\in L^p_t B^{\sigma_p}_{p,\infty}$ and $\theta\in L^q_t B^{\beta_q}_{q,\infty}$, it must hold 
            $$
            \frac{2\beta_q}{1-\sigma_p}\leq 1- \frac{p(q-2)-q}{pq} (d+1-\gamma).
            $$
\end{theorem}

The proof of this result is given in \cite{DDII25}.  It follows from a similar formula for the scalar anomaly as in Proposition \ref{P:decomposition_NS}.

This result immediately recovers the Obukhov--Corrsin and Eyink bounds, which say that non-trivial dissipation requires  $\beta_s\leq \frac{1-\sigma_p}{2}$ for all $p\in [1,\infty]$.
If more is known about the dissipation measure, the above result gives a substantial refinement.  Numerical work on scalar advection by three dimensional turbulent velocity fields \cite{ISS18} suggests that exponents saturate $\zeta_q^\theta:=q\beta_q \to 1.2$ as $q\rightarrow \infty$, approximately, thus quite far from being monofractal. See Figure \ref{scalanom}. We believe that constructions in the spirit of  \cite{DEJI22,AV25,colombo2023anomalous,elgindi2024norm,BSW23,HS} could be made to show the sharpness of this intermittent Obukhov--Corrsin theory. See, in particular, the recent work \cite{HS2}.

\begin{figure}[h!]
  \begin{center}
    \includegraphics[width=0.4\textwidth]{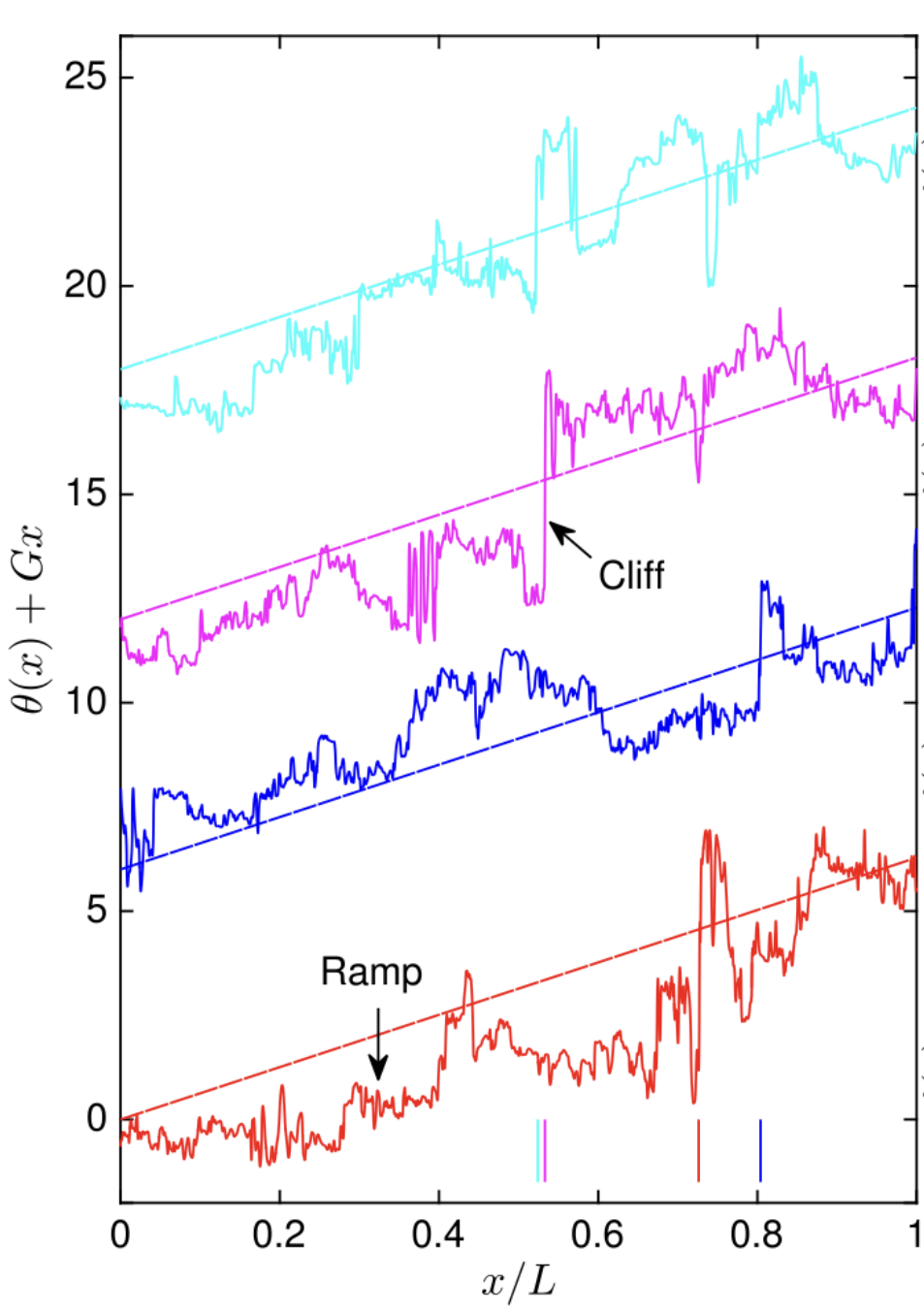}
       \hspace{2mm} 
       {
        \includegraphics[width=0.4\textwidth]{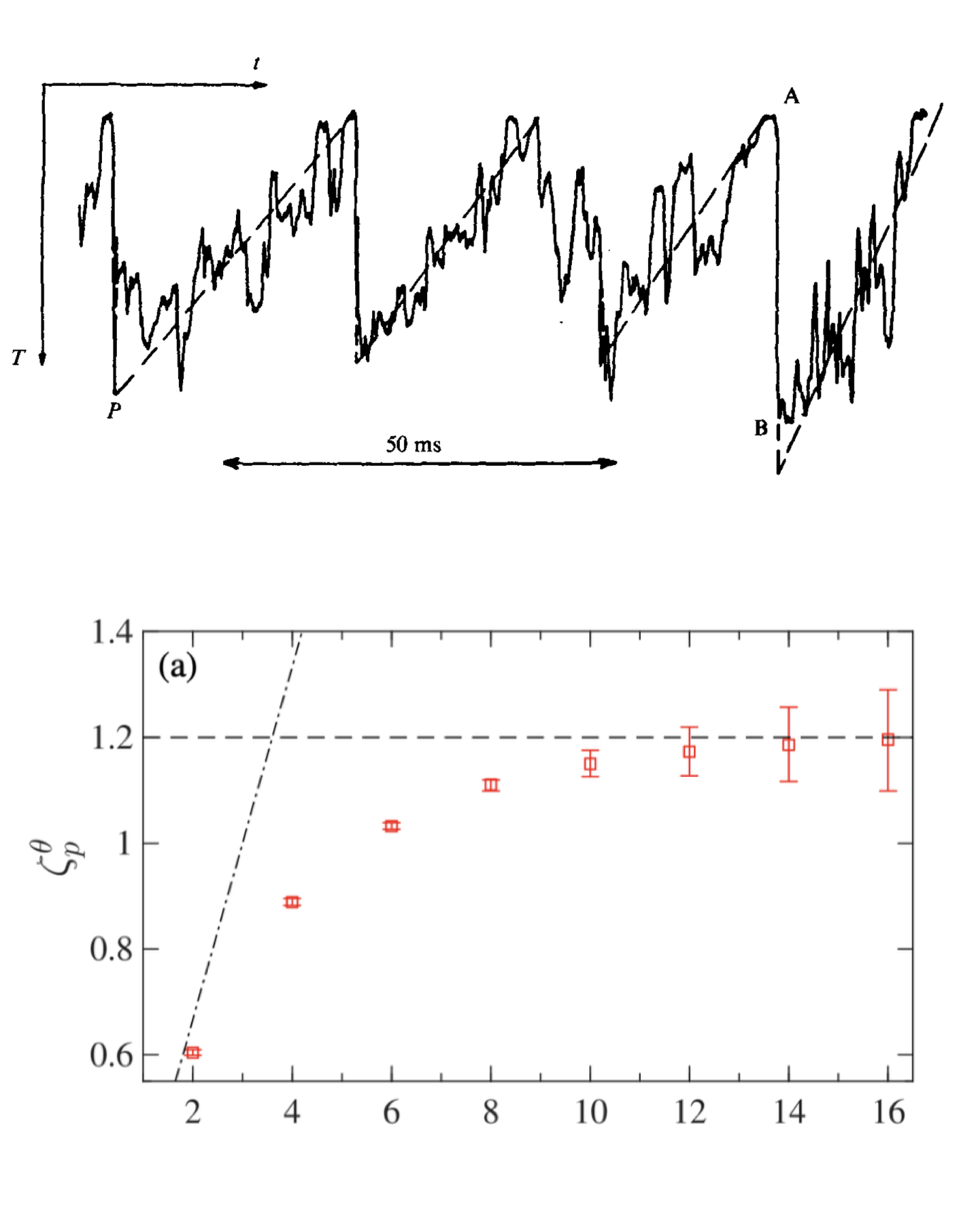}}
  \end{center}
  \caption{The structure of a scalar field in turbulence \cite{SM86,ISS18}. }\label{scalanom}
\end{figure}

An approach which is alternate (and perhaps opposite) to making specific constructions is to study the behavior of the scalar in a random field. Robert Kraichnan  forwarded this idea in his celebrated work \cite{Kr68}, see Figure \ref{kraichfig}.
Specifically, consider the stochastic advection by a  Gaussian random field $u:=u(t,x): \R_+ \times \mathbb{R}^d \to \mathbb{R}^d$, $d \geq 2$, which is spatially colored and white-in-time correlated:%
\begin{equation} \label{eq:Kraichnan} 
\partial_t   \theta_t^\kappa +    u \circ \nabla \theta_t^\kappa = \kappa \Delta \theta_t^\kappa,
\end{equation}
where the symbol $\circ$ denotes that we are interpreting the stochastic integral in the Stratonovich sense.  See \cite{FGV}. The velocity field, being Gaussian, is completely prescribed by its mean (taken zero for simplicity) and covariance
\begin{align*}
    \mathbb{E}[ u(t,x) \otimes u(t',x') ] = (t \wedge t') C(x,x') ,
 \quad
 C : \R^d \times \R^d \to \R^{d \times d}.
\end{align*}
To simplify the situation, one can consider the velocity field to be spatially  \textit{homogeneous}, i.e. $C(x,x')= C(x-x')$, and  \textit{non-degenerate}, i.e. $C(0)$ is a positive definite matrix.  We also assume $u$ is incompressible, although this is not really necessary for many of the results we will mention. When $u$ is spatially smooth the analysis of \eqref{eq:Kraichnan} is classical \cite{Kunita,L85}. However, we would like a model for turbulent transport which requires $u$ be spatially rough: just shy of  $\alpha$--H\"{o}lder in space for $\alpha\in (0,1)$ and not better. This amounts to considering a covariance matrix with the following small-distance asymptotic behavior
\be\label{covassum}
Q(z) := \langle |u(x)-u(x+z)|^2\rangle = 2(C(0) - C(z)) =O( | z |^{2 \alpha}).
\ee
		\begin{figure}[h!]
		\centering
			\includegraphics[width=0.9\textwidth]{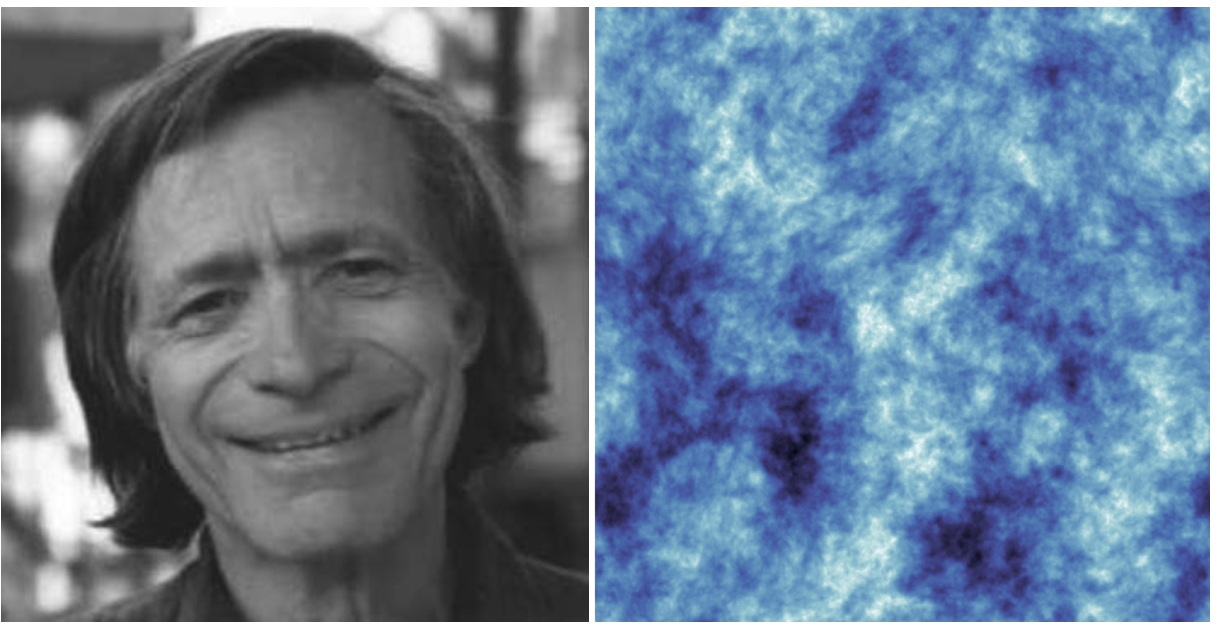} 
					\caption{R. Kraichnan and a realization of his model velocity.} \label{kraichfig}
	\end{figure}

This model and the phenomena associated to it have been studied in great detail by physicists.  I reference to the excellent review \cite{FGV}, and report here only some recent results that make more precise some of these insights. First is the quality of anomalous diffusion, studied by Rowan \cite{R1} and after in \cite{DGP25}
\begin{theorem}[Anomalous Diffusion in the Kraichnan model]\label{adthm}
Let $u$ be a divergence-free Kraichnan velocity with regularity exponent $\alpha\in (0,1)$ in $d$ dimensions.
There exists a constant $C=C(\alpha, d)$ such that 
\be
\mathbb{E} \|\theta^\kappa_t\|_{L^2}^2\lesssim e^{-t/C}\|\theta_0\|_{L^2}^2,
\ee
uniformly in $\kappa\in (0,1]$.
Moreover, the dissipation happens continuously in time.
\end{theorem}
Thus, anomalous dissipation is a general phenomenon in this model. In \cite{DGP25} we study also the local energy balance for \eqref{eq:Kraichnan}, showing that the solution $\theta=\lim_{\kappa\to 0} \theta^\kappa$ (which is unique, for reasons special to the Kraichnan model) admits a local energy balance that is governed by the stochastic PDE
\begin{equation*}
\partial_t  |\theta_t|^2   +u \circ \nabla  | \theta_t |^2  = - \chi[\theta_t],
\end{equation*}
where $\chi[\theta]$ is a non-negative distribution.  This is analogous to the Duchon-Robert distribution for the Euler equations \cite{DR00} and describes exactly how fluctuations of $\theta$ dissipate in space and time.
It is shown that
$
\chi[\theta](\rmd t, \rmd x) = \lim_{\kappa\to 0} \kappa |\nabla \theta_t^\kappa(x)|^2 \rmd t \rmd  x,
$
where the equality is interpreted in the sense of distributions, and the limit in probability.  In view of this identification $\chi[\theta]$ is, in fact, almost surely a non-negative Radon measure.  Assuming that the velocity is divergence-free, we can derive an explicit formula for  $\mathbb{E}\chi[\theta]$.

\begin{theorem} \label{thm:regularity_Kraichnan}
Let $u$ be a divergence-free Kraichnan velocity with regularity exponent $\alpha\in (0,1)$ in $d$ dimensions.
Then  for every non-zero $\theta_0 \in L^2_x$ and $t \geq 0$,  the unique solution of \eqref{eq:Kraichnan} satisfies, for some positive constant ${c}={c}(d,\alpha)$, that
\begin{align} \label{idealdefect}
\mathbb{E}\chi[\theta](\rmd t, \rmd x) ={c} \lim_{\ell \rightarrow 0}
\oint_{\mathbb{S}^{d-1}} 
 \frac{\mathbb{E}| \delta_{\ell \hat{z}} \theta_t(x)|^2}{\ell^{2- 2\alpha}} \sigma ( \mathrm{d} \hat{z}) \rmd  t \rmd x.
\end{align}
\end{theorem}
Formula \eqref{idealdefect} is a manifestation of Yaglom's law for scalar turbulence \cite{MY}, and can be regarded as an analogue of Duchon-Robert's result for the Kraichnan model. Upon integration,  \eqref{idealdefect} implies (from the energy balance) that
\begin{align} \label{eq:regularization_besov_intro}
{c} \lim_{\ell \rightarrow 0}
\int_0^t
\oint_{\mathbb{S}^{d-1}} 
 \frac{\mathbb{E}\| \delta_{\ell \hat{z}} \theta_s \|^2_{L^2_x}}{\ell^{2- 2\alpha}} \sigma ( \mathrm{d} \hat{z}) \rmd  s
= \|\theta_0\|_{L^2_x}^2 - \mathbb{E} \|\theta_t\|_{L^2_x}^2
\in (0,\infty).
\end{align}
From \eqref{idealdefect}, one immediately deduces that
\begin{align}\label{OCthkr}
\int_0^T\mathbb{E} \| \theta_t \|_{B^{(1 - \alpha)+}_{2,\infty}}^2\rmd  t < \infty \quad \Rightarrow\quad
    \mathbb{E} [\| \theta_t \|_{L^2_x}^2] 
=
\mathbb{E} [\| \theta_s \|_{L^2_x}^2] \quad \text{ for all } s, t \in [0,T].
\end{align}
The threshold regularity exponent $1-\alpha$ corresponds to the Obukhov-Corrsin regularity theory \cite{Ob49, C51} for the Kraichnan model, representing the minimal requirement for anomalous dissipation to take place.  In fact solutions possess \emph{exactly the borderline regularity} allowing dissipation of mean energy. Namely, the "Yaglom law"  \eqref{eq:regularization_besov_intro} shows (see also \cite{R2})
\be\label{regu}
\int_0^T\mathbb{E} \| \theta_t \|_{B^{1 - \alpha}_{2,\infty}}^2\rmd  t  <\infty, 
\ee
with the Besov-like spaces introduced in Lemma \ref{Slem}. In view of Theorem \ref{adthm} and \eqref{OCthkr}, the solution has sharply $1-\alpha$ derivatives in this $L^2$ based Besov space, at least in the mean sense, thereby establishing the analogue of Conjecture \ref{conj45} for the Kraichnan model.   This kind of regularization was established first in the work \cite{GGM24} by different means, and it shows the \emph{sharpness} of the Obukhov-Corrsin regularity theory at the level of $L^2$, while allowing for intermittency \cite{Kr94}.   In fact, it is remarkable that in the Kraichnan model, an asymptotic approximation of the structure function scaling exponents $S_p^\theta(\ell):= \mathbb{E} |\delta_\ell \theta|^p \sim \ell^{\zeta_p^\theta}$ have been computed \cite{bernard1996anomalous}
\be
\zeta_p^\theta = p(1-\alpha) - \frac{p(p-2)}{d+2} \alpha +\mathcal{O}(\alpha^2).
\ee
Intermittency vanishes as $\alpha\to 0$, where the scalar field is expected to be Gaussian.  It also vanishes at $p=2$, where the exact Yaglom law holds. 
It would be very interesting to establish a rigorous  understanding of these results.

\begin{proof}[Idea behind proof of Yaglom's law \eqref{eq:regularization_besov_intro}]
Let us present a heuristic computation to give an intuition about the anomalous regularization mechanism. Let $\theta^\kappa$ be a solution to \eqref{eq:Kraichnan} and consider its second-order structure function $S^\kappa_t(z)=\EE[\| \delta_z {\theta}^\kappa_t\|_{L^2_x}^2]$.
For simplicity, let the velocity be statistically self-similar, so that  $Q(z) \propto |z|^{2\alpha}$.
Applying It\^o's formula one can see that $S^\kappa$ formally solves the linear SPDE
\begin{align*}
    \partial_t S^\kappa_t(z)
    & =\nabla\cdot (Q \nabla S^\kappa_t)(z)-2\kappa \EE[\|\nabla\delta_z\theta^\kappa_t\|_{L^2_x}^2],
\end{align*}
since $Q$ is divergence free. Integrating over the ball $B_\ell=\{|z|\leq \ell\}$, one finds
\begin{align*}
    \frac{\rmd }{\rmd  t} \fint_{B_\ell} S^\kappa_t(z) \rmd  z
    = \fint_{B_\ell}\nabla\cdot (Q \nabla S^\kappa_t)(z) \rmd  z -2\kappa \fint_{B_\ell} \EE[\|\nabla\delta_z\theta^\kappa_t\|_{L^2_x}^2] \rmd  z.
\end{align*}
Applying the Divergence Theorem and exact self-similarity $Q(z) \propto |z|^{2\alpha}$, one finds
\begin{align*}
\fint_{B_\ell}\nabla\cdot (Q \nabla S^\kappa_t)(z) \rmd  z
    &\sim 
    \ell^{-d} \int_{\partial B_\ell} \hat z\cdot Q(z)\nabla S^\kappa_t(z) \sigma(\rmd z)
    \sim 
    \ell^{2\alpha-d} \int_{\partial B_\ell} \hat z\cdot \nabla S^\kappa_t(z) \sigma(\rmd z)\\
    & \sim \ell^{2\alpha-d} \frac{\rmd}{\rmd \ell} \int_{\partial B_\ell} S^\kappa_t(\hat y) \sigma(\rmd \hat y)\\
    & \sim
    \ell^{2\alpha-1} \frac{\rmd}{\rmd \ell} \fint_{\partial B_\ell} S^\kappa_t(\hat y) \sigma(\rmd \hat y)
    +(d-1)\ell^{2\alpha-2} \fint_{\partial B_\ell} S^\kappa_t(\hat y) \sigma(\rmd \hat y).
\end{align*}
As the last term on the right-hand side is positive, we can obtain a lower bound by dropping it.
Rearranging the above relations and 
integrating in time we obtain
\begin{align*}
    \frac{\rmd}{\rmd \ell}\int_0^T \oint_{\mathbb{S}^{d-1}} S^\kappa_t(\ell \hat y) \sigma(\rmd \hat y) \rmd  t
    &\lesssim 
    \ell^{1-2\alpha} \Big[ \fint_{B_\ell} S^\kappa_T(z) \rmd  z - \fint_{B_\ell} S^\kappa_0(z) \rmd  z\Big] + 2\kappa \int_0^T I^2_t \rmd  t \\
    &\lesssim \ell^{1-2\alpha} \sup_{z\in\R^d} \Big( S^\kappa_T(z) + 2\kappa \fint_{B_\ell} \int_0^T \EE[\|\nabla\delta_z\theta^\kappa_t\|_{L^2_x}^2] \rmd  t \rmd  z\Big).
\end{align*}
Integrating in $\ell$ and using triangular inequality, we arrive at
\begin{align*}
    \int_0^T \oint_{\mathbb{S}^{d-1}} \frac{\EE[\| \delta_{\ell \hat{y}}  \theta^\kappa_t\|_{L^2_x}^2]}{\ell^{2-2\alpha}} \sigma(\rmd \hat y)
    \lesssim \EE[\|\theta^\kappa_T\|_{L^2_x}^2] + 2\kappa \int_0^T \| \nabla {\theta}^\kappa_t\|_{L^2_x}^2] \rmd  t
    = \|\theta_0\|_{L^2_x}^2,
\end{align*}
where in the last step we used the diffusive energy balance.
The resulting estimate is uniform in $\kappa$, thus allowing to take the limit as $\kappa\to 0^+$ and obtain a corresponding statement for $\theta$ solving \eqref{eq:Kraichnan}.
\end{proof}

In summary, in the case of the Kraichnan model we have found that \emph{anomalous diffusion is a mechanism for anomalous regularization}.  This is indeed the intuition behind Conjecture \ref{conj45} for hydrodynamic turbulence.
We now ask the question: what is the mechanism for anomalous dissipation? For  scalar transport, we can answer this. It relates to non-uniqueness of particle trajectories of the velocity field

\begin{theorem}[Anomalous Diffusion and Spontaneous Stochasticity]\label{adthm}
Let $u$ be a divergence-free velocity. Suppose a passive scalar driven by $u$ exhibits anomalous dissipation. Then, the integral curves of $u$ (backwards in time) are non-unique for a positive measure set of initial conditions.
\end{theorem}
For a more precise statement and its proof, see \cite{DE17a,DE17b} and Rowan \cite{R1}.  The discovery of the general phenomenon, however, is due to Bernard, Gaw\c{e}dzki and Kupiainen \cite{bernard1996anomalous1}. The basic idea is this.  If one introduces stochastic trajectories 
\be\label{stochtraj}
\rmd X_t(a) = u(X_t(a),t) \rmd t+ \sqrt{2\kappa} \rmd W_t,
\ee
then, by the Feynman-Kac representation, the scalar field can be constructed via the inverse flow as an average
\be
\theta_t^\kappa(x) = \mathbb{E} [\theta_0 (X_t^{-1}(x))].
\ee
The scalar dissipation is given by a \emph{Lagrangian fluctuation--dissipation relation} \cite{CCMV04, DE17a, DE17b}
\be
\kappa \int_0^T \|\nabla \theta_t^\kappa\|_{L^2}^2 \rmd t = \int_M {\rm Var}  [\theta_0 (X_t^{-1}(x))]\ \rmd x.
\ee
This relation makes clear the connection between non-uniqueness and anomalous diffusion.  If anomalous diffusion occurs, the variance of the (fixed) initial conditions sampled by the backwards flow must remain non-trivial as $\kappa\to 0$ for a positive measure set of $x\in M$.  This can only happen if the ideal particle trajectories passing through a positive measure of such points are non-unique. 

In the Kraichnan model where the phenomenon was discovered, this process can be more explicitly studied.  Indeed, for the homogeneous isotropic model, one can show that $\rho_t(a,b) := |X_t(a) - X_t(b)|$ shares a law of the one-dimensional drift diffusion process
\be
\rmd \rho_t = b(\rho_t) \rmd t + \sqrt{2 \sigma(\rho_t)} \rmd B_t + \sqrt{2\kappa} \rmd W_t. \qquad \rho_0= |a-b|.
\ee
Here $B_t$ is the noise coming from the Kraichnan velocity field and $W_t$ is an independent noise arising from modeling dissipation as in \eqref{stochtraj}.  The following occurs
\begin{itemize}
\item  if the velocity is Lipschitz in space, then $\rho=0$ is an absorbing point for the diffusion. This implies particle trajectories are unique.
\item if the velocity is $\alpha$--H\"older, then $b(\rho),\sigma(\rho)\sim \rho^{\alpha}$ and the process $\rho_t$ is instantaneously reflected at $\rho=0$.  This implies non-uniqueness of trajectories. 
\end{itemize}
The latter case of non-uniqueness has been named \emph{spontaneous stochasticity} \cite{CGV03}.
These results were rigorously proved by Le Jan and Raimond \cite{LR1,LR2}, and can be visualized in Figure \ref{kraichvis} which shows an ensemble of noisy trajectories with small $\kappa$ in a Lipschitz and H\"older velocity, respectively.

\begin{figure}[h!]
  \begin{center}
    \includegraphics[width=0.47\textwidth]{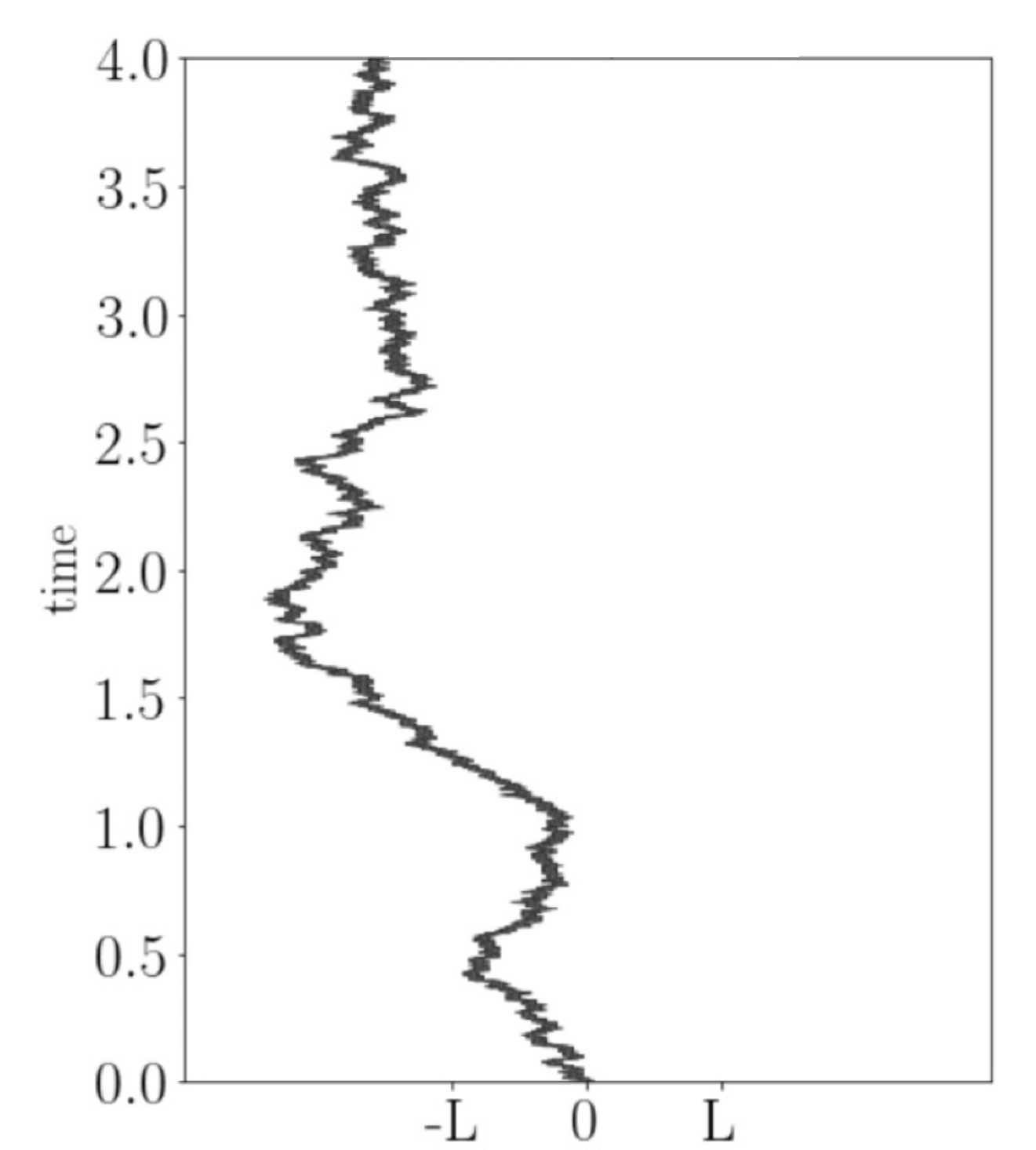}
       \hspace{0mm} 
       {
        \includegraphics[width=0.5\textwidth]{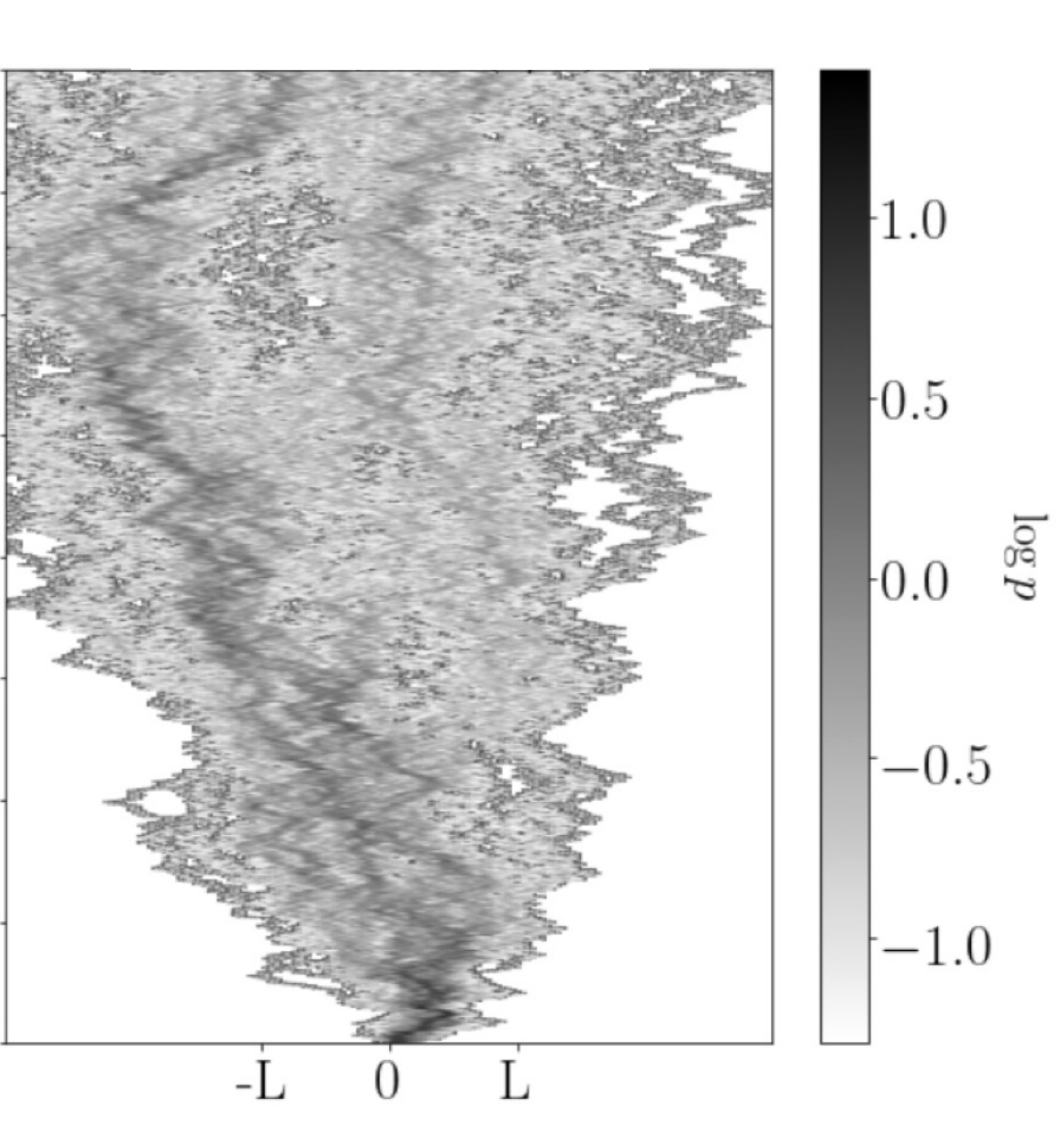}}
  \end{center}
  \caption{Determinism/spontaneous stochasticity of particle trajectories in a smooth/rough Kraichnan velocity field.} \label{kraichvis}
\end{figure}

\section{Taylor, Richardson and the motion of particles}\label{particles}

We saw in the last section that a property of Lagrangian (or particle) trajectories:
\be
\dot{X}_t(a) =u(X_t(a),t) \qquad X_0(a)=a,
\ee
 namely spontaneous stochasticity, is responsible for the phenomenon of anomalous dissipation in some model problems (both for Burger \cite{ED15} and passive scalars \cite{bernard1996anomalous1,DE17a,DE17b}). 
 
 Should we expect this phenomenon in real turbulence?   One of the first investigations on particles motion in turbulence is due to G.I. Taylor \cite{T21} (1921), who studied how a single particle diffuses at long time.  Taylor predicted that a single particle acts diffusively at long times with effective diffusivity $D$, e.g. with the Lagrangian velocity $v_t(a):= \dot{X}_t(a)$, 
 \be
\frac{\langle | X_t(a) - a|^2\rangle}{2t}  \ \stackrel{t \to \infty}{\longrightarrow} \  D \quad \text{where}\quad D := \int_0^\infty\langle v_\tau \cdot v_0\rangle\   \rmd \tau,
 \ee
 provided the latter limit exists.
 Here $\langle\cdot \rangle$ denotes average over the label $a\in \mathbb{T}^d$, and possibly over an ensemble of velocity fields.  Taylor's result holds under the assumption that the Lagrangian velocity  is stationary in time after averaging, namely  $\langle v_t(a)\cdot v_{t'}(a)\rangle = \langle v_{t-t'}(a) \cdot v_{0}(a)\rangle$, which should be an integrable function of $\tau=t-t'$.  Stationarity of the Lagrangian velocity holds, for instance, provided the Eulerian velocity field is statistical stationary and homogeneous. See \cite[\S 7.1]{TL72} and \cite{KP97} for an example of an early rigorous statement.  Thus, we take it that a single particle looks like Brownian motion at long times, albeit with a diffusivity which is only implicitly predicted.
\begin{figure}[h!]
  \begin{center}
    \includegraphics[width=0.38\textwidth]{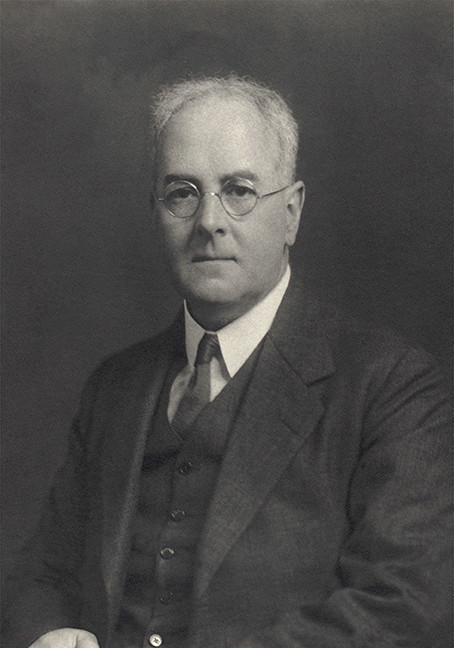}
     \hspace{5mm} 
        \includegraphics[width=0.37\textwidth]{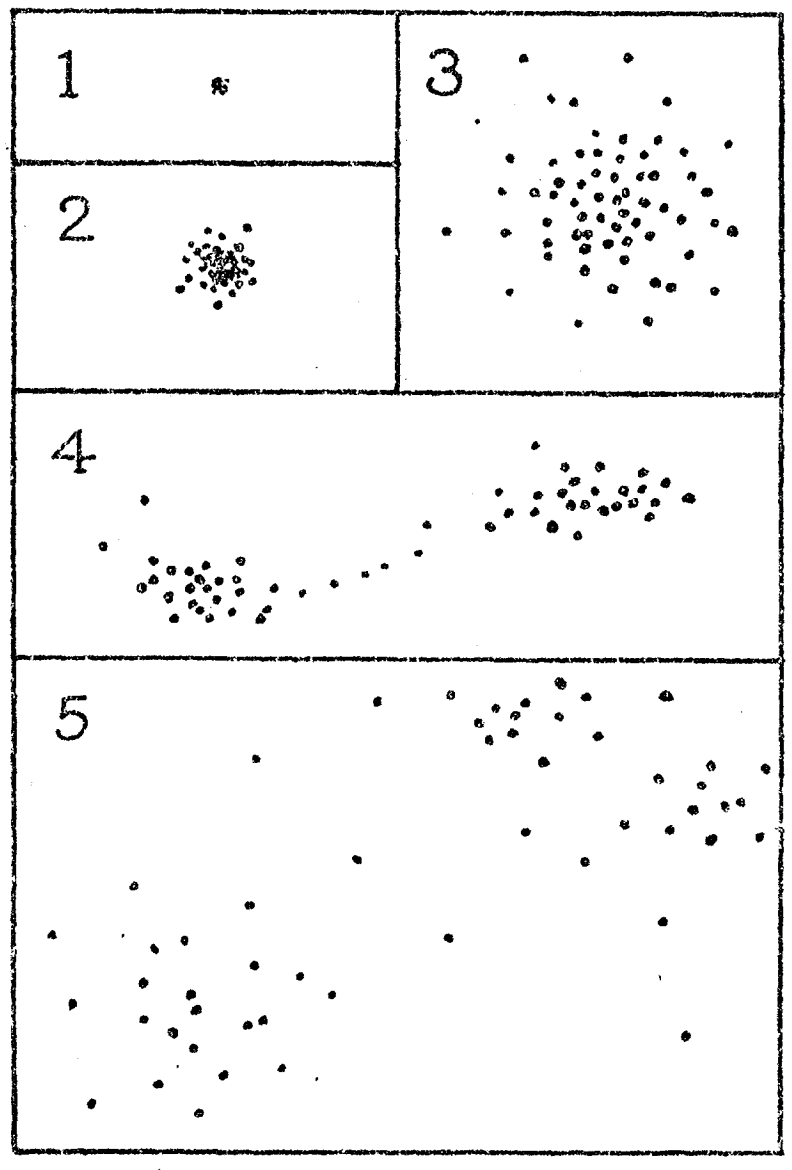}
  \end{center}
  \caption{Lewis Fry Richardson and his particle dispersion \cite{R26}}
\end{figure}
What about multiple particles?  In 1926, shortly after the work of Taylor, Lewis Fry Richardson investigated the dispersion of two-particles by a turbulent flow.  Based on theoretical considerations and weather balloon data, Richardson predicted a $t^3$--law:
\begin{myshade}
\be\label{tcubed}
\langle |X_t(a+r)- X_t(a)|^2\rangle  \sim G \langle \varepsilon\rangle t^3,
\ee
\end{myshade}
for some constant $G$ and 
where $|r|\gg \ell_\nu$, namely the particles are separated in the inertial range.  The striking thing about Richardson's prediction is that the right hand side is independent of the initial separation of the particles!  Indeed, the inertial range extends to all scales as $\nu \to 0$, so that $r=0$ can be taken.  In that case, Richardson predicts an ensemble of particles burst from the same starting point and have mean-squared displacement growing cubically in time, a rate which is faster than simple ballistic separation of particles experiencing a constant relative velocity.  In short, already in 1926, Richardson foresaw the shadow of spontaneous stochasticity.

\begin{figure}[h!]
  \begin{center}
    \includegraphics[width=0.4\textwidth]{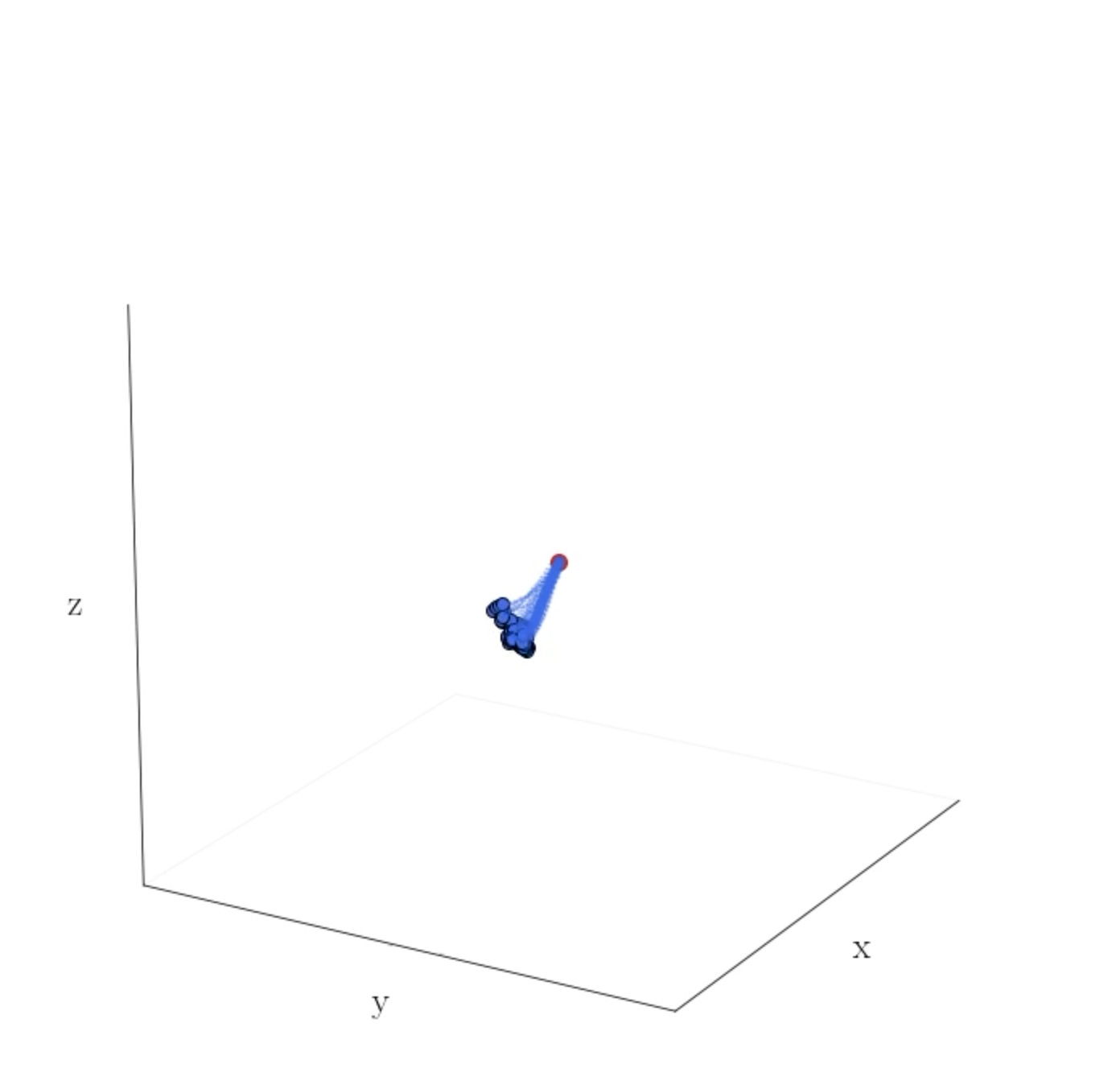}
       {
        \includegraphics[width=0.4\textwidth]{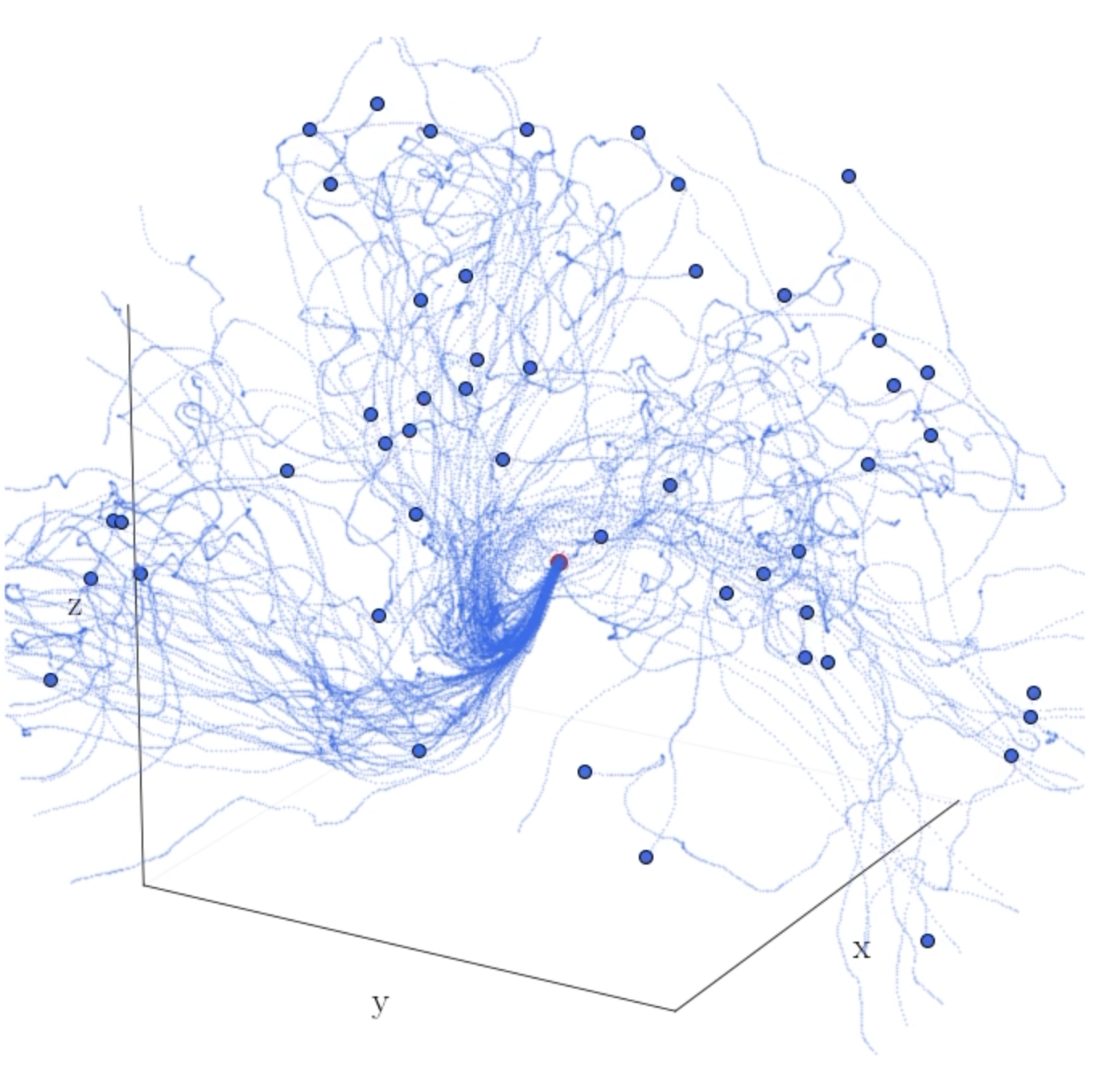}}
  \end{center}
  \caption{Explosive separation of particles in numerical simulation of homogeneous isotropic turbulence \cite{Lcomp}.}
\end{figure}

Although Richardson arrived at his result by a different theoretical consideration, let us quickly give a rigorous bound on such dispersion assuming some regularity on the velocity field that is coherent with \eqref{tcubed} assuming Kolmogorov's 1941 picture of turbulence. 
\begin{proposition}[Bound on splitting of trajectories]\label{boundtraj}
Suppose $u\in L^\infty_t C^\alpha_x$. Then, \emph{any} trajectory $\dot{X}_t=u(X_t,t)$ satisfies
\be
|X_t(a+r) - X_t(a)| \leq \left(|r|^{1-\alpha} + (1-\alpha) [u]_{L^\infty_t C^\alpha_x} t \right)^{\frac{1}{1-\alpha}}   \qquad \forall |r|\ \geq 0.
\ee
\end{proposition}
\begin{proof}[Proof of Proposition \ref{boundtraj}]
Let $\rho_t(r;a) :=|X_t(a+r) - X_t(a)|$.  By reverse triangle inequality, as well as the H\"{o}lder assumption on the velocity field $u$, we have
\be
\dot{\rho}_t\leq  [u]_{L^\infty_t C^\alpha_x} \rho_t^\alpha, \qquad \rho_0 = |r|.
\ee
Integrating this differential inequality gives the desired result.
\end{proof}
Note that, for long times the bound becomes independent of $|r|$.  Its asymptotic form takes the shape of a bound on the deviation of non-unique trajectories of a H\"{o}lder field $u$ passing through the same point $a\in M$:
\be
|X_t^{(1)} (a) - X_t^{(2)}(a)| \leq \left((1-\alpha) [u]_{L^\infty_t C^\alpha_x} t \right)^{\frac{1}{1-\alpha}}.
\ee
Of course, there are simple examples saturating this upper bound (take $u(x)={\rm sgn}(x)  |x|^\alpha$ in one dimension), showing Proposition \ref{boundtraj} to be sharp.  See \cite{ED15b,EB20} for a realization of this simple example in quantum systems, and a renormalization group perspective. The works \cite{DM21,DMR24} give higher dimensional examples and investigates when a unique measure valued solution might be obtained. The works \cite{M1,M2,M3, TBM} concern the same behavior for infinite dimensional models, where these features robustly appear.

Recalling K41 phenomenology, $\alpha=1/3$ and  $[u]_{L^\infty_t C^{1/3}_x}= (\langle \ve \rangle \ell_I )^{1/3}$ (by dimensional considerations, taking e.g. $\ell_I= {\rm diam}(\mathbb{T}^d)$).  Since this field is ideally rough, understood as arising in the $\nu\to 0$ limit, we may take $r=0$ so that particles start at the same point.  The bound is then exactly consistent with the behavior \eqref{tcubed}.

Although Richardson's prediction has proved difficult to conclusively observe, impressive partial confirmations have been made experimentally \cite{OM, Betal06} and numerically \cite{BBetal, BS02, BD14, BHB13}.  The following Figure \ref{richfig}, taken from \cite{BHB13}, represents essentially the state of the art in our observations of Richardson's law from numerical simulation data.  The $t^3$--law is observed, as is the "forgetting" of the initial separations as expected from Richardson's theory as well as the simple-minded bound in Proposition \ref{boundtraj}.

\begin{figure}[h!]
  \begin{center}
    \includegraphics[width=1\textwidth]{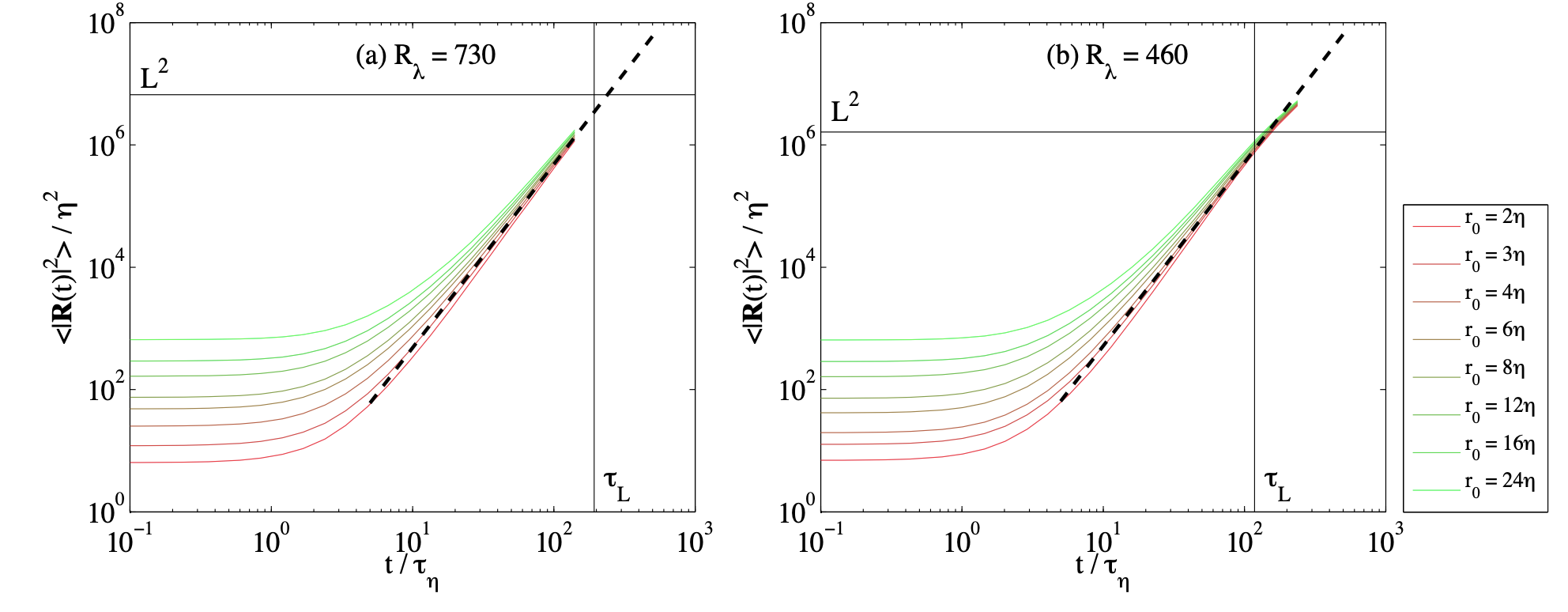}
  \end{center}
  \caption{Time-evolution of the mean-squared dispersion for various initial separations $r_0$ from numerical simulation of homogeneous isotropic turbulence at two different Reynolds numbers \cite{BHB13}. The dashed line corresponds to the Richardson law.}\label{richfig}
\end{figure}

Another striking aspect of Richardson's prediction \eqref{tcubed} is that the net anomalous energy dissipation rate  $\langle \ve \rangle$ appears. Recall that for the Burgers equation and passive scalars both, non-uniqueness of trajectories gave an origin or explanation of the observed anomalies.  Richardson's law seems to do this as well, but it's theoretical foundation is lacking.  Does $\langle \ve \rangle$ appear  in his theory for purely dimensional reasons, and true particle dispersion is not directly related to dissipation?  

As it turns out, this aspect of Richardson's theory can be rigorously proved, although only for a very short time regime and not for the dispersion itself, but rather a relative forward/backwards dispersion. To state the result, we must introduce a regularization of the trajectories, namely the flow in a mollified field:
\be\label{molflow}
\frac{\rmd}{\rmd t} X_t^\ell(a) = \ol{u}_\ell ( X_t^\ell(a) ,t) , \qquad X_0^\ell(a) = a.
\ee
Fix an standard mollifier $\psi\in C_0^\infty(\mathbb{T}^d)$  with  ${\rm supp} (\psi)\subseteq B_1(0)$ and $\psi_R(r)  = R^{-d} \psi(r/R)$ and denote 
$
 \langle f(r) \rangle_R:= \int_{\mathbb{T}^d} f(r) \psi_R(r)dr
$
 for any $f\in L^1( \mathbb{T}^d)$.    The result is:
 
\begin{theorem}[Lagrangian formula for dissipation]\label{LAGthm}
Let $u\in L^3(0,T; L^3(\mathbb{T}^d))$ be a locally dissipative weak solution of incompressible Euler as in Theorem \ref{invthm}. Then, the dissipation measure $\ve[u]$  in the distributional equality
\be
 \partial_t \left(\tfrac{1}{2} |u|^2 \right) + {\rm div} \left((\tfrac{1}{2} |u|^2+p) u \right) = -\ve[u] + u \cdot f
\ee
admits the following representation in the sense of distributions:
\begin{align}\label{Lagrel}
\varepsilon[u]&= \lim_{R \to 0}\lim_{\ell \to 0}  \lim_{\tau\to 0}\frac{  \langle |\delta X_{t,t-\tau}^{\ell}(r;x)-r|^2\rangle_R-\langle |\delta X_{t,t+\tau}^{\ell}(r;x)-r|^2\rangle_R}{4\tau^3}.
\end{align}
where $\delta\bX_{s,t}^\ell(\br;\bx):= \bX_{s,t}^\ell(\bx+\br)-  \bX_{s,t}^\ell(\bx)$ is the Lagrangian deviation.
\end{theorem}

We won't give the proof of this result here, it can be found in \cite{D19}.  We just give the main idea behind it, which derives from the works of {Ott-Mann} and Gaw\c{e}dzki \cite{OM,FGV}, as well as  Jucha et. al. \cite{Jucha}.
Indeed, the \emph{Ott-Mann--Gaw\c{e}dzki} relation is sometimes  described as the "Lagrangian analog of the $4/5$--law". For  inertial range separations $r$, this relation states:
\begin{equation}\label{OttMann0}
\left.\frac{1}{2} \frac{\rmd }{\rmd \tau} \Big\langle |\delta_r \bv(\tau;\bx,t) |^2\Big\rangle\right|_{\tau=0} \simeq -2 \langle  \varepsilon \rangle
\end{equation}
with the Lagrangian velocity $v(\tau,x;t):=\bu(\bX_{t,t+\tau} (\bx),t+\tau)$ and $\delta_r  \bv(\tau;\bx,t):=  v(\tau,x+r;t)-v(\tau,x;t)$.  Standard derivations of the relationship assume spatial isotropy and the average must be either interpreted as over the spatial domain, or as a time/ensemble average provided the fields are homogenous (\cite{D19} gives a more robust interpreation/derivation holding for any weak solution).
With the Ott-Mann-Gaw\c{e}dzki relation in hand, the relative mean--squared dispersion of Lagrangian tracers for short-times can be calculated using only local (in time) Eulerian quantities by Taylor expansion:
\begin{align} \label{reldispFG}
\Big\langle | \delta \bX_{t,t+\tau}(\br; \bx)-&\br|^2\Big\rangle \approx   S_2^{\bu_t}(\br)  \tau^2 - 2   \langle  \varepsilon \rangle \tau^3 +{O} (\tau^4).
\end{align}
Inspecting the behavior of \eqref{reldispFG} under time reversal $\tau\to -\tau$, one sees that the ${O}(\tau^2)$ term is invariant whereas the ${O}(\tau^3)$ term changes sign.  Taking of the forwards and backward in time dispersion suggests (the global analogue of) the formula \eqref{Lagrel}.

Unfortunately, at high-Reynolds numbers, the realm of validity of the expansion \eqref{reldispFG} becomes vanishing small. In particular, the Taylor series expansion of the particle trajectories used to derive \eqref{reldispFG} is only guaranteed to converge in a neighborhood of times  that is vanishingly small as $\nu\to 0$ (although Frishman \& Falkovich \cite{Falk13,Frish14}  argue on theoretical grounds that a similar expansion for the difference may not). Thus, it is desirable to have an alternative Lagrangian measure of time--asymmetry that remains valid for arbitrarily large Reynolds numbers. This is what is accomplished in Theorem \ref{LAGthm} by introducing an (arbitrary) mollification of the velocity field.

\begin{remark}[Lagrangian arrow of time] Theorem \ref{LAGthm} predicts that in three-dimensional turbulence, where $\ve[u]$ is a non-negative and non-trivial measure, particles on average split faster backwards in time than they do forwards in time,  thereby establishing a Lagrangian "arrow of time".   Irreversibility of this nature has been the subject of numerous numerical/experimental investigations \cite{Saw05, Faber09, Xu16} and  this particular prediction along with the detailed relation \eqref{Lagrel} has been recently partially confirmed in physical experiment \cite{CDB22}. On the other hand, in two-dimensional turbulence (in some extreme regime)  the inverse-energy cascade can brings energy in from infinite frequency (pumped in by a infinitely small scale forcing), leading to negative $\ve[u]$, namely an anomalous input 
\cite{E96, Kr67, Leith68}.  In this setup, Theorem \ref{LAGthm} predicts that particles disperse faster forward in time than backwards.  See \cite{D19} for a more detailed discussion.

There is a simple minded mechanical explanation for all of this, in a system without explicit friction.  Namely, consider a system of two particles moving freely in space, which stick upon collision and form one new particle (a purely inelastic collision). Suppose $m_1$ and $m_2$ are the masses, and $v_1$ and $v_2$ the velocities before the collision.  Then, conservation of mass and momentum dictate that the mass  after collision is $m=m_1+m_2$ and the velocity $v$ is such that $mv= m_1 v_1 + m_2 v_2$. It is then easy to see that the kinetic energy of this combined particle $\frac{1}{2}mv^2$ is strictly less than the initial energy of the system $\frac{1}{2}m_1v_1^2+ \frac{1}{2}m_2v_2^2$, the deficit being $\triangle {\rm KE} = -\frac{m_1m_2}{2m} (v_1-v_2)^2  <0$.   Thus, in this simple model the direction of time in which particles coalesce corresponds to the direction in which energy is lost.  This bares some similarity to the fluid case discussed above, and the discrepancy between two and three dimensions.
\begin{figure}[h!]
  \begin{center}
    \includegraphics[width=0.5\textwidth]{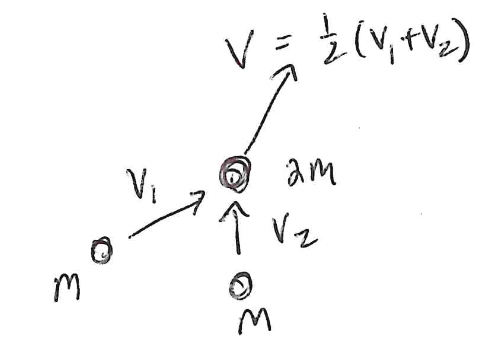}
  \end{center}
  \caption{Cartoon of the inelastic collision of particles of the same mass.  Forwards in time the coalesce, backwards they split apart.  The interaction loses energy in the direction of time of coalescence. }\label{stickfig}
\end{figure}
\end{remark}

We give one final result on particles in (H\"{o}lder-continuous) turbulent velocity fields, concerning their smoothness in time.
\begin{theorem}[Time-regularity of trajectories]\label{timeregthm}
Let $u\in L^\infty_tC^\alpha_{x}$ be a weak solution of the Euler equations. For each $a\in M$, consider \underline{any} solution of
\be\label{ode}
\frac{\rmd}{\rmd t} X_t(a) ={u} ( X_t(a) ,t) , \qquad X_0(a) = a.
\ee
Define $v(t,a) :=\frac{\rmd}{\rmd t} X_t(a)$. Then $v \in L^\infty_aC^{\frac{\alpha}{1-\alpha}}_t$. 
\end{theorem} 
\begin{remark}
If $\alpha=\frac{1}{3}$, then $\frac{\alpha}{1-\alpha}= \frac{1}{2}$, in qualitative agreement with the prediction of Landau and Lifshitz \cite{LL}, based on K41 phenomenology, that 
\begin{equation}
\langle |v(\cdot+\tau)-v(\cdot)|^2\rangle \sim \langle \ve \rangle \tau.
\end{equation}
As $\alpha\to 1$, the trajectories become $C^\infty_t$.  This is coherent with the fact that trajectories in classical solutions are real analytic (see e.g. \cite{CVW} and references therein).
\end{remark}
\begin{proof}[Proof of Theorem \ref{timeregthm}]
Isett proved that if $u\in L^\infty_tC^\alpha_{x}$ is a weak solution of the Euler equations, then every particle trajectory of $u$ is of class $C_t^{{1}/{(1-\alpha)}}$, see \cite{Isettreg2,Isettreg}. Here we will content ourselves by proving the result assuming that $\alpha\in (0, 1/2]$, following the discussion of Eyink \cite{Eyink}. It is not too difficult to obtain the general case from here.

We shall consider trajectories $X_t^\ell(a)$ of the mollified velocity field $\ol{u}_\ell$ \eqref{molflow}. 
Let $v^\ell(t,a) :=\frac{\rmd}{\rmd t} X_t^\ell(a)$ and $\delta_\tau f(t,a) :=f(\tau+ t, a)- f(t,a)$. 
The natural time-scale of the Lagrangian velocity mollified at length-scale $\ell$ is   the local eddy turnover time $\ell/\delta_\ell u$, defined by
 $
 \tau_\ell :={\ell^{1-\alpha}}/{\|u\|_{L^\infty_t C^\alpha_x}}.
$
Indeed, we have
  \begin{lemma}\label{lemmapart1} Let  $X_t$ be any solution of \eqref{ode}.  Then
 \be
 \|X_t(a)- X_t^\ell(a) \| \lesssim \ell \exp (t/\tau_\ell),
 \ee
 where $ \tau_\ell :={\ell^{1-\alpha}}/{\|u\|_{L^\infty_t C^\alpha_x}}$.
 \end{lemma}
 \begin{proof}[Proof of Lemma \ref{lemmapart1}]
Let us introduce $\delta X^\ell_t(a):=X_t(a) - X_t^\ell(a)$.  Then $ \delta X^\ell_0(a) = 0$ while
 \begin{align*}
\frac{\rmd}{\rmd t}  \delta X^\ell_t(a)&={u} ( X_t(a) ,t)- \ol{u}_\ell ( X_t^\ell(a) ,t) \\
&={u} ( X_t(a) ,t)- \ol{u}_\ell ( X_t(a) ,t)  +\ol{u}_\ell ( X_t(a) ,t) - \ol{u}_\ell ( X_t^\ell(a) ,t).
\end{align*}
Using $ \|\nabla \ol{u}_\ell(\cdot, t)\|\lesssim  \sup_{\ell'\leq \ell}\|\delta_{\ell'} u\|_{L^\infty_{t,x}}/\ell \leq {\tau_\ell}^{-1}$ and $|u-\ol{u}_\ell|\leq \sup_{\ell'\leq \ell}\|\delta_{\ell'} u\|_{L^\infty_{t,x}}$,  we have
 \begin{align*} 
\frac{\rmd}{\rmd t} \| \delta X^\ell_t(a)\| &\leq  \sup_{\ell'\leq \ell}\|\delta_{\ell'} u\|_{L^\infty_{t,x}} + \|\nabla \ol{u}_\ell(\cdot, t)\|_{L^\infty} \| \delta X^\ell_t(a)\|  \lesssim   \frac{1}{\tau_\ell}\left(\ell +    \| \delta X^\ell_t(a)\| \right).
\end{align*}
The claimed result follows by Gronwall's inequality.
 \end{proof}
   \begin{lemma} \label{lemmapart2} Let  $X_t$ be any solution of \eqref{ode} and $v= \dot{X}_t$.  Then
   \be\label{holderv}
\sup_{0\leq \tau\leq \tau_\ell} \|v(t+\tau, \cdot) - v(t,\cdot)  \|_{L^\infty_{a}}  \lesssim  \big( \|u\|_{L^\infty_t C^\alpha_x}\big)^{\frac{1}{1-\alpha}}   \tau_\ell^{\frac{\alpha}{1-\alpha}}.
 \ee
 \end{lemma}
 \begin{proof}[Proof of Lemma \ref{lemmapart2}]
 Recalling that $ \ell = \big( \|u\|_{L^\infty_t C^\alpha_x} \tau_\ell\big)^{\frac{1}{1-\alpha}}$, the result follows by triangle inequality once we show the following two bounds:
  \begin{align}\label{veqn1}
\|v^\ell(t, \cdot) - v(t,\cdot)  \|_{L^\infty}  &\lesssim \frac{\ell}{\tau_\ell}   \qquad \text{for times} \qquad 0\leq t\leq \tau_\ell,\\ \label{veqn2}
\|v^\ell(t+\tau, \cdot) - v^\ell(t,\cdot)  \|_{L^\infty}  &\lesssim  \frac{\ell}{\tau_\ell}   \qquad \text{for times} \qquad 0\leq \tau \leq \tau_\ell.
 \end{align}
Fix $0\leq t\leq \tau_\ell$ so that $\|\delta X^\ell_t\|_{L^\infty} \leq \ell$ by Lemma \ref{lemmapart1}.  
 Thus, on these timescales, we have
 \begin{align*} 
\|v^\ell(t, \cdot) - v(t,\cdot)  \|_{L^\infty} &= \| \ol{u}_\ell ( X_t^\ell(\cdot) ,t)  -  {u} ( X_t(\cdot) ,t)  \|_{L^\infty} \\
&= \| \ol{u}_\ell (\cdot ,t)  -  {u} ( \cdot+  \delta X^\ell_t((X_t^\ell)^{-1}(\cdot)) ,t)  \|_{L^\infty} \lesssim \sup_{\ell'\leq \ell}\|\delta_{\ell'} u\|_{L^\infty_{t,x}}.
\end{align*}
This establishes \eqref{veqn1}.  As for \eqref{veqn2}, from the Euler equation $\delta_\tau v^\ell(t,a):=v^\ell(t+\tau, a) - v^\ell(t,a) $ satisfies
$
 \delta_\tau v^\ell(t,a) =\int_0^\tau a_\ell( s, X_s(a)) \rmd s,
$
 where $a_\ell = \nabla \ol{p}_\ell + \nabla \cdot \tau_\ell(u,u)$.   It follows from standard estimates (using that the pressure enjoys double regularity \cite{CDRF}) that 
$
 \| a_\ell( t,\cdot)\|_{L^\infty} \lesssim  \tfrac{1}{\ell}\sup_{\ell'\leq \ell} \|\delta_{\ell'} u \|_{L^\infty_{t,x}}^2 = \tfrac{\ell}{\tau_\ell^2}.
$
The conclusion follows  since $\tau \leq \tau_\ell$.  
  \end{proof}
  
Since $\tau_\ell \sim \ell^{1-\alpha} \to 0$ as $\ell \to 0$, the \eqref{holderv} shows H\"older continuity at an arbitrary $t$.
\end{proof}

In view of the above result on trajectory regularity, we ask
\begin{question}\label{consques1}
Suppose that \emph{all} trajectories $X_t$ of a weak solution $u$ are $C^{1, 1/2+}_t$.  Does $u$ conserve energy?
\end{question}
Note that trajectory regularity does not imply any kind of general regularity of the vector field $u$.  For instance, if $u$ is a bounded shear flow, the profile can be very irregular but the trajectories are straight lines.

It is unfortunate that many of the rigorous results discussed in this section require H\"older continuity as opposed to Besov.  For instance, given the observations do not support $\sfrac{1}{3}$ H\"older continuity,  we wonder if there is an analogous Besov statement for Lagrangian regularity.  Namely, if $u\in L_t^p (0,T;B_{p,\infty}^\sigma(\mathbb{T}^d))$, is it true that $v =L^{p/2}(\mathbb{T}^d; B^{\frac{\sigma}{1-\sigma}}_{p/2,\infty}(0,T))$?   The difficulty in establishing this type of result is that, unlike the H\"{o}lder case, it seems necessary to bundle trajectories and the regularity in the label $a$ can be very bad. It may be that is impossible to bundle trajectories in such as way as to define a measure-preserving flowmap $X_t(a)$ that depends measurably on $a$ (see \cite{P25} for an example).  Could such regularity at least be preserved uniformly along solutions of Navier-Stokes?

\section{Conclusion}

We have not come much closer to developing the "first principled" theory of turbulence that Kolmogorov sought.  However, as a result of decades of profound work and deep understanding from many great minds, we  have presented a framework that allows us to paint  a precise conjectural picture of certain features of turbulence.  Simply by making precise statements, we can perceive and prove some fundamental constraints on the phenomenon. It is useful to have in mind \emph{some} mathematically well defined features and their consequences, even if we don't know what characterizes turbulence in totality, or even whether those features are universal to all phenomena one might call turbulent. 

Though it seems far off, a more complete understanding of turbulence might come from a dynamical systems approach as suggested by Kolmogorov himself \cite{A91}, and/or by embracing elements of probability theory in a more fundamental way.  We end with Richard Feynman's plan in this direction, which sounds quite modern.

\begin{figure}[h!]
  \begin{center}
    \includegraphics[width=1\textwidth]{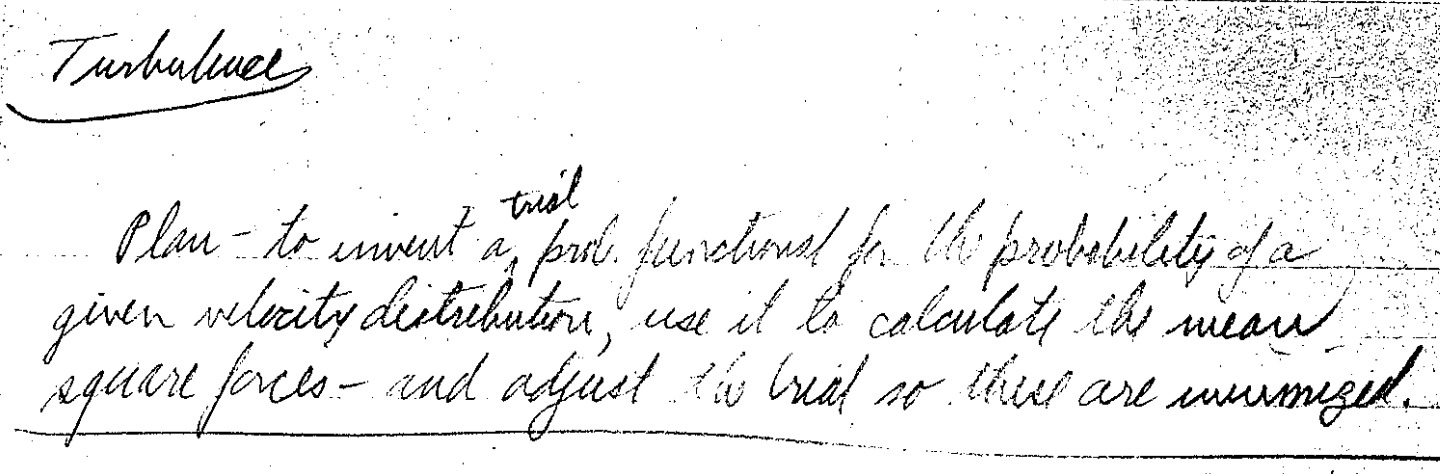}
  \end{center}
  \caption{ R. P. Feynman, unpublished (Folder 76.14 Caltech archives).}\label{feynman}
\end{figure}

\vspace{4mm}

\noindent \textbf{Acknowledgements.}  These notes are based on a series of lectures delivered during the ``Turbulence on the Banks of the Arno" school in Scuola Normale Superiore in Pisa, Italy, in  April 2025,
as well as the Summer School on Advances in Hyperbolic Balance Laws, Hirschegg, Austria, in September 2025. I am grateful to the organizers for the invitation and hospitality, and to the audience and participants for engaging discussions and questions.  I would like to thank Tarek Elgindi and Vlad Vicol for their remarks and comments. I am particularly grateful to Luigi De Rosa and Gregory Eyink for their detailed comments but mainly for the many discussions on turbulence over the years. 
 My work was supported by the NSF CAREER award \#2235395, a Stony Brook University Trustee’s award as well as an Alfred P. Sloan Fellowship.


\begin{thebibliography}{99}

\bibitem{A10}
Alinhac, S. Blowup for nonlinear hyperbolic equations. Vol. 17. Springer Science, 2013.

\bibitem{A91}
Arnold, V. I. Kolmogorov’s hydrodynamic attractors. Vladimir I. Arnold-Collected Works. Springer,
Berlin, Heidelberg, 1991. 429-432.

\bibitem{A99}
Arnold, V. I., and Khesin, B. A.: Topological methods in hydrodynamics. Vol. 125. Springer Science \& Business Media, 1999.

\bibitem{AV25}
Armstrong S. and Vicol V. Anomalous diffusion by fractal homogenization, Annals of PDE 11 (2025), no. 1, 2.

\bibitem{AFLV}
Aurell, E., Frisch, U., Lutsko, J., and Vergassola, M. (1992). On the multifractal properties of the energy dissipation derived from turbulence data. Journal of Fluid Mechanics, 238, 467-486.

\bibitem{BGME22}
Bandak, D., Goldenfeld, N., Mailybaev, A. and Eyink. G. "Dissipation-range fluid turbulence and thermal noise." Physical Review E 105, no. 6 (2022): 065113.

\bibitem{BTW19}
Bardos, C., Titi, E. S., and Wiedemann, E. (2019). Onsager’s conjecture with physical boundaries and an application to the vanishing viscosity limit. Communications in Mathematical Physics, 370(1), 291-310.


\bibitem{B}
Batchelor, G. K. Small-scale variation of convected quantities like temperature in turbulent
fluid Part 1. General discussion and the case of small conductivity. J. Fluid Mech. 5 (1959),
no. 1, 113–133. 

\bibitem{BT}
 Batchelor, G. K. and Townsend, A. A.  The nature of turbulent motion at large wave-numbers, Proc. R. Soc. Lond. A 199
(1949), 238–255

\bibitem{BK}
Bec, J., and Khanin, K. Burgers turbulence. Physics reports 447.1-2 (2007): 1-66.

\bibitem{BCZPSW}
Bedrossian, J., Coti Zelati, M., Punshon-Smith, S., and Weber, F. (2019). A sufficient condition for the Kolmogorov 4/5 law for stationary martingale solutions to the 3D Navier–Stokes equations. Communications in Mathematical Physics, 367(3), 1045-1075.

\bibitem{BCZPSW2}
Bedrossian, J., Coti Zelati, M., Punshon-Smith, S., and Weber, F. (2020). Sufficient conditions for dual cascade flux laws in the stochastic 2d Navier–Stokes equations. Archive for Rational Mechanics and Analysis, 237(1), 103-145.

\bibitem{BBPS1}
Bedrossian, J., Blumenthal, A., and Punshon-Smith, S. (2021). Almost-sure enhanced dissipation and uniform-in-diffusivity exponential mixing for advection–diffusion by stochastic Navier–Stokes. Probability Theory and Related Fields, 179(3), 777-834.

\bibitem{BBPS2}
Bedrossian, J., Blumenthal, A., and Punshon-Smith, S. (2022). Almost-sure exponential mixing of passive scalars by the stochastic Navier–Stokes equations. The Annals of Probability, 50(1), 241-303.

\bibitem{BBPS3}
Bedrossian, J., Blumenthal, A., and Punshon‐Smith, S. (2022). The Batchelor spectrum of passive scalar turbulence in stochastic fluid mechanics at fixed Reynolds number. Communications on Pure and Applied Mathematics, 75(6), 1237-1291.


\bibitem{BDLW}
Bentkamp, L., Drivas, T. D., Lalescu, C. C., and Wilczek, M. (2022). The statistical geometry of material loops in turbulence. Nature Communications, 13(1), 2088.

\bibitem{BD14}
Benveniste, D., and Drivas, T. D. (2014). Asymptotic results for backwards two-particle dispersion in a turbulent flow. Physical Review E, 89(4), 041003.

\bibitem{bernard1996anomalous1}
Bernard, D., Gaw\c{e}dzki, K.,  and Kupiainen, A. Slow modes in passive advection. J. Statist.
Phys., 90(3-4):519–569, 1998.

\bibitem{bernard1996anomalous}
Bernard, D., Gaw\c{e}dzki, K.,  and Kupiainen, A. Anomalous scaling in the n-point functions of a passive scalar, Physical
Review E 54 (1996), no. 3, 2564.

\bibitem{BM17}
Bianchini, S., and Marconi, E. (2017). On the structure of $L^\infty$-entropy solutions to scalar conservation laws in one-space dimension. Archive for Rational Mechanics and Analysis, 226(1), 441-493.

\bibitem{BBetal}
Biferale, L., Boffetta, G., Celani, A., Devenish, B. J., Lanotte, A., and Toschi, F. (2005). Lagrangian statistics of particle pairs in homogeneous isotropic turbulence. Physics of Fluids, 17(11).

\bibitem{BHB13}
Bitane, R., Homann, H., and Bec, J. (2013). Geometry and violent events in turbulent pair dispersion. Journal of Turbulence, 14(2), 23-45.

\bibitem{Betal06}
Bourgoin, M., Ouellette, N. T., Xu, H., Berg, J., and Bodenschatz, E. (2006). The role of pair dispersion in turbulent flow. Science, 311(5762), 835-838.

\bibitem{BS02}
Boffetta, G., and Sokolov, I. M. (2002). Relative dispersion in fully developed turbulence: the Richardson’s law and intermittency corrections. Physical review letters, 88(9), 094501.

\bibitem{BDSV22}
Buckmaster, T., Drivas, T. D., Shkoller, S., and Vicol, V. (2022). Simultaneous development of shocks and cusps for 2D Euler with azimuthal symmetry from smooth data. Annals of PDE, 8(2).

\bibitem{BV20} 
Buckmaster, T., and Vicol, V. (2020). Convex integration and phenomenologies in turbulence. EMS Surveys in Mathematical Sciences, 6(1), 173-263.

\bibitem{BDSV} 
Buckmaster, T., De Lellis,  C., Sz\'ekelyhidi  Jr., L., and Vicol, V. Onsager’s conjecture for admissible weak
solutions, Comm. Pure Appl. Math. 72 (2019), no. 2, 229–274.

\bibitem{BSW23} 
Burczak, J., Sz\'ekelyhidi Jr., L.  and Wu, B. Anomalous dissipation and Euler flows (2023). Preprint available at
arXiv:2310.02934.

\bibitem{B48}
Burgers, J. M. (1948). A mathematical model illustrating the theory of turbulence. Adv. App. Mech. 1, 171-199.

\bibitem{B39}
Burgers, J. M. (1939). Mathematical examples illustrating relations occurring in the theory of turbulent fluid motion. In Selected Papers of JM Burgers (pp. 281-334). Dordrecht: Springer Netherlands.

\bibitem{CCMV04} 
Celani, A., Cencini, M., Mazzino, A. and Vergassola, M.  Active and passive fields face to face, New J. Phys. 6, 72 (2004)

\bibitem{CD95} 
Chae, D., Dubovskii, P. Functional and measure-valued solutions of the euler equations for flows of incompressible fluids. Arch. Rational Mech. Anal. 129, 385–396 (1995).


\bibitem{C78} 
Champagne, F. 1978. The fine-scale structure of the turbulent velocity field. 1. Fluid
Mech. 86: 67-108 

\bibitem{CGV03} 
Chaves, M., Gaw\c{e}dzki, K., Horvai, P., Kupiainen, A. and Vergassola, M. 2003 Lagrangian
dispersion in Gaussian self-similar velocity ensembles. J. Stat. Phys. 113 (5-6), 643–692.

\bibitem{CH22} 
Chen, J., and Hou, T. Stable nearly self-similar blowup of the 2D Boussinesq and 3D Euler equations with smooth data I: Analysis. arXiv preprint arXiv:2210.07191 (2022).

\bibitem{CG12}
Chen, G, and Glimm, J.: Kolmogorov’s Theory of Turbulence and Inviscid Limit of the Navier-Stokes Equations in $\mathbb{R}^3$
Comm. Math. Phys. 310.1 (2012): 267-283.

\bibitem{CDB22}
Cheminet, A., Geneste, D., Barlet, A., Ostovan, Y., Chaabo, T., Valori, V., ... and Dubrulle, B. (2022). Eulerian vs Lagrangian irreversibility in an experimental turbulent swirling flow. Physical Review Letters, 129(12), 124501.

\bibitem{CL} 
Chertkov, M., and Lebedev, V. "Decay of scalar turbulence revisited." Physical review letters 90.3 (2003): 034501.


\bibitem{CDRF}
Colombo, M., De Rosa, L., and Forcella, L. (2020). Regularity results for rough solutions of the incompressible Euler equations via interpolation methods. Nonlinearity, 33(9), 4818.

\bibitem{colombo2023anomalous}
Colombo, M., Crippa, G. and Sorella, M.  Anomalous dissipation and lack of selection in the Obukhov–Corrsin theory of
scalar turbulence, Annals of PDE 9 (2023), no. 2, 21.

\bibitem{C51}
Corrsin, S. On the spectrum of isotropic temperature fluctuations in an isotropic turbulence, Journal of Applied Physics
22 (1951), no. 4, 469–473.

\bibitem{CET94}
Constantin, P., E, W., and Titi, E. (1994). Onsager's conjecture on the energy conservation for solutions of Euler's equation.

\bibitem{CF88}
Constantin, P., and Foias, C.: Navier-stokes equations. University of Chicago Press, 1988.

\bibitem{CV18}
Constantin, P. and  Vicol, V. (2018). Remarks on high Reynolds numbers hydrodynamics and the inviscid limit. Journal of Nonlinear Science, 28(2), 711-724.

\bibitem{CP1}
Constantin P. and Procaccia, I.  Scaling in fluid turbulence: a geometric theory, Physical Review E 47 (1993), no. 5, 3307.

\bibitem{CP2} Constantin  P. and Procaccia, I.   The geometry of turbulent advection: sharp estimates for the dimensions of level sets, Nonlinearity 7 (1994), no. 3,
1045.

\bibitem{CVW}
Constantin, P., Vicol, V., and Wu, J. (2015). Analyticity of Lagrangian trajectories for well posed inviscid incompressible fluid models. Advances in Mathematics, 285, 352-393.

\bibitem{CDG}
Coti Zelati, M., Drivas, T. D., and Gvalani, R. S. (2024). Mixing by statistically self-similar Gaussian random fields. Journal of Statistical Physics, 191(5), 61.

\bibitem{Daf83}
Dafermos, C. M. (1983). Hyperbolic systems of conservation laws. In Systems of nonlinear partial differential equations (pp. 25-70). Dordrecht: Springer Netherlands.

\bibitem{Daf}
Dafermos, C. M. (2006). Continuous solutions for balance laws. Ricerche di Matematica, 55(1), 79-92.


\bibitem{DL1}
De Lellis, C.  and Sz\'ekelyhidi  Jr., L. The Euler equations as a differential inclusion, Ann. of Math. 170 (2009), no. 3, 1417–1436.

\bibitem{DL2} De Lellis, C.  and Sz\'ekelyhidi  Jr., L. Dissipative continuous Euler flows, Invent. Math. 193 (2013), no. 2, 377–407.

\bibitem{Dubrulle18}
Debue, P., Shukla, V., Kuzzay, D., Faranda, D., Saw, E. W., Daviaud, F., and Dubrulle, B. (2018). Dissipation, intermittency, and singularities in incompressible turbulent flows. Physical Review E, 97(5), 053101.

\bibitem{Dubrulle19}
Dubrulle, B. Beyond Kolmogorov cascades. Journal of Fluid Mechanics 867 (2019): P1.

\bibitem{DI24}
De Rosa, L., and Inversi, M., Dissipation in Onsager’s Critical Classes and Energy Conservation in $BV\cap L^\infty$
with and Without Boundary, Comm. Math. Phys. 405 (2024), no. 1, 6.

\bibitem{DIN24}
De Rosa, L., Inversi, M., and Nesi, M. (2024). Dissipation for codimension 1 singular structures to incompressible Euler. arXiv preprint arXiv:2412.08493.

\bibitem{DI24}
De Rosa, L., and Isett, P. (2024). Intermittency and lower dimensional dissipation in incompressible fluids. Archive for Rational Mechanics and Analysis, 248(1), 11.

\bibitem{DT22}
De Rosa, L., and Tione, R. (2022). Sharp energy regularity and typicality results for Hölder solutions of incompressible Euler equations. Analysis \& PDE, 15(2), 405-428.

\bibitem{DDII24}
De Rosa, L., Drivas, T. D., and Inversi, M. (2024). On the support of anomalous dissipation measures. Journal of Mathematical Fluid Mechanics, 26(4), 56.

\bibitem{DDII25}
De Rosa, L., Drivas, T. D., Inversi, M., and Isett, P. (2025). Intermittency and Dissipation Regularity in Turbulence. arXiv preprint arXiv:2502.10032.

\bibitem{DM87}
DiPerna, R. J., and  Majda, A.J. Oscillations and concentrations in weak solutions of the incompressible fluid equations. Communications in mathematical physics 108.4 (1987): 667-689.

\bibitem{DSY}
Donzis, D. A., Sreenivasan, K. R., and Yeung, P. K. Scalar dissipation rate and dissipative anomaly in isotropic turbulence,
Journal of Fluid Mechanics 532 (2005), 199–216.

\bibitem{DR00} 
Duchon, J., and Robert, R. (2000). Inertial energy dissipation for weak solutions of incompressible Euler and Navier-Stokes equations. Nonlinearity, 13(1), 249.

\bibitem{D19} 
Drivas, T. D. (2019). Turbulent cascade direction and Lagrangian time-asymmetry. Journal of Nonlinear Science, 29(1), 65-88.

\bibitem{D22} Drivas, T.D. Self-regularization in turbulence from the Kolmogorov 4/5-law and alignment. Philosophical Transactions of the Royal Society A 380.2226 (2022): 20210033.

\bibitem{DE23}
Drivas, T. D., and Elgindi, T. M. (2023). Singularity formation in the incompressible Euler equation in finite and infinite time. EMS Surveys in Mathematical Sciences, 10(1), 1-100.

\bibitem{DEJI22}
Drivas, T. D., Elgindi, T. M., Iyer, G., and Jeong, I. J. (2022). Anomalous dissipation in passive scalar transport. Archive for Rational Mechanics and Analysis, 243(3), 1151-1180.

\bibitem{DE18} 
Drivas, T. D., and Eyink, G. L. (2018). An Onsager singularity theorem for turbulent solutions of compressible Euler equations. Communications in Mathematical Physics, 359(2), 733-763.

\bibitem{DE19} 
Drivas, T. D. and Eyink, G. L. An Onsager singularity theorem for Leray solutions of incompressible Navier–Stokes. Nonlinearity 32.11 (2019): 4465.

\bibitem{DE17a} 
Drivas, T. D., and Eyink, G. L. (2017). A Lagrangian fluctuation–dissipation relation for scalar turbulence. Part I. Flows with no bounding walls. Journal of Fluid Mechanics, 829, 153-189.

\bibitem{DE17b} 
Drivas, T. D., and Eyink, G. L. (2017). A Lagrangian fluctuation–dissipation relation for scalar turbulence. Part II. Wall-bounded flows. Journal of Fluid Mechanics, 829, 236-279.

\bibitem{DGP25} 
Drivas, T. D., Galeati, L., and Pappalettera, U. (2025). Anomalous dissipation and regularization in isotropic Gaussian turbulence. arXiv preprint arXiv:2509.10211.

\bibitem{DH20} 
Drivas, T. D., and Holm, D. D. (2020). Circulation and energy theorem preserving stochastic fluids. Proceedings of the Royal Society of Edinburgh Section A: Mathematics, 150(6), 2776-2814.

\bibitem{DMR24} 
Drivas, T. D., Mailybaev, A. A., and  Raibekas, A. (2024). Statistical determinism in non-Lipschitz dynamical systems. Ergodic Theory and Dynamical Systems, 44(7), 1856-1884.

\bibitem{DM21} 
Drivas, T. D., and Mailybaev, A. A. (2021). ‘Life after death’ in ordinary differential equations with a non-Lipschitz singularity. Nonlinearity, 34(4), 2296.

\bibitem{DN19} 
Drivas, T. D. and Nguyen, H. Q. (2019). Remarks on the emergence of weak Euler solutions in the vanishing viscosity limit. Journal of Nonlinear Science, 29(2), 709-721.

\bibitem{DN18} 
Drivas, T. D., and Nguyen, H. Q. (2018). Onsager's conjecture and anomalous dissipation on domains with boundary. SIAM Journal on Mathematical Analysis, 50(5), 4785-4811.

\bibitem{DJLW17} 
Drivas, T. D., Johnson, P. L., Lalescu, C. C., and Wilczek, M. (2017). Large-scale sweeping of small-scale eddies in turbulence: A filtering approach. Physical Review Fluids, 2(10), 104603.

 \bibitem{EM70} 
Ebin, D. G., and Marsden. J. "Groups of diffeomorphisms and the motion of an incompressible fluid." Annals of Mathematics 92.1 (1970): 102-163.

\bibitem{E21} 
Elgindi, T. M. (2021). Finite-time Singularity Formation for $C^{1,\alpha}$ Solutions to the Incompressible Euler Equations on $\mathbb{R}^3$. Annals of Mathematics, 194(3), 647-727.

\bibitem{EGM19} 
Elgindi, T. M., Ghoul, T. E., and Masmoudi, N. (2019). On the stability of self-similar blow-up for $ C^{1,\alpha} $ solutions to the incompressible Euler equations on $\mathbb {R}^ 3$. arXiv preprint arXiv:1910.14071.

\bibitem{elgindi2024norm} 
Elgindi T. M. and Liss, K.  Norm growth, non-uniqueness, and anomalous dissipation in passive scalars, Archive for
Rational Mechanics and Analysis 248 (2024), no. 6, 120.

\bibitem{Eyink}
Eyink, G. L.  Turbulence theory. Course notes. The Johns Hopkins University, 2007-08.

\bibitem{Eyink93}
Eyink, G. L. (1993). Lagrangian field theory, multifractals, and universal scaling in turbulence. Physics Letters A, 172(5), 355-360.

\bibitem{Eyink95}
Eyink, G. L.  Besov spaces and the multifractal hypothesis. Journal of statistical physics 78.1 (1995).

\bibitem{Eyink96}
Eyink, G. L. Intermittency and anomalous scaling of passive scalars in any space dimension, Physical Review E 54 (1996), no. 2, 1497.

\bibitem{E94}
Eyink, G. L. (1994). Energy dissipation without viscosity in ideal hydrodynamics I. Fourier analysis and local energy transfer. Physica D: Nonlinear Phenomena, 78(3-4), 222-240.

\bibitem{E03}
Eyink, G. L.  Local 4/5 -law and energy dissipation anomaly in turbulence, Nonlinearity, 16 137-145(2003).

\bibitem{E06}
Eyink, G. L. (2006). Cascade of circulations in fluid turbulence. Physical Review E—Statistical, Nonlinear, and Soft Matter Physics, 74(6), 066302.

\bibitem{Eyinkw}
Eyink, G. L.  (2024). Onsager's ‘ideal turbulence’theory. Journal of Fluid Mechanics, 988, P1.

\bibitem{E96} Eyink, G. L.  Exact results on stationary turbulence in 2d: consequences of vorticity conservation. Physica D, {\bf{91}} (1-2):97--142, (1996)

\bibitem{EB20} 
Eyink, G. L.  and  Bandak, D. “A Renormalization Group Approach to Spontaneous Stochasticity”, Physical Review Research, 2, 043161 (2020)

\bibitem{ED15b} 
Eyink, G. L., and Drivas, T. D. (2015). Quantum spontaneous stochasticity. arXiv preprint arXiv:1509.04941.

\bibitem{ED15} 
Eyink, G. L., and Drivas, T. D. (2015). Spontaneous stochasticity and anomalous dissipation for Burgers equation. Journal of Statistical Physics, 158(2), 386-432.

\bibitem{EJ22} 
Eyink, G. L. and Jafari, A. High schmidt-number turbulent advection and giant concentration fluctuations, Physical
Review Research 4 (2022), no. 2, 023246.

\bibitem{EQ25} 
Eyink, G. L.,and Quan, H. (2025). Weak-strong uniqueness and extreme wall events at high Reynolds number. Physical Review Fluids, 10(6), 064610.

\bibitem{EKQ22} 
Eyink, G. L., Kumar, S., and Quan, H. (2022). The Onsager theory of wall-bounded turbulence and Taylor’s momentum anomaly. Philosophical Transactions of the Royal Society A, 380(2218).

\bibitem{EP25} 
Eyink, G. L., and Peng, L. (2025). Space-time statistical solutions of the incompressible Euler equations and Landau-Lifshitz fluctuating hydrodynamics. Nonlinearity, 38(8), 085011.

\bibitem{ES06} 
Eyink, G. L., and Sreenivasan, K. R. (2006). Onsager and the theory of hydrodynamic turbulence. Reviews of modern physics, 78(1), 87-135.

\bibitem{Faber09}  Faber T.  and Vassilicos, J. C. Turbulent pair separation due to multiscale stagnation point structure and its time asymmetry in two-dimensional turbulence. Phys. Fluids, 21(1):015106, (2009)

\bibitem{FGV} 
Falkovich,  G., Gaw\c{e}dzki, K.  and Vergassola. M. Particles and fields in fluid turbulence. Rev. Mod. Phys., 73: 913--975, (2001)

\bibitem{Falk13} Falkovich, G.  and Frishman, A. Single flow snapshot reveals the future and the past of pairs of particles in turbulence. Phys. Rev. Lett., 110(21):214502, (2013)

\bibitem{F95} 
Frisch, U. Turbulence: the legacy of AN Kolmogorov. Cambridge university press, 1995.

\bibitem{FB02} 
Frisch, U., and Bec, J. "Burgulence." New trends in turbulence Turbulence: nouveaux aspects: 31 July–1 September 2000. Berlin, Heidelberg: Springer Berlin Heidelberg, 2002. 341-383.

\bibitem{FSN} 
Frisch, U.,  Sulem, P.-L., and Nelkin, M. "A simple dynamical model of intermittent fully developed turbulence." Journal of Fluid Mechanics 87.4 (1978): 719-736.

\bibitem{FP} 
Frisch U.  and Parisi, G. On the singularity structure of fully developed turbulence, Turbulence and Predictability of Geophysical Flows and Climate Dynamics, (North-Holland, Amsterdam) (1985).

\bibitem{Frish14}
Frishman, A.  and Falkovich, G. New type of anomaly in turbulence. Phys. Rev. Letts. 113.2 : 024501. (2014)

\bibitem{GGM24} 
 Galeati, L., Grotto, F., and Maurelli, M. (2024). Anomalous Regularization in Kraichnan's Passive Scalar Model. arXiv preprint arXiv:2407.16668. 
 
\bibitem{gawedzki1998intermittency} 
 Gaw\c{e}dzki, K. Intermittency of passive advection, Advances in turbulence vii: Proceedings of the seventh European turbulence conference, held in Saint-Jean cap ferrat, France,1998, pp. 493–502.

\bibitem{GKN} 
Giri, V., Kwon, H., and Novack, M. (2024). A Wavelet-Inspired L 3-Based Convex Integration Framework for the Euler Equations. Annals of PDE, 10(2), 19.

\bibitem{GP} 
Golse F, Perthame B. 2013 Optimal regularizing effect for scalar conservation laws. Revista
Matemática Iberoamericana 29, 1477–1504. (doi:10.4171/RMI/765)

\bibitem{GJO15} 
Goldman, M.,  Josien, M. and  Otto, F. New bounds for the inhomogenous Burgers and the Kuramoto-Sivashinsky equations. Communications in Partial Differential Equations 40.12 (2015)

\bibitem{G02} 
Gotoh, T., Fukayama, D., and Nakano, T. (2002). Velocity field statistics in homogeneous steady turbulence obtained using a high-resolution direct numerical simulation. Physics of Fluids, 14(3).

\bibitem{G27} 
Gunther, N.  On the motion of fluid in a moving container, Izvestia Akad. Nauk
USSR, Ser. Fiz.–Mat. 20 (1927), 1323–1348. JFM 53.0786.08.

\bibitem{HS} 
Hess-Childs, E., and. Rowan, K. A universal total anomalous dissipator, arXiv preprint arXiv:2501.18526 (2025).

\bibitem{HS2} 
Hess-Childs, E., and. Rowan, K. (2025). Turbulent and intermittent phenomena in a universal total anomalous dissipator. arXiv preprint arXiv:2508.00115.

\bibitem{H33} 
Hölder, E. Über die unbeschränkte Fortsetzbarkeit einer stetigen ebenen Bewegung in einer unbegrenzten inkompressiblen Flüssigkeit. Mathematische Zeitschrift, 37(1), 727-738 (1933).
 
 \bibitem{I25} 
 Inversi, M. (2025). Fine dissipative properties of Euler solutions with measure first derivatives. arXiv preprint arXiv:2510.10704.
 
  \bibitem{Isetthesis}
 Isett, Philip. H\"{o}lder continuous Euler flows in three dimensions with compact support in time. Vol. 196. Princeton University Press, 2017.
 
 
 \bibitem{Isettreg}
 Isett, P. Regularity in time along the coarse scale flow for the incompressible Euler equations. Transactions of the American Mathematical Society 376.10 (2023): 6927-6987.
 
  \bibitem{Isettreg2}
 Isett, P.  "Regularity of Trajectories and Smooth Observables in Rough Euler Flows." In The Abel Symposium, pp. 185-204. Cham: Springer Nature Switzerland, 2023.

 \bibitem{Isett}
Isett, P. (2018). A proof of Onsager's conjecture. Annals of Mathematics, 188(3), 871-963.

 \bibitem{IsettOh}
Isett, P.  and Oh, S-J. On nonperiodic Euler flows with Hölder regularity. Archive for Rational Mechanics and Analysis 221.2 (2016): 725-804.

 \bibitem{IDES25} 
 Iyer, K. P., Drivas, T. D., Eyink, G. L., and Sreenivasan, K. R. (2025). Turbulence without Walls: Whither the Zeroth Law of Turbulence?. Physical Review Letters, 135(13), 134001.
 
  \bibitem{ISY20} 
 Iyer, K. P., Sreenivasan, K. R., and Yeung, P. K. (2020). Scaling exponents saturate in three-dimensional isotropic turbulence. Physical Review Fluids, 5(5), 054605.
 
   \bibitem{ISS18} 
 Iyer,  K. P., Schumacher, J., Sreenivasan, K. R., and  Yeung, P. K. Steep cliffs and saturated exponents in three-dimensional scalar turbulence, Physical review letters 121 (2018), no. 26, 264501.


\bibitem{JKM} 
Jauslin, H. R., Kreiss, H. O., and Moser, J. (1998). On the forced Burgers equation with periodic boundary conditions. Differential Equations: La Pietra 1996, 65, 133-153.

\bibitem{Jucha} 
Jucha, J., Xu, H., Pumir, A., Bodenschatz, E.: Time-reversal-symmetry breaking in turbulence. Phys. Rev.
Lett. 113(5), 054501 (2014)

\bibitem{JHUDB} 
Kanov, K., Burns, R., Lalescu, C., and Eyink, G. (2015). The Johns Hopkins turbulence databases: An open simulation laboratory for turbulence research. Computing in Science \& Engineering, 17(5), 10-17.


\bibitem{K03} 
Kaneda, Y., Ishihara, T., Yokokawa, M., Itakura, K. I., and Uno, A. (2003). Energy dissipation rate and energy spectrum in high resolution direct numerical simulations of turbulence in a periodic box. Physics of Fluids, 15(2), L21-L24.

   \bibitem{Kelvin} 
Kelvin.  T. On vortex motion. Trans. Roy. Soc. Edinb. 25 (1869), 217–260.
 
\bibitem{K} 
 Kolmogorov A. N. 1941 Dissipation of energy in locally isotropic turbulence. Dokl. Akad. Nauk
SSSR. 32, 16–18.

\bibitem{K62} 
Kolmogorov, A. N. (1962). A refinement of previous hypotheses concerning the local structure of turbulence in a viscous incompressible fluid at high Reynolds number. J. Fluid Mech, 13(1), 82-85.

   \bibitem{KP97} 
Komorowski, T., and Papanicolaou, G. (1997). Motion in a Gaussian incompressible flow. The Annals of Applied Probability, 7(1), 229-264.

\bibitem{Kunita} 
Kunita, H. Stochastic differential equations and stochastic flows of diffeomorphisms. In P. L. Hennequin, editor, Ecole d’ ´ Et´e de Probabilit´es de Saint-Flour XII - 1982, pages 143–303, Berlin, Heidelberg,
1984. Springer Berlin Heidelberg

\bibitem{Kr67} 
Kraichnan, R. H. Inertial ranges in two-dimensional turbulence. Phys. Fluids, {\bf{10}} (7), (1967)

\bibitem{Kr68} 
Kraichnan, R. H. Small-Scale Structure of a Scalar Field Convected by Turbulence. The Physics of
Fluids, 11(5):945–953, 1968.

\bibitem{Kr94} 
Kraichnan, R. H. Anomalous scaling of a 
randomly advected passive scalar. Physical review letters 72, no. 7 (1994): 1016. 


\bibitem{Kr94} 
Kraichnan, R. H. (1974). On Kolmogorov's inertial-range theories. J.  Fluid Mech, 62(2), 305-330.

   \bibitem{Lcomp} 
Lalescu, C. C., Bramas, B., Rampp, M., and Wilczek, M. (2022). An efficient particle tracking algorithm for large-scale parallel pseudo-spectral simulations of turbulence. Computer Physics Communications, 278, 108406.

   \bibitem{LL} 
 Landau, L. D. and Lifshitz, E. M. (1987). Fluid Mechanics: Volume 6 (Vol. 6). Elsevier. 
 
 \bibitem{LMP21} 
Lanthaler, S., Mishra, S., and Parés-Pulido, C. (2021). On the conservation of energy in two-dimensional incompressible flows. Nonlinearity, 34(2), 1084.

\bibitem{L85} 
Le Jan, Y. (1985). On isotropic Brownian motions. Zeit. für Wahr. und verwandte Gebiete, 70(4).

\bibitem{LR1} 
Le Jan, Y.  and Raimond, O. Integration of Brownian vector fields. Ann. Probab., 30(2),
2002.

\bibitem{LR2} 
Le Jan, Y.  and Raimond, O. Flows, coalescence and noise. Ann. Probab., 32(2):1247–1315, 2004

\bibitem{Leith68} 
Leith. C. E., Diffusion approximation for two-dimensional turbulence. Phys. Fluids, {\bf{11}} (3), (1968)

\bibitem{Lopez25}
Lopez, A., Barral, A., Costa, G., Pikeroen, Q., Shukla, V., and Dubrulle, B. (2025). Efficiency of turbulence. Journal of Fluid Mechanics, 1019, A55.

\bibitem{L25} 
Lichtenstein, L.
Uber einige Existenzprobleme der Hydrodynamik homogener,
unzusammendr\"uckbarer, reibungsloser Fl\"ussigkeiten und die Helmholtzschen
Wirbels\"atze, Math. Z. 23 no. 1 (1925)

\bibitem{L96} 
Lions,  P.-L. Mathematical topics in fluid dynamics, Vol. 1, Incompressible models, Oxford Science Publication, Oxford, 1996.

\bibitem{M1}
Mailybaev, A. A. 2016 Spontaneously stochastic solutions in one-dimensional inviscid systems.
Nonlinearity 29 (8), 2238.

\bibitem{M2}
Mailybaev, A. A. and Raibekas, A. 2023a Spontaneous stochasticity and renormalization group in
discrete multi-scale dynamics. Commun. Math. Phys. 401, 2643–2671.

\bibitem{M3}
Mailybaev, A. A. and Raibekas, A. 2023b Spontaneously stochastic Arnold’s cat. Arnold Math. J.
9 (3), 339–357.

\bibitem{MHPF}
Matsumoto, T., Roy, D., Khanin, K., Pandit, R., and Frisch, U. (2025). Large-scale multifractality and lack of self-similar decay for Burgers and 3D Navier-Stokes turbulence.  arXiv:2503.08983.

\bibitem{MS91}
Meneveau C. and Sreenivasan, K. R. The multifractal nature of turbulent energy dissipation, Journal of Fluid Mechanics
224 (1991), 429–484.

\bibitem{MS87}
Meneveau C. and Sreenivasan, K. R. (1987). The multifractal spectrum of the dissipation field in turbulent flows. Nuclear Physics B-Proceedings Supplements, 2, 49-76.

\bibitem{MY}
Monin, A. S. and  Yaglom, A. M. Statistical fluid mechanics: mechanics of turbulence. Vol. II.
Dover Publications, Inc., Mineola, NY, english edition, 2007. Translated from the 1965 Russian original,
Edited and with a preface by John L. Lumley, Reprinted from the 1975 edition

\bibitem{N24} 
Novack, M. (2024). Scaling laws and exact results in turbulence. Nonlinearity, 37(9), 095002.

\bibitem{NV23} 
Novack, M., and Vicol, V. (2023). An intermittent Onsager theorem. Inventiones, 233(1), 223-323.

\bibitem{NFS11} 
Nguyen van Yen, R., Farge, M., and Schneider, K. (2011). Energy Dissipating Structures Produced by Walls in Two-Dimensional Flows at Vanishing Viscosity. Physical Review Letters, 106(18).

\bibitem{Ob49} 
Obukhov, A. M.  Structure of the temperature field in a turbulent flow, Izv. Akad. Nauk SSSR, Ser. Geogr. Geofiz 13 (1949),
no. 1, 58–69.

\bibitem{O49} 
Onsager, L. (1949). Statistical hydrodynamics. Il Nuovo Cimento (1943-1954), 6(Suppl 2), 279-287.

\bibitem{O75} 
Onsager, L., circa 1975, private notes. LOA, NTNU, 11.132.

\bibitem{OM} 
Ott, S. and Mann, J.  An experimental investigation of the relative diffusion of particle pairs in three-dimensional turbulent flow, J. Fluid Mech. 422 207-223(2000)

\bibitem{O} 
Otto, F. A generalization of the entropy identity for Burgers’ equation and
application to the Kuramoto-Sivashinsky equation, Oberwolfach Science Reports, 2019

 \bibitem{P25} 
 Pappalettera, U. (2025). On measure-preserving selection of solutions of ODEs. Proceedings of the American Mathematical Society, 153(05), 2037-2051.
 
 \bibitem{R26} 
Richardson,  L. F.  Atmospheric diffusion shown on a distance-neighbor graph,
Proc. Roy. Soc. Lond. A 110 709-737 (1926)

\bibitem{R1} 
Rowan, K. (2024). On anomalous diffusion in the Kraichnan model and correlated-in-time variants. Archive for Rational Mechanics and Analysis, 248(5), 93.

\bibitem{R2} 
Rowan, K. (2025). The Obukhov--Corrsin spectrum of passive scalar turbulence through anomalous regularization. arXiv preprint arXiv:2512.02853.

\bibitem{TT} 
Tadmor E, Tao T. 2007 Velocity averaging, kinetic formulations, and regularizing effects in
quasi–linear PDEs. Commun. Pure Appl. Math. 60, 1488–1521. (doi:10.1002/cpa.20180)

\bibitem{Salmon}
Salmon, R. (1988). Hamiltonian fluid mechanics. Annual review of fluid mechanics, 20(1), 225-256.

\bibitem{Sch}
Scheffer, V. An inviscid flow with compact support in space-time, J. Geom. Anal. 3 (1993), no. 4.

\bibitem{SAF92}
She, Z. S., Aurell, E., and Frisch, U. (1992). The inviscid Burgers equation with initial data of Brownian type. Communications in mathematical physics, 148(3), 623-641.

\bibitem{SL}
She  Z. S., and L\'eveque, E. “Universal scaling laws in fully developed turbulence,”
Phys. Rev. Lett. 72 336-339(1994)

\bibitem{Saw05} 
Sawford,  B. L., Yeung,  P. K., and Borgas, M. S.Comparison of backwards and forwards relative dispersion in turbulence. Phys. Fluids, 17:095109, (2005)

\bibitem{Shn}
Shnirelman, A., Weak solutions with decreasing energy of incompressible Euler equations, Comm. Math. Phys.
210 (2000), no. 3, 541–603.

\bibitem{Shiryaev}
Shiryaev, A.. "Kolmogorov and the turbulence." (2006).

\bibitem{Sh9}
Shvydkoy, R. (2009). On the energy of inviscid singular flows. Journal of mathematical analysis and applications, 349(2), 583-595.

\bibitem{SM88}
Sreenivasan, K. R.  and  Meneveau, C. Singularities of the equations of fluid motion, Physical Review A 38 (1988), no. 12.

\bibitem{SM86}
Sreenivasan, K. R., and Meneveau, C.  (1986). The fractal facets of turbulence. Journal of Fluid Mechanics, 173, 357-386.

\bibitem{S19} 
Sreenivasan, K. R., Turbulent mixing: A perspective, Proceedings of the National Academy of Sciences 116 (2019), no. 37,
18175–18183.

\bibitem{S84} 
Sreenivasan, K. R.,  On the scaling of the turbulence energy dissipation rate, Phys.
Fluids 27: 1048-1051 (1984)

\bibitem{S96} 
Sreenivasan, K. R., Vainshtein, S. I., Bhiladvala, R., San Gil, I., Chen, S., and Cao, N. (1996). Asymmetry of velocity increments in fully developed turbulence and the scaling of low-order moments. Physical review letters, 77(8), 1488.

\bibitem{T21} 
Taylor, G. I. Diffusion by continuous movements, Proc. Lond. Math. Soc. Series
2, 20 196 (1921)

\bibitem{TL72} 
Tennekes, H., and Lumley, J. A first course in turbulence. MIT press, 1972.

\bibitem{TBM} 
Thalabard, S., Bec, J. and Mailybaev, A. A. 2020 From the butterfly effect to spontaneous
stochasticity in singular shear flows. Commun. Phys. 3 (1), 122.

\bibitem{W18} 
Waidmann, M., Klein, R., Farge, M., and Schneider, K. (2018). Energy dissipation caused by boundary layer instability at vanishing viscosity. Journal of Fluid Mechanics, 849, 676-717.

\bibitem{VD82} 
Van Dyke, M.. An album of fluid motion. NASA STI/Recon Technical Report A 82 (1982): 36549.

\bibitem{V19} 
Vigneron, F. On the Wiener-Khinchin transform of functions that behave as approximate power-laws. Applications to fluid turbulence. arXiv preprint arXiv:1909.06078 (2019).

\bibitem{VL} 
Vulpiani, A., and Livi, R. The Kolmogorov Legacy in Physics: A Century of Turbulence and Complexity. Lecture Notes in Physics, (642).

\bibitem{Xu16}  
Xu, H.,   Pumir, A. and Bodenschatz, E. Lagrangian view of time irreversibility of fluid turbulence. Science China Physics, Mechanics \& Astronomy, 59 (1):1--9, (2016)

\bibitem{Y01} 
Yakhot, V. (2001). Mean-field approximation and a small parameter in turbulence theory. Physical Review E, 63(2), 026307.

\bibitem{Y06} 
Yudovich, V. I. Global solvability versus collapse in the dynamics of an incompressible fluid. Mathematical events of the twentieth century. Springer, Berlin, Heidelberg, 2006. 501-528.

\bibitem{Z70} 
 Zeldovich, Ya.B., “Gravitational instability:
an approximate theory for large density perturbations,”
Astron. Astrophys. 5, 84–89 (1970).

\end{thebibliography}
\end{document}